\documentclass[reqno]{amsart}
\usepackage[
backend=biber,
style=numeric,
maxcitenames=5,
maxbibnames=9]{biblatex}

\renewbibmacro{in:}{}
\usepackage{a4,amsmath,amssymb,amsthm,slashed,thmtools}
\usepackage[top=1in, bottom=1in, left=1.25in, right=1.25in]{geometry}
\usepackage{xfrac}
\usepackage{latexsym}
\usepackage{mathrsfs}
\usepackage{enumerate}   

\usepackage{leftidx}

\usepackage{scrextend}
\usepackage{pgfplots}
\usepackage{tikz}
\usetikzlibrary{patterns,patterns.meta}
\usepackage{tkz-euclide}
\usetikzlibrary{positioning}
\usetikzlibrary{calc,fadings,decorations.pathreplacing}
\usepackage{calc}
\usepackage{enumitem}

\usepackage{amsaddr}

\usepackage[T1]{fontenc}

\usepackage{leftidx}
\usepackage{tikz}

\usepackage{hyperref}

\usepackage{color}

\usepackage{url}

\numberwithin{equation}{section}

\usepackage[mathscr]{eucal}

\usepackage{graphicx}
\graphicspath{pic/}

\usepackage{hyperref}
 \hypersetup{
     colorlinks=true,
     linkcolor=red,
     filecolor=blue,
     citecolor = red,      
     urlcolor=cyan,
     }

\newcommand{\teukhat}{\widehat{\teuk}}
\newcommand{\intS}{\int_{S(u,\ubar)}}
\newcommand{\modnab}{|\partial_\theta\Psi|^2+|\mathcal{U}\Psi|^2}

\newcommand{\quatre}{\mathbf{IV}}

\newcommand{\Real}{\mathfrak{R}}

\newcommand{\mch}{\mathcal{H}}
\newcommand{\mcm}{\mathcal{M}}
\newcommand{\mcq}{\mathcal{Q}}
\newcommand{\mcu}{\mathcal{U}}
\newcommand{\app}{‘}

\newcommand{\mcd}{\mathcal{D}}
\newcommand{\mcr}{\mathcal{R}}

\newcommand{\wbar}{\underline{w}}

\newcommand{\ubar}{{\underline{u}}}

\newcommand{\errp}{\mathrm{Err}[\Psip]}
\newcommand{\errm}{\mathrm{Err}[\Psim]}
\newcommand{\carterp}{\mcq_{+2}}
\newcommand{\carterm}{\mcq_{-2}}

\newcommand{\nablabar}{\overline{\nabla}}

\newcommand{\R}{\mathbb R}

\newcommand{\dee}{\mathrm{d}}

\newcommand{\ch}{\mathcal{CH_+}}

\newcommand{\ener}{\mathbf{e}}

\newcommand{\drond}{\mathring{\eth}}
\newcommand{\teuk}{\mathbf{T}}
\newcommand{\Psihat}{\widehat{\psi}}
\newcommand{\psihat}{\widehat{\psi}}

\newcommand{\ethat}{\widehat{e_3}}
\newcommand{\eqhat}{\widehat{e_4}}
\newcommand{\un}{\mathbf{I}}
\newcommand{\deux}{\mathbf{II}}
\newcommand{\trois}{\mathbf{III}}

\newcommand{\rb}{r_\mathfrak{b}}

\newcommand{\Psip}{\psi_{+2}}
\newcommand{\psip}{\psi_{+2}}

\newcommand{\Psim}{\psi_{-2}}
\newcommand{\psim}{\psi_{-2}}

\newcommand{\Psihatp}{\Psihat_{+2}}
\newcommand{\psihatp}{\psihat_{+2}}

\newcommand{\psihatm}{\Psihat_{-2}}

\newcommand{\sfrak}{2}
\newcommand{\fraks}{2}
\newcommand{\drin}{\partial_{r_{in}}}
\newcommand{\ansatzm}{\frac{1}{\ubar^7}\sum_{|m|\leq2}Q_{m,\sfrak}Y_{m,\sfrak}^{-\sfrak}(\cos\theta)e^{im\phi_{+}}}
\newcommand{\ansatzp}{\frac{1}{\ubar^{7}}\sum_{|m|\leq2}A_m(r)Q_{m,\sfrak}Y_{m,\sfrak}^{+\sfrak}(\cos\theta)e^{im\phi_{+}}}

\newcommand{\sph}{\mathbb S}

\newcommand{\exptruc}{e^{im\phi_{+}}}

\theoremstyle{plain}
\declaretheorem[name=Theorem, numberwithin=section]{thm}
\declaretheorem[name=Proposition,sibling=thm]{prop}
\declaretheorem[name=Corollary,sibling=thm]{cor}
\declaretheorem[name=Lemma,sibling=thm]{lem}
\declaretheorem[name=Conjecture,sibling=thm]{conj}
\theoremstyle{definition}
\declaretheorem[name=Definition,sibling=thm]{defi}
\theoremstyle{remark}
\declaretheorem[name=Remark,sibling=thm]{rem}

\numberwithin{equation}{section}

\title{Precise asymptotics of the spin $+2$ Teukolsky field in the Kerr black hole interior}

\author[Sebastian Gurriaran]{Sebastian Gurriaran}
\email{sebastian.gurriaran@sorbonne-universite.fr}
\address{Laboratoire Jacques-Louis Lions,
Sorbonne Universit\'{e},
4 place Jussieu 75005 Paris, France}

\begin{document}

\begin{abstract}

    Using a purely physical-space analysis, we prove the precise oscillatory blow-up asymptotics of the spin $+2$ Teukolsky field in the interior of a subextremal Kerr black hole. In particular, this work gives a new proof of the blueshift instability of the Kerr Cauchy horizon against linearized gravitational perturbations that was first shown by Sbierski \cite{sbierski}. In that sense, this work supports the Strong Cosmic Censorship conjecture in Kerr spacetimes. The proof is an extension to the Teukolsky equation of the work \cite{scalarMZ} by Ma and Zhang that treats the scalar wave equation in the interior of Kerr. The analysis relies on the generic polynomial decay on the event horizon of solutions of the Teukolsky equation that arise from compactly supported initial data, as recently proved by Ma and Zhang \cite{pricelaw} and Millet \cite{millet} in subextremal Kerr.

\end{abstract}

\maketitle

\tableofcontents

%%%%%%%%%%%%%%%%%%%%%%%%%%%%%

%%%%%%%%%%%%%%%%%%%%%%
\section{Introduction}
\subsection{The Kerr black hole interior}
We begin with a review of the main features of the interior of Kerr black holes. The Kerr metric describes the spacetime around and inside a rotating non-charged black hole. It is a two-parameter family of stationary and axisymmetric solutions of the Einstein {vacuum equations,
\begin{align}\label{eqn:eve}
    \mathrm{Ric}[\mathbf{g}]=0,
\end{align}
where $\mathrm{Ric}[\mathbf{g}]$ is the Ricci tensor for a Lorentzian metric $\mathbf{g}$.} The two parameters are the mass $M$ and the angular momentum per unit mass $a$ of the black hole. In this work, we only consider subextremal Kerr with non-zero angular momentum, i.e. such that $0<|a|<M$. The Kerr metric is given in Boyer-Lindquist coordinates $(t,r,\theta,\phi)\in\mathbb{R}^2\times\mathbb{S}^2$ by 
$$\mathbf{g}_{a,M}=-\frac{(\Delta-a^2\sin^2\theta)}{\Sigma}\dee t^2-\frac{4aMr}{\Sigma}\sin^2\theta \dee t \dee \phi+\frac{\Sigma}{\Delta}\dee r^2+\Sigma\dee \theta^2+\frac{{(r^2+a^2)^2-a^2\sin^2\theta\Delta}}{\Sigma}\sin^2\theta \dee \phi^2\:,$$
where $\Delta:=r^2-2rM+a^2$, $\Sigma:=r^2+a^2\cos^2\theta$. We define
$$\mu:=\frac{\Delta}{r^2+a^2},\quad r_{\pm}:=M\pm\sqrt{M^2-a^2},$$
where $r_{\pm}$ are the roots of $\Delta$. {The level sets $\{r=r_\pm\}$ are Boyer-Lindquist coordinate degeneracies that vanish} when considering Eddington-Finkelstein coordinates
$$u:=r^*-t,\quad\ubar:=r^*+t,\quad\phi_{\pm}:=\phi\pm r_{mod}\:\:\text{mod}\:2\pi,$$
where $\dee r^*/\dee r=\mu^{-1}$, $\dee r_{mod}/\dee r=a/\Delta$, see Section \ref{section:geometry}. The event horizon $\{r=r_+\}$ and the Cauchy horizon $\{r=r_-\}$ can then be properly attached to the Lorentzian manifold $(r_-,r_+)\times\R\times \mathbb{S}^2$ equipped with the metric $\mathbf{g}_{a,M}$, see \cite[Chapter 2]{oneill} for more details. We will mainly be interested in a region containing the right event horizon and the right Cauchy horizon that are respectively defined by
$$\mch_+:=\{r=r_+\}\cap\{u=-\infty\},\quad\ch:=\{r=r_-\}\cap\{\ubar=+\infty\}.$$
Note that a result similar to our main theorem can be deduced in a region containing the left event and Cauchy horizons, that are defined by
$$\mch_+':=\{r=r_+\}\cap\{\ubar=-\infty\},\quad\mathcal{CH}'_+:=\{r=r_-\}\cap\{u=+\infty\}.$$
We denote $\sph^2_{\mch}:=\mch'_+\cap\mch_+$ and $\sph^2_{\mathcal{CH}}:=\mathcal{CH}'_+\cap\ch$ the bifurcations spheres, and $i_+$, $\mathcal{I}_+$ (resp. $i_+'$, $\mathcal{I}_+'$) the right (resp. left) timelike and null infinities.

In this work, by ‘Kerr interior' we mean the resulting Lorentzian manifold $(\mcm,\mathbf{g}_{a,M})$ which is $((r_-,r_+)\times\R\times \mathbb{S}^2,\mathbf{g}_{a,M})$ to which we attach its boundaries, the event and Cauchy horizons. We will be interested in the asymptotics at $\ch$ of solutions to the Teukolsky equation, that arise from compactly supported initial data on a spacelike hypersurface $\Sigma_0$. See Figure \ref{fig:sigma0} for the Penrose diagram of Kerr interior and an illustration of the hypersurface $\Sigma_0$.
\begin{figure}[h!]
    \centering
    \includegraphics[scale=0.4]{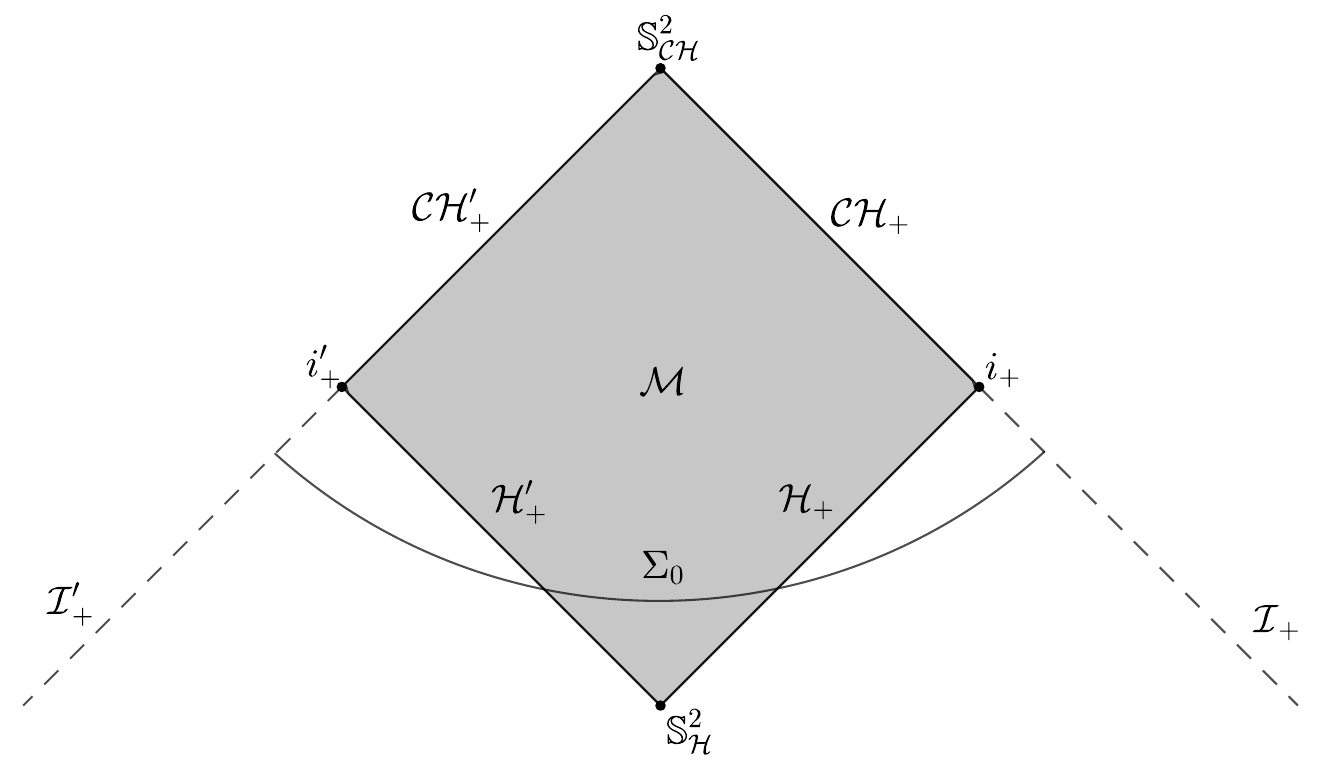}
    \caption{Kerr interior $\mcm$ in grey, Kerr exterior in white, and $\Sigma_0$.}
    \label{fig:sigma0}
\end{figure}
\subsection{Teukolsky equations}\label{section:teukolskyeq}
Teukolsky \cite{teukolsky} found that, when linearizing a gravitational perturbation of a Kerr black hole in the Newman-Penrose formalism, two curvature components decouple from the linearized gravity system and satisfy wave equations: the \textit{Teukolsky equations}. {These curvature components, which govern the linearized dynamics, are called the \textit{Teukolsky scalars}. Before stating the Teukolsky equations, we first define the following pair of null vector fields:}
{$$n:=\frac{r^2+a^2}{2\Sigma}\partial_t-\frac{\Delta}{2\Sigma}\partial_r+\frac{a}{2\Sigma}\partial_\phi,\quad l:=\frac{r^2+a^2}{\Delta}\partial_t+\partial_r+\frac{a}{\Delta}\partial_\phi,$$}
which are aligned with the principal null directions of Kerr spacetime, and the complex vector field
$$m:=\frac{1}{\sqrt{2}(r+ia\cos\theta)}\left(ia\sin\theta\partial_t+\partial_\theta+\frac{i}{\sin\theta}\partial_\phi\right).$$
Then, denoting $\dot{\mathbf{R}}$ the linearized curvature tensor, the scalars
$$\psihatp:=\dot{\mathbf{R}}_{lmlm},\quad\psihatm:=(r-ia\cos\theta)^4\dot{\mathbf{R}}_{n\overline{m}n\overline{m}},$$
are called respectively the spin $+2$ and spin $-2$ Teukolsky scalars. They satisfy the Teukolsky equations that write, for $s=\pm 2$,
\begin{align}
    & -\left[\frac{\left(r^2+a^2\right)^2}{\Delta}-a^2 \sin ^2 \theta\right] \partial_t^2\psihat_s -\frac{4 M a r}{\Delta}\partial_t\partial_\phi\psihat_s  -\left[\frac{a^2}{\Delta}-\frac{1}{\sin ^2 \theta}\right] \partial_\phi^2\psihat_s \nonumber \\
    & +\Delta^{-s} \partial_r\left(\Delta^{s+1} \partial_r \psihat_s\right)+\frac{1}{\sin \theta} \partial_\theta\left(\sin \theta \partial_\theta \psihat_s\right)+2 s\left[\frac{a(r-M)}{\Delta}+\frac{i \cos \theta}{\sin ^2 \theta}\right] \partial_\phi\psihat_s  \nonumber\\
    & +2 s\left[\frac{M\left(r^2-a^2\right)}{\Delta}-r-i a \cos \theta\right]\partial_t\psihat_s -\left[\frac{s^2 \cos ^2 \theta}{\sin ^2 \theta}-s\right]\psihat_s=0.\label{eqn:premteuk}
\end{align}

The rescaled scalars 
$$\psip:=\Delta^2\psihatp,\quad\psim:=\Delta^{-2}\psihatm,$$
satisfy a rescaled version of the Teukolsky equation, see Section \ref{subsection:teukoperators} for the precise definition of the Teukolsky wave operators. Notice that $\psip$ and $\psim$ are projections of the linearized curvature on a frame that is regular on $\mch_+$, while $\psihatp$, $\psihatm$ are projections of the linearized curvature on a frame that is regular on $\ch$. Thus the main result of this work, namely the blow-up asymptotics of $\psihatp$ at $\ch$, is a linear curvature instability statement for the Kerr Cauchy horizon.

{The Teukolsky equations were originally introduced to study the stability of the exterior of black holes, for example in \cite{schlinDHR} for the linear stability of the exterior of Schwarzschild black holes, and in \cite{mamorawetz, mablueandersson} for the linearized stability of Kerr black holes. In the non-linear setting, they were used in \cite{schpolarized, schNLDHRT} to prove the non-linear stability of the exterior of Schwarzschild black holes. }For a full review of the literature concerning the Teukolsky equations, see the introduction of \cite{KSwaveeq}, where the decay estimates for the nonlinear analog of the Teukolsky equations derived in \cite{KSwaveeq} are used to prove the nonlinear stability of the exterior of slowly rotating Kerr black holes in \cite{KS21}.

\subsection{Blueshift effect and Strong Cosmic Censorship conjecture}
\subsubsection{Blueshift effect}

The Kerr spacetime is globally hyperbolic only up to the Cauchy horizon $\ch\cup\mathcal{CH}'_+$. {Indeed, there is an infinite number of smooth extension of  Kerr spacetime across the Cauchy horizon as a regular solution to the Einstein vacuum equations \eqref{eqn:eve}, which give rise to a failure of determinism in exact Kerr.}

However, it is expected that this feature is just an {artifact} of exact Kerr. Indeed, realistic astrophysical black holes are perturbations of Kerr, i.e. they are the maximal globally hyperbolic development of initial data close to the one of Kerr. It is expected that these perturbations kill the non-physical features of exact Kerr. In the linear setting, this is the so-called blueshift effect, introduced by Simpson and Penrose in 1972 \cite{blueshift}. It is a heuristic argument according to which the geometry of Kerr interior (they initially wrote the argument for Reissner–Nordström) 
forces propagating waves to blow-up in some way at the Cauchy horizon. This effect is illustrated on Figure \ref{fig:blueshift}, and is linked to the Strong Cosmic Censorship conjecture.

\begin{figure}[h!]
    \centering

    \definecolor{sqsqbv}{rgb}{0.12549019607843137,0.12549019607843137,0.7098039215686275}
\definecolor{qqqqfq}{rgb}{0,0,0.9411764705882353}
\definecolor{rurued}{rgb}{0.0784313725490196,0.0784313725490196,0.9294117647058824}
\definecolor{ttttff}{rgb}{0.2,0.2,1}
\definecolor{xdxdff}{rgb}{0.49019607843137253,0.49019607843137253,1}
\definecolor{ttqqqq}{rgb}{0.2,0,0}
\begin{tikzpicture}[line cap=round,line join=round,>=triangle 45,x=1cm,y=1cm,scale=0.5]

\draw [line width=1.2pt] (10,14)-- (16,8);
\draw [line width=1.2pt] (16,8)-- (10,2);
\draw [line width=1.2pt,dash pattern=on 3pt off 3pt] (16,8)-- (22,2);
\draw [shift={(-15.17521638908835,6.5899794261234685)},line width=0.8pt]  plot[domain=-0.17012330039041945:0.045198091451180325,variable=\t]({1*31.20708693430449*cos(\t r)+0*31.20708693430449*sin(\t r)},{0*31.20708693430449*cos(\t r)+1*31.20708693430449*sin(\t r)});
\draw [shift={(85.60970106349461,-0.2777922147044219)},line width=0.8pt]  plot[domain=2.979076821423759:3.120574085214437,variable=\t]({1*74.33973037640283*cos(\t r)+0*74.33973037640283*sin(\t r)},{0*74.33973037640283*cos(\t r)+1*74.33973037640283*sin(\t r)});
\draw [shift={(51.27611476877602,4.421579471289647)},line width=0.8pt,dash pattern=on 2pt off 3pt]  plot[domain=2.898559294906488:2.9559621822695226,variable=\t]({1*39.70878857340338*cos(\t r)+0*39.70878857340338*sin(\t r)},{0*39.70878857340338*cos(\t r)+1*39.70878857340338*sin(\t r)});
\draw [line width=0.8pt,color=xdxdff] (16.010006974706656,7.757934337689908)-- (12.217070368589418,11.550870943807144);
\draw [line width=0.8pt,color=ttttff] (16.01830144749907,7.510153270427146)-- (12.184086351333018,11.344368366593194);
\draw [line width=0.8pt,color=rurued] (16.02739763139825,7.1183289446915055)-- (12.132444961856306,11.01328161423345);
\draw [line width=0.8pt,color=qqqqfq] (16.031868367684552,6.601637420395699)-- (12.065413093752838,10.568092694327415);
\draw [line width=0.8pt,color=sqsqbv] (16.020627531690522,5.7523655558693685)-- (11.958250000892743,9.81474308666715);
\draw [line width=0.8pt,color=black] (15.957865081531198,4.442070426886367)-- (11.801756702006124,8.598178806411443);
\draw [fill=ttqqqq] (16,8) circle (1pt);
\draw[color=black] (16.42589592752342,8.542088988265245) node {$i_+$};
\draw[color=black] (14.877091998601705,10.310240335110661) node {$\mathcal{CH}_+$};
\draw[color=black] (13.6005027135667949,4.489481624514186) node {$\mathcal{H}_+$};
\draw[color=black] (19.97090022232352,5.066981377468012) node {$\mathcal{I}_+$};
\draw[color=black] (16.283287575587655,3.1236123385663825) node {$B$};
\draw[color=black] (11.00426743688675,6.4212201964097915) node {$A$};
\end{tikzpicture}
    
    \caption{Illustration of the blueshift effect heuristic. Observer $B$ lives outside of the black hole and reaches timelike infinity in infinite proper time, while sending signals to observer $A$, that crosses the event horizon and then the Cauchy horizon in finite proper time. The waves are sent periodically by $B$, but as $A$ approaches $\ch$, the frequency of the waves becomes infinite.}
    \label{fig:blueshift}
\end{figure}
\subsubsection{Strong Cosmic Censorship conjecture}\label{subsection:SCC}
The Strong Cosmic Censorship (SCC) conjecture was formulated by Penrose in \cite{sccpenrose}, and, in its rough version, states the following :
\begin{conj}[Strong Cosmic Censorship conjecture, rough version]
    The maximal globally hyperbolic development (MGHD) of generic initial data for the Einstein equations is inextendible.
\end{conj}
In other words, this conjecture states {that the failure of determinism in} spacetimes with non-empty Cauchy horizons (for example Kerr and Reissner-Nordström spacetimes) is non-generic, and vanishes upon small perturbations.  See \cite{christo, chruscc} for more modern versions of the SCC conjecture. 

A fundamental question in SCC is the regularity for which the MGHD of initial data should be inextendible. The $C^0$ formulation of SCC was disproved in Kerr by Dafermos and Luk \cite{stabC0}. They showed that generic perturbations of the interior of Kerr still present a Cauchy horizon across which the metric is continuously extendible. They also argued that the perturbed Cauchy horizon may be a so-called weak null singularity, which is a singularity weaker than a spacelike curvature singularity as in Schwarzschild. For references on weak null singularities, see \cite{weakluk}, \cite{weakvdm}. See \cite{sbierskiholo} for a link between weak null singularities and the $C^{0,1}_{loc}$ formulation of SCC. 

In spherical symmetry, the $C^2$ instability of the Cauchy horizon for the model of the Einstein-Maxwell-scalar field system was proven in \cite{RNNLI} and \cite{RNNLII}, extending the results of \cite{dafscc}. See also \cite{vdmmmm,vdminsta} for analog results for the Einstein-Maxwell-Klein-Gordon equations. In Kerr, which is only axisymmetric, the full nonlinear problem for the Einstein equations is still open, and we focus in this paper on the model of linearized gravity, where the Teukolsky scalar $\psihatp$ represents a specific component of the perturbed curvature tensor. Our main result, namely the blow-up of $\Psihatp$ on $\ch$, thus supports the SCC conjecture in the linearized setting.

\subsection{Black hole interior perturbations}\label{section:litterature}
\subsubsection{Results related to Price's law}
The starting point to prove the instability of solutions of the Teukolsky equation in Kerr interior is Price's law for Teukolsky, i.e. the polynomial lower bound on the event horizon for solutions of the Teukolsky equation arising from compactly supported initial data, see \cite{price1, price2, price3, price4} for the original works on Price's law. A version of Price's law for the Teukolsky equations in Kerr was heuristically found by Barack and Ori in \cite{barackori}. 

In this paper we use the precise Price's law asymptotics on $\mch_+$ given by Ma and Zhang in \cite{pricelaw} for the Teukolsky equations\footnote{The Price's law in \cite{pricelaw} holds for $|a|\ll M$, and for $|a|<M$ conditionally on an energy-Morawetz bound. This energy-Morawetz estimate has since been proved for $|a|<M$ by Teixeira da Costa and Shlapentokh-Rothman in \cite{TDCSR2, TDCSR1}, so that the Price's law in \cite{pricelaw} holds for the full range $|a|<M$.}. For another proof of the polynomial lower bound in the full subextremal range $|a|<M$ for solutions of the Teukolsky equation, see the work of Millet \cite{millet}, that uses spectral methods. For a complete account of results related to Price's law, see \cite{pricelaw}.

\subsubsection{Previous results on black hole interior perturbations}

The first works on the linear instability of the Cauchy horizons in Kerr and Reissner-Nordström black holes consisted in finding explicit solutions that become unbounded in some way at the Cauchy horizon, see for example \cite{mac}. In \cite{oriscalar}, a heuristic power tail asymptotic for scalar waves in the interior of Kerr black holes was obtained. Regarding the Teukolsky equations, the oscillatory blow-up asymptotic of our main result in the interior of Kerr black holes, see \eqref{eqn:asympt}, was first predicted heuristically by Ori \cite{ori}, writing the azimuthal $m$-mode of the solution as a late-time expansion ansatz of the form
$$\sum_k\psi_k(r,\theta)t^{-k}.$$
The asymptotic behavior \eqref{eqn:asympt} was also confirmed in a numerical simulation \cite{simunum}.

A rigorous boundedness statement for solutions of the scalar wave equation inside the spherically symmetric Reissner-Nordström spacetime was proven in \cite{franzen}. Still for the scalar wave equation in the interior of Reissner-Nordström black holes, the blow-up of the energy of generic scalar waves was obtained in \cite{RNscalar}.  A scattering approach to Cauchy horizon instability in Reissner-Nordström, as well as an application to mass inflation, was presented in \cite{lukk}, on top of the non-linear instability results \cite{RNNLI,RNNLII,dafscc,vdmmmm,vdminsta} already mentionned in Section \ref{subsection:SCC}.

In Kerr, for the scalar wave equation, a generic blow-up result for the energy of solutions on the Cauchy horizon was obtained in \cite{kerrwave}, while the boundedness of solutions at the Cauchy horizon was proven in \cite{hintzkerrwave} in the slowly rotating case. The boundedness result was then extended to the full subextremal range in \cite{franzen2}. A construction of solutions that remain bounded but have infinite energy at the Cauchy horizon was presented in \cite{daf3}. Finally, the precise asymptotics of the scalar field in the interior of a Kerr black hole was proven in \cite{scalarMZ} using a purely physical-space analysis.

Concerning the Teukolsky equations in Kerr interior, the method of proof of \cite{kerrwave} was extended to the spin $+2$ Teukolsky equation in the work \cite{sbierski}, that proved the blow-up of a weighted $L^2$ norm on a hypersurface transverse to the Cauchy horizon, relying on frequency analysis. 
 
The goal of the present paper is to rigorously prove the oscillatory blow-up asymptotics of the spin $+2$ Teukolsky field in the Kerr black hole interior, by extending the physical-space approach of \cite{scalarMZ} to Teukolsky equations, thus providing a new proof of the blow-up results of \cite{sbierski}.

We discussed here the references on black hole interior perturbations that are the most relevant to this work. For a more complete account of the results related to black hole interior perturbations, for example in Schwarzschild interior or in the cosmological setting, see the introduction in \cite{sbierski}.

\subsection{Rough version of the main theorem}\label{section:mainthm}

This paper rigorously proves the  blueshift instability on the Kerr Cauchy horizon for solutions of the spin $+2$ Teukolsky equations, by finding the precise oscillatory blow-up asymptotics of the spin $+2$ Teukolsky scalar. The rough version of the main result of this paper is the following, see Theorem \ref{thm:main} for the precise formulation.
{
\begin{thm}[Main theorem, rough version] Let $\psihatp$ be the spin $+2$ Teukolsky scalar obtained in a principal null frame regular on the Kerr Cauchy horizon, which satisfies the spin $+2$ Teukolsky equation in subextremal Kerr spacetime with $a\neq 0$, and which arises from smooth and compactly supported initial data. Then $\psihatp$ blows up at the Cauchy horizon, exponentially in the Eddington-Finkelstein coordinate $\ubar$, and oscillates at a frequency that blows up at the Cauchy horizon. More precisely, we have the following precise asymptotic behavior near $\ch$:
\begin{align}\label{eqn:asympt}
    \psihatp\sim\frac{\Delta^{-2}(u,\ubar)}{\ubar^7}\sum_{|m|\leq 2}A_m(r_-)Q_{m,2}e^{2imr_{mod}(u,\ubar)}Y_{m,2}^{+2}(\cos\theta)e^{im\phi_{-}},
\end{align}
        where the constants $Q_{m,2}$ depend on the initial data and are generically non-zero, the constants $A_m(r_-)$ are non-zero for $m\neq 0$ (see \eqref{eq:am(r-)}), $\phi_{-}$ is an angular coordinate that is regular on $\ch$, the functions $Y_{m,2}^{+2}(\cos\theta)$ are the spin $+2$ spherical harmonics, and $r_{mod}\sim a\log(r-r_-)$ near $\ch$ (see Sections \ref{section:geometry} and \ref{section:spinscalars} for more details).
\end{thm}}
\begin{rem} We remark the following :
\begin{itemize}
        \item {Note that on every slice $\{u=cst\}$, we have the exponential blow-up
        $$\frac{\Delta^{-2}(u,\ubar)}{\ubar^7}\sim\frac{e^{\frac{r_+-r_-}{2Mr_-}\ubar}}{\ubar^7}\underset{\ubar\to+\infty}{\longrightarrow}+\infty.$$}
        \item The blueshift instability at $\ch$ for $\psihatp$ was first proven recently by Sbierski \cite{sbierski} who showed the blow-up of a weighted $L^2$ norm along a hypersurface transverse to $\ch$. This result suggests a blow-up that is exponential in the Eddington-Finkelstein coordinate $\ubar$, the Cauchy horizon corresponding to $\ubar=+\infty$. We prove in this paper a pointwise exponential blow-up, along with an oscillatory behavior, which were both heuristically predicted by Ori \cite{ori}.
        \item We will use the version of Price's law proven in \cite{pricelaw}. It holds true in subextremal Kerr conditionally on the existence of an energy and Morawetz estimate for the Teukolsky equation in the whole subextremal range $|a|<M$, which was recently proven in \cite{TDCSR2, TDCSR1}.
\end{itemize}

\end{rem}

\subsection{Structure of the proof}
Although our main result will be about the precise asymptotics of $\psihatp$ at $\ch$, we will actually also obtain precise estimates for $\psim$ near $\mch_+$, and we will use the Teukolsky-Starobinsky identities (see Section \ref{section:tsi}) to link $\psip$ and $\psim$ there. 

For $s=\pm 2$, we denote $\teukhat_s$ the differential operator on the left-hand side of \eqref{eqn:premteuk}, called the Teukolsky operator, and $\teuk_s=\Delta^{s}\teukhat_s(\Delta^{-s}\:\cdot\:)$ the rescaled Teukolsky operator, such that $$\teukhat_s\psihat_s=0,\quad\teuk_s\psi_s=0.$$
The analysis is done entirely in physical space, using energy estimates to prove upper bounds on \app error' quantities, in the hope that this method of proof is robust enough to be applied in the non-linear setting. The first step is to notice that the energy estimates to get polynomial upper bounds done in \cite{scalarMZ} for the scalar wave equation $\Box_g\psi=0$ in the Kerr interior can be extended to the Teukolsky equations $\teuk_s\psi_s=0$, but only for negative spin close to $\mch_+$, and for positive spin close to $\ch$. This is because we need a fixed sign for the scalar $s(r-M)$ that appears at crucial places in the energy estimates, where $s=\pm 2$ is the spin, and because $r_+-M>0$ while $r_--M<0$.

We denote $\un$ the region containing $\mch_+$, $\trois$ and $\quatre$ the regions close to $\ch$ where $\Delta$ has exponential decay in $\ubar$, and $\deux$ the intermediate region between $\un$ and $\trois\cup\quatre$. See Figure \ref{fig:regions2}, and Section \ref{section:regions} for the precise definitions of the regions. {The motivation for dividing the Kerr black hole interior in this manner will be clarified in the main steps of the proof below. Roughly stated, region $\un$ is the redshift region where we use $r-M>0$ near $\mch_+$, and $-\Delta\sim 1$ in the rest of region $\un$ to obtain a redshift energy estimate for $\Psim$. Region $\deux$ is still relatively far from $\ch$ so there, relying on $r-M<0$, we are able to use an effective blueshift energy estimate to control the degenerate Teukolsky field $\Psip$. Regions $\trois$ and $\quatre$ contain $\ch$, and their main feature is that $\Delta$ decays exponentially (respectively in $\ubar$ and $u,\ubar$) towards $\ch$ in $\trois$ and $\quatre$. Using this exponential decay, we obtain \eqref{eqn:asympt} in $\trois$ by integrating an algebraic identity extracted from the Teukolsky equation, see item \eqref{item:1+1} below. Then some technical work is required to propagate the asymptotic to region $\quatre$.}

The assumptions on $\mch_+$ that we will use are the ones given by \cite{pricelaw}, i.e. that the error quantities
\begin{align}
    \errm&:=\psim-\ansatzm,\\
    \errp&:=\psip-\ansatzp,\label{eqn:deferrp}
\end{align}
are bounded by $\ubar^{-7-\delta}$ on $\mch_+\cap\{\ubar\geq 1\}$ where $\delta>0$, see Sections \ref{section:geometry} and \ref{section:assumptions} for the definitions of $Q_{m,2}$, $A_m(r)$, $Y_{m,2}^{\pm 2}(\cos\theta)$, $\phi_+$. {These polynomial bounds on $\mch_+$ for $\mathrm{Err}[\psi_{\pm 2}]$ is the statement of Price's law in that case.} Defining the rescaled null pair 
$$e_3=-\frac{2\Sigma}{\Delta}n,\quad e_4=\frac{\mu}{2}l,$$
the main steps of the proof then go as follows :
\begin{enumerate}    \item First, we propagate Price's law lower bound for $\psim$ in $\un$. To do this, we propagate the $O(\ubar^{-7-\delta})$ upper bound of $\errm$ from $\mch_+$ to $\un$ using a redshift energy estimate. This is only possible for the spin $-2$ Teukolsky field, in region $\un$ that contains the event horizon. {The argument also relies on the following key computational fact: 
$$\teuk_{-2}\errm=O(\ubar^{-8}),$$
which is obtained using the Teukolsky equation $\teuk_{-2}\psim=0$, the identity $e_3(\psim-\errm)=\drond\drond'(\psim-\errm)=0$ where $\drond\drond'$ is the spin-weighted Laplacian, and the fact that the remaining terms in $\teuk_{-2}(\psim-\errm)$ all involve at least one $\partial_t$ derivative.} 
    \item We then use the Teukolsky-Starobinsky identity \eqref{eqn:tsi2} to propagate the $O(\ubar^{-7-\delta})$ Price's law upper bound on $\mch_+$ for $\errp$ into $\un$.
    \item Next, similarly as in $\un$ we propagate the lower bound for $\psip$ in $\deux$ using an effective blueshift energy estimate for $\errp$, which is only possible for $r$ close to $r_-$, and for the spin $+2$ field. {Similarly as for the spin $-2$ field in $\un$, we rely on the following key computational fact: 
    $$\teuk_{+2}\errp=O(\ubar^{-8}),$$
    see Lemma \ref{prop:teukansatz}.} This step in region $\deux$ is necessary to propagate the lower bound up to $\trois$ where the analysis becomes more delicate\footnote{{In region $\trois$, the geometry of the Kerr black hole interior near the Cauchy horizon is such that we cannot prove sharp decay of the energy using the same energy method as in regions $\un$ and $\deux$.}}.
    \item We then
    obtain a non-sharp $L^\infty$ bound for $\psip$ in $\trois$ using an energy estimate.
    \item\label{item:1+1} To get the blow-up asymptotics in $\trois$, we rewrite the Teukolsky equation $\teukhat_{+2}\psihatp=0$ as a $1+1$ wave equation in $(u,\ubar,\theta,\phi_-)$ coordinates, 
    \begin{align}\label{eqn:reflol}
        \left(\partial_u-\frac{a}{r^2+a^2}\partial_{\phi_-}\right)\left((r^2+a^2)\Delta^2\partial_\ubar\Psihatp\right)=O(\Delta),
    \end{align}
        where we use the previous $L^\infty$ bound to control the right-hand-side. Integrating \eqref{eqn:reflol} from $\Gamma$ (see Figure \ref{fig:regions2}) directly gives the blow-up asymptotics for $\psihatp$ stated in \eqref{eqn:asympt}. This is where we use the definition of region $\trois$ : $\Delta$ decays exponentially in $\ubar$ in $\trois$, which easily bounds the error terms. 
        \item We cannot integrate the $1+1$ wave equation in $\quatre$ because the slices $\{u=cst\}$ starting inside $\quatre$ do not cross $\Gamma$. Instead, commuting $e_3$ and the Teukolsky operator, we estimate $e_3\psip$ in $\quatre$, which allows us to propagate the blow-up of $\psihatp$ from $\trois$ to $\quatre$ by integration.
\end{enumerate}
\begin{rem}Here are some further remarks on the analysis :

\begin{itemize}
    \item To control the derivatives of the error terms, we commute the Teukolsky equation with operators that have good commutations properties, namely $\partial_t$, $\partial_\phi$, $e_3$, and the Carter operator.
    \item Note that, at any point in the analysis, in any region $\un$, $\deux$, $\trois$, $\quatre$, any estimate that we get can be turned into an estimate in the symmetric region $\{u\geq 1\}$ by replacing $\ubar$ with $u$, $\psi_s$ with $\psihat_{-s}$, $e_3$ with $\Delta^{-1}e_4$ and $e_4$ with $\Delta e_3$. As in \cite{scalarMZ}, this will be useful in region $\quatre$. 
\end{itemize}
Here are some differences from the previous works in the method of proof :

\begin{itemize}
    \item Our physical-space proof differs substantially from \cite{sbierski} that relies in a crucial way on frequency analysis.
    \item In \cite{scalarMZ} for the scalar wave equation, the proof is made by decomposing the scalar field $\psi$ onto its $\ell=0$ and $\ell\geq 1$ modes : $\psi=\psi_{\ell=0}+\psi_{\ell\geq 1}$. Using the decay of both quantities on the event horizon, it is shown, commuting the wave equation with the projection on the modes, that the $\ubar^{-3}$ lower bound for $\psi_{\ell=0}$ and the $O(\ubar^{-4-\delta})$ better decay of $\psi_{\ell\geq 1}$ propagate. This is different from our proof, where there is no projection on the (spin-weighted) spherical harmonics. We do not need to assume that the modes $\ell\geq 3$ of $\psi_{\pm 2}$ decay better than $\ubar^{-7}$ on $\mch_+$. Moreover, in the blueshift region (i.e. for $r$ close to $r_-$) we will use easier estimates than \cite{scalarMZ} that introduces $\log$ multipliers, taking advantage of the fact that we deal with a non-zero spin.
\end{itemize}

\end{rem}
\subsection{Overview of the paper}
In Section \ref{section:preliminaries}, we introduce the geometric background, the operators, the coordinates that we will use in the analysis, as well as the system of equations that we consider. In Section \ref{section:eq+assumptions}, we recall the decay assumptions on the event horizon, the so-called Price's law, and we write the precise version of the main theorem. In Section \ref{section:redshift}, we obtain the redshift energy estimates in $\un$ for $\psim$ and use the Teukolsky-Starobinsky identity to get the lower bound for $\psip$ in $\un$. Section \ref{section:blueshiftII} deals with the effective blueshift energy estimates to propagate the lower bound for $\psip$ in $\deux$. In Section \ref{section:blueshiftIII}, we get the non-sharp $L^\infty$ bound for $\psip$ in $\trois$ and we compute and integrate the $1+1$ wave equation that will eventually give the precise oscillatory blow-up asymptotics for $\psihatp$ in $\trois$, before propagating this blow-up to region $\quatre$.

\subsection{Acknowledgments}
The author would like to thank Jérémie Szeftel for his support and many helpful discussions, and Siyuan Ma for helpful suggestions. This work was partially supported by ERC-2023 AdG 101141855 BlaHSt.

\section{Preliminaries}\label{section:preliminaries}
We first introduce some notations. By \app LHS' and \app RHS' we mean respectively \app left-hand side' and \app right-hand side'. If $P$ is an operator acting on a (spin-weighted) scalar $\psi$, for an integer $k\geq 1$ and for any norm $\|\cdot\|$, we use the notation 
$$\|P^{\leq k}\psi\|:=\sum_{j=0}^k\|P^j\psi\|.$$
If $f,g$ are two non-negative scalars, we write $f\lesssim g$ whenever there is a constant $C>0$ that depends only on the black hole parameters $a,M$, on the initial data, and on the smallness constants $\rb$, $\gamma$, such that $f\leq Cg$ on the region considered. We write $f=O(g)$ when $|f|\lesssim |g|$. We write $f\sim g$ when $f\lesssim g$ and $g\lesssim f$.

\subsection{Geometry of the interior of subextremal Kerr spacetimes}\label{section:geometry}
First, we fix a notation for the Boyer-Lindquist (B-L) coordinate Killing vector fields : 
 $$T:=\partial_t,\quad\Phi:=\partial_\phi.$$
Note that when we write $\partial_r$, $\partial_\theta$, we mean the B-L coordinate vector fields. We recall the following null pair, which is a more convenient rescaled version of $(n,l)$:
\begin{align}
    e_3:=\frac{1}{2}\left(-\partial_r+\frac{r^2+a^2}{\Delta}T+\frac{a}{\Delta}\Phi\right),\quad\quad e_4:=\frac{1}{2}\left(\frac{\Delta}{r^2+a^2}\partial_r+T+\frac{a}{r^2+a^2}\Phi\right).
\end{align}
This pair satisfies
$$\mathbf{g}_{a,M}(e_3,e_3)=\mathbf{g}_{a,M}(e_4,e_4)=0,\quad \mathbf{g}_{a,M}(e_3,e_4)=-\frac{\Sigma}{2(r^2+a^2)},$$
and is regular on $\mch_+$, as can be seen by expressing $e_3$ and $e_4$ with the ingoing Eddington-Finkelstein coordinate vector fields, see below. We also define the rescaled null pair
\begin{align}
    \ethat:=(-\mu)e_3,\quad \eqhat:=(-\mu)^{-1}e_4,
\end{align}
that is regular on $\ch$.

As recalled in the introduction, the B-L coordinates do not cover the event and Cauchy horizons. In this paper, we will use both the Eddington-Finkelstein coordinates and double null-like coordinates (we borrow the terminology from \cite{scalarMZ}) introduced in Section \ref{section:doublenulllike}. Notice that the scalar function 
$$\mu=\frac{\Delta}{r^2+a^2}$$
vanishes only on the horizons, and that unlike in the Kerr black hole exterior region, $\mu\leq 0$ on $\mcm$. Define the tortoise coordinate $r^*$ by 
$$\frac{\dee r^*}{\dee r}=\mu^{-1},\quad r^*(M)=0.$$
Notice that $r^*\to -\infty$ as $r\to r_+$ and $r^*\to +\infty$ as $r\to r_-$. We denote by $\kappa_+$ and $\kappa_-$ the surface gravities of the event and Cauchy horizon, which are defined by
$$\kappa_+:=\frac{r_+-r_-}{4Mr_+}>0,\quad\kappa_-:=\frac{r_--r_+}{4Mr_-}<0.$$
The tortoise coordinate $r^*$ can be expressed as
\begin{align}\label{eqn:rstar}
    r^*=r+\frac{1}{2\kappa_+}\log|r-r_+|+\frac{1}{2\kappa_-}\log|r-r_-|-M-\frac{1}{2\kappa_+}\log|M-r_+|-\frac{1}{2\kappa_-}\log|M-r_-|.
\end{align}
Notice that \eqref{eqn:rstar} implies that, for $r$ close to $r_+$, 
$$-\Delta\sim\exp(2\kappa_+r^*),$$
while for $r$ close to $r_-$, 
$$-\Delta\sim\exp(-2|\kappa_-|r^*).$$
Next, we define the coordinates
$$\ubar:=r^*+t,\quad u:=r^*-t.$$
The range of the coordinates $u,\ubar,r^*$ is indicated on Figure \ref{fig:range}. Recall that the right event horizon corresponds to $\{u=-\infty\}$, while the right Cauchy horizon corresponds to $\{\ubar=+\infty\}$.
\begin{figure}[h!]
    \centering
    \includegraphics[scale=0.5]{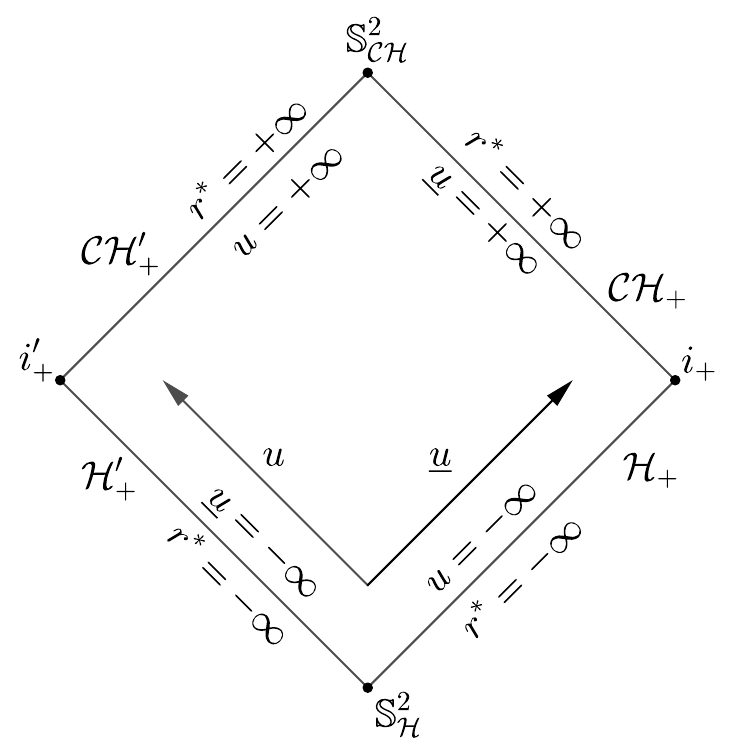}
    \caption{The range of the coordinates $u,\ubar,r^*$ in the Kerr black hole interior.}
    \label{fig:range}
\end{figure}

As in \cite{scalarMZ}, define also the function $r_{mod}(r)$ by
$$\frac{\dee r_{mod}}{\dee r}=\frac{a}{\Delta},\quad r_{mod}(M)=0.$$
Then define the {ingoing} and {outgoing} angular coordinates
$$\phi_{+}:=\phi+r_{mod}\:\:\text{mod}\:2\pi,\quad \phi_{-}:=\phi-r_{mod}\:\:\text{mod}\:2\pi.$$
The coordinates $\ubar$ and $\phi_+$ are regular on $\mch_+=\{u=-\infty\}$ and $\mathcal{CH}'_+=\{u=+\infty\}$, while $u$ and $\phi_-$ are regular on $\mch'_+=\{\ubar=-\infty\}$ and $\ch=\{\ubar=+\infty\}$.
\subsubsection{Ingoing Eddington-Finkelstein coordinates}
The ingoing Eddington-Finkelstein coordinates are $(\ubar_{in}=\ubar, r_{in}=r, \theta_{in}=\theta, \phi_{+})$. It is a set of coordinates that is regular on $\mch_+$ and $\mathcal{CH}'_+$. The coordinate vector fields are 
\begin{align}\label{eqn:coordEFin}
    \partial_{\ubar_\mathrm{in}}=T,\quad \partial_{r_\mathrm{in}}=-2e_3,\quad \partial_{\theta_{in}}=\partial_\theta,\quad \partial_{\phi_{+}}=\Phi.
\end{align}
\subsubsection{Outgoing Eddington-Finkelstein coordinates}
The outgoing Eddington-Finkelstein coordinates are $(u_{out}=u, r_{out}=r, \theta_{out}=\theta, \phi_{-})$. It is a set of coordinates that is regular on $\mch'_+$ and $\ch$. The coordinate vector fields are 
\begin{align}\label{eqn:coordEFout}
    \partial_{u_\mathrm{out}}=T,\quad \partial_{r_\mathrm{out}}=-2\eqhat,\quad \partial_{\theta_{out}}=\partial_\theta,\quad \partial_{\phi_{-}}=\Phi.
\end{align}
\subsubsection{Double null-like coordinate systems}\label{section:doublenulllike}
As in \cite{scalarMZ}, we also use the \emph{ingoing} double null-like coordinates $(\ubar,u,\theta,\phi_{+})$ that are regular at $\mch_+$ (where $u=-\infty$), with coordinate vector fields
\begin{align}\label{eqn:coordDNLin}
    \partial_\ubar=e_4-\frac{a}{r^2+a^2}\Phi,\quad\partial_u=-\mu e_3, \quad\partial_\theta=\Theta,\quad\partial_{\phi_{+}}=\Phi,
\end{align}
where we denote $\Theta$ the B-L coordinate vector field $\partial_\theta$ to avoid confusion. The equivalent \emph{outgoing} double null-like coordinates $(\ubar,u,\theta,\phi_{-})$ are regular at $\ch$ (where $\ubar=+\infty$), with coordinate vector fields
\begin{align}\label{eqn:coordDNLout}
    \partial_\ubar=e_4,\quad\partial_u=-\mu e_3+\frac{a}{r^2+a^2}\Phi, \quad\partial_\theta=\Theta,\quad\partial_{\phi_{-}}=\Phi.
\end{align}
Note that we will only use the ingoing double null-like coordinate system in the redshift region $\un$, and the outgoing double null-like coordinate system in the blueshift region $\deux\cup\trois\cup\quatre$ so we use the same notations for the ingoing and outgoing double null-like coordinate vector fields $\partial_u$, $\partial_\ubar$, as there is no danger of confusion.
\subsubsection{Constant $\wbar$ and $w$ spacelike hypersurfaces}
We need a family of spacelike hypersurfaces, for which we will apply the energy estimates and get decay in $\ubar$. 
\begin{figure}[h!]
    \centering
    \includegraphics[scale=0.45]{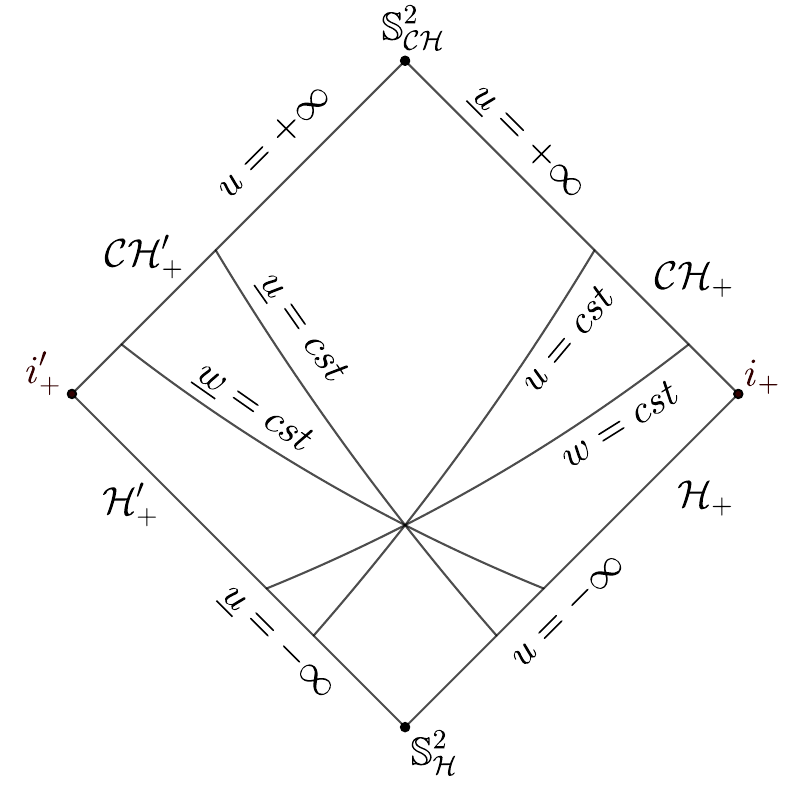}
    \caption{Examples of constant $u,\ubar,w,\wbar$ hypersurfaces and their causal nature.}
    \label{fig:range2}
\end{figure}
Note that the constant $\ubar, u$ hypersurfaces are not spacelike. Indeed we have $\mathbf{g}_{a,M}(\nabla\ubar,\nabla\ubar)=\mathbf{g}_{a,M}(\nabla u,\nabla u)={a^2\sin^2\theta}/{\Sigma}$ so that the constant $\ubar, u$ hypersurfaces are null at the poles and timelike away from the poles. We define, as in \cite{scalarMZ}, 
$$\wbar:=\ubar-r+r_+,\quad w:=u-r+r_-,$$
such that the constant $\wbar$ and $w$ hypersurfaces are spacelike. Indeed, we have \cite{scalarMZ} 
$$\mathbf{g}_{a,M}(\nabla\wbar,\nabla\wbar)=\mathbf{g}_{a,M}(\nabla w,\nabla w)=-\frac{r^2+2Mr+a^2\cos^2\theta}{\Sigma}<-1.$$

\subsection{Spin-weighted scalars}\label{section:spinscalars}
\subsubsection{Spin-weighted scalars and spin-weighted spherical operators}
In this section, we consider the round sphere $\sph^2$ equipped with its volume element, written in coordinates $(\theta,\phi)\in(0,\pi)\times[0,2\pi)$ :  
$$\dee\nu := \sin\theta\dee\theta\dee\phi.$$

Let $s$ be an integer. Note that in this work we only consider $s=\pm 2$. A spin $s$ scalar is a scalar function that has zero boost weight and proper spin weight, as defined by Geroch, Held and Penrose \cite{spin}. {Roughly speaking, a spin $+2$ (respectively spin $-2$) scalar is a the contraction $\alpha({\eta},{\eta})$ (respectively $\alpha(\overline{\eta},\overline{\eta})$) of a smooth 
symetric tensor $\alpha$ on $\mathbb{S}^2$, where $\eta=\partial_\theta+\frac{i}{\sin\theta}\partial_\phi$.} See \cite{sbierski} for a rigorous presentation of the spaces of spin-weighted scalars on the sphere $\sph^2$ and on the Kerr interior, and a proof that the Teukosky scalar obtained in the linearisation of a gravitational perturbation in the Newman-Penrose formalism is a spin-weighted scalar on spacetime. See also \cite{millet2} for a precise review of the geometric background of the Teukolsky equation. By 'spin-weighted operator' we mean an operator that takes a spin-weighted scalar into a spin-weighted scalar. Note that a spin $0$ scalar is a scalar function.

In Kerr spacetime, the volume element induced on the topological spheres 
$$S(u,\ubar):=\{r=r(u,\ubar)\}\cap\{t=t(u,\ubar)\}$$
by the metric $\mathbf{g}_{a,M}$ is 
$$\mathrm{vol}_{S(u,\ubar)}=\sigma\sin\theta\dee\theta\dee\phi$$
where {$\sigma:=\sqrt{(r^2+a^2)^2-a^2\sin^2\theta\Delta}\sim 1$} in the Kerr interior. Thus, although they are not round, we still rely on the round volume element $\dee\nu$ on the Kerr spheres $S(u,\ubar)$ to define $L^2(S(u,\ubar))$ norms. 
\begin{defi}
    For $\psi$ a spin-weighted scalar on $\mcm$, we define
$$\|\psi\|_{L^2(S(u,\ubar))}:=\left(\int_{S(u,\ubar)}|\psi^2|\dee\nu\right)^{1/2}.$$
\end{defi}
Notice that on a region $\{r_0\leq r\leq r_+\}$ with $r_0\in(r_-,r_+)$, in the ingoing Eddington-Finkelstein coordinates, as $\phi_+-\phi=r_{mod}$ is constant on $S(u,\ubar)$, the definition of the $L^2(S(u,\ubar))$ norm gives 
$$\|\psi\|_{L^2(S(u,\ubar))}^2=\int_0^\pi\int_0^{2\pi}|\psi|^2(u,\ubar,\theta,\phi_+)\sin\theta\dee\theta\dee\phi_+,$$
which is a regular norm up to the event horizon $\{u=-\infty\}\cup\{\ubar=-\infty\}$. For the same reason, on a region $\{r_-\leq r\leq r_0\}$ with $r_0\in(r_-,r_+)$, we have 
$$\|\psi\|_{L^2(S(u,\ubar))}^2=\int_0^\pi\int_0^{2\pi}|\psi|^2(u,\ubar,\theta,\phi_-)\sin\theta\dee\theta\dee\phi_-,$$
which is a regular norm up to the Cauchy horizon $\{\ubar=+\infty\}\cup\{u=+\infty\}$.
\begin{defi}
    We recall the definition of the following standard spin-weighted differential operators, called the spherical eth operators :
    \begin{align}
        \drond_s&:=\partial_\theta+\frac{i}{\cos\theta}\partial_\phi-s\cot\theta,\\
        \drond'_s&:=\partial_\theta-\frac{i}{\cos\theta}\partial_\phi+s\cot\theta.
    \end{align}
\end{defi}
\begin{rem}
    The spherical eth operators modify the spin when applied to a spin-weighted scalar. More precisely, $\drond$ increases the spin by $1$ while $\drond'$ decreases the spin by $1$. See in Section \ref{section:spinweightedspherical} their effect on spin-weighted spherical harmonics. Note that the spin of the scalar to which we apply the eth operators will be clear in the context, so we drop the subscript $s$ in what follows. 
\end{rem}
\begin{defi}
    We define the spin-weighted Laplacian as
    \begin{align}
        \drond\drond'=\frac{1}{\sin\theta}\partial_\theta(\sin\theta\partial_\theta)+\frac{1}{\sin^2\theta}\partial_\phi^2+\frac{2is\cos\theta}{\sin^2\theta}\partial_\phi-(s^2\cot^2\theta+s).
    \end{align}
\end{defi}
\begin{rem}
    We also have the expression
        \begin{align}
        \drond'\drond=\frac{1}{\sin\theta}\partial_\theta(\sin\theta\partial_\theta)+\frac{1}{\sin^2\theta}\partial_\phi^2+\frac{2is\cos\theta}{\sin^2\theta}\partial_\phi-(s^2\cot^2\theta-s)=\drond\drond'+2s.
    \end{align}
\end{rem}
\subsubsection{Spin-weighted spherical harmonics}\label{section:spinweightedspherical}
Let $s$ be a fixed spin. The spin-weighted spherical harmonics are the eigenfunctions of the spin-weighted Laplacian, that is self-adjoint on $L^2(\sph^2)$. They are given by the following family
$$Y_{m,\ell}^s(\cos\theta)e^{im\phi},\quad\ell\geq |s|, -\ell\leq m\leq \ell$$
of spin $s$ scalars on the sphere $\sph^2$. They form a complete orthonormal basis of the space of spin $s$ scalars on $\sph^2$, for the $L^2(\sph^2)$ scalar product. When considering the spheres of Kerr interior, as the B-L coordinate $\phi$ is singular at the horizons, we need a slightly modified family. Let $r_0\in(r_-,r_+)$. Then for $(u,\ubar)$ such that $r(u,\ubar)\in [r_0,r_+]$, the spin $s$ scalars 
$$Y_{m,\ell}^s(\cos\theta)e^{im\phi_+},\quad\ell\geq |s|, -\ell\leq m\leq \ell$$
form a complete orthonormal basis of the space of spin $s$ scalars on $S(u,\ubar)$, for the $L^2(S(u,\ubar))$ scalar product. Similarly, if $r(u,\ubar)\in [r_-,r_0]$, the spin $s$ scalars 
$$Y_{m,\ell}^s(\cos\theta)e^{im\phi_-},\quad\ell\geq |s|, -\ell\leq m\leq \ell$$
form a complete orthonormal basis of the space of spin $s$ scalars on $S(u,\ubar)$, for the $L^2(S(u,\ubar))$ scalar product.

\begin{rem}
    In this work, we will never project equations or spin-weighted scalars on some $\ell$ or $\geq\ell$ modes, unlike in \cite{scalarMZ} and \cite{pricelaw}. Instead, we will simply propagate the leading-order term term\footnote{Namely, the second term in \eqref{eqn:deferrp}.} of $\psip$ on $\mch_+$ up to $\ch$. This dominant term is a linear combination of spin-weighted spherical harmonics, and satisfies key algebraic properties (see Lemma \ref{prop:teukansatz}).
\end{rem}
The following facts are standards. We have :
\begin{align}
    \drond'\drond(Y_{m,\ell}^s(\cos\theta)e^{im\phi_\pm})&=-(\ell-s)(\ell+s+1)Y_{m,\ell}^s(\cos\theta)e^{im\phi_\pm},\label{eq:vpspin}\\
    \drond\drond'(Y_{m,\ell}^s(\cos\theta)e^{im\phi_\pm})&=-(\ell+s)(\ell-s+1)Y_{m,\ell}^s(\cos\theta)e^{im\phi_\pm},\label{eqn:annulemode}\\
    \drond(Y_{m,\ell}^{s}(\cos\theta)e^{im\phi_\pm})&=-\sqrt{(\ell-s)(\ell+s+1)}Y_{m,\ell}^{s+1}(\cos\theta)e^{im\phi_\pm},\label{eqn:drooo}\\
    \drond'(Y_{m,\ell}^s(\cos\theta)e^{im\phi_\pm})&=\sqrt{(\ell+s)(\ell-s+1)}Y_{m,\ell}^{s-1}(\cos\theta)e^{im\phi_\pm}.
\end{align}

\subsection{Functional inequalities}

We start this section by recalling \cite[Lemma 3.4]{scalarMZ} with $\alpha=0$, which will be used to deduce decay of the energy from an energy estimate.
\begin{lem}\label{lem:decay}
    Let $p>1$ and let $f:[1,+\infty) \rightarrow[0,+\infty)$ be a continuous function. Assume that there are constants $C_0>0, C_1>0, C_2 \geq 0$, $C_3 \geq 0$ such that for $1 \leq x_1<x_2$,
$$
f(x_2)+C_1 \int_{x_1}^{x_2}  f(x) \mathrm{d} x \leq C_0 f(x_1)+C_2 \int_{x_1}^{x_2} x^{-p} \mathrm{~d} x+C_3 x_1^{-p} .
$$

Then for any $x_1 \geq 1$,
$$
f(x_1) \leq C x_1^{-p}
$$
where $C$ is a constant that depends only on $f(1), C_0, C_1, C_2, C_3$, and $p$.
\end{lem}
\begin{proof}
    See \cite[Lemma 3.4]{scalarMZ}.
\end{proof}
\subsubsection{$\mcq_s$ and $\mcu$ spin-weighted operators in Kerr spacetime}
\begin{defi}
We define the spin-weighted Carter operator 
\begin{align}\label{eqn:carter}
    \mcq_s:=a^2\sin^2\theta T^2-2ias\cos\theta T +\drond'\drond.
\end{align}
\end{defi}
The main property of the spin-weighted Carter operator is that it commutes with the Teukolsky operator (see \eqref{eqn:opteuk}, \eqref{eqn:opteukhat}). We will use it to bound the angular derivates of solutions to the Teukolsky equation. The following operator is (up to a bounded potential) the spin-weighted scalar equivalent of the tensor covariant derivative $\nabla_2$, where 
$$e_2:=\frac{1}{\sin\theta}\Phi+a\sin\theta T,$$
and will be useful to get a Poincaré-type inequality to absorb the zero order terms of the Teukolsky operator when doing energy estimates.
\begin{defi}We define the spin-weighted operator
\begin{align}\label{eqn:defU}
    \mcu:=\frac{1}{\sin\theta}\Phi+a\sin\theta T+is\cot\theta\frac{\Sigma}{r^2+a^2}.
\end{align}
\end{defi}
Note that $\mcu$ is not regular at the poles, but for a spin-weighted scalar $\psi$ on spacetime we still have $\mcu\psi\in L^\infty(S(u,\ubar))$, see for example \cite[Eq. (2.34)]{sbierski}. We will need the following integration by parts lemma for $\mcu$.
\begin{lem}\label{lem:intpartU}
    Let $\psi$, $\varphi$ be spin $s$ scalars on $\mcm$. Then 
    $$\int_{S(u,\ubar)}\overline{\varphi}\:\mcu\psi\:\dee\nu=T\left(\int_{S(u,\ubar)}a\sin\theta\overline{\varphi}\psi\:\dee\nu\right)-\int_{S(u,\ubar)}\overline{\mcu\varphi}\psi\:\dee\nu.$$
\end{lem}
\begin{proof}
Define the real valued function
$$f(r,\theta):=\cot\theta\frac{\Sigma}{r^2+a^2}.$$
We have
\begin{align*}
    \int_{S(u,\ubar)}\overline{\varphi}\:\mcu\psi\:\dee\nu&=\int_{S(u,\ubar)}\overline{\varphi}\left(\frac{1}{\sin\theta}\Phi+a\sin\theta T\right)\psi\:\dee\nu+\int_{S(u,\ubar)}\overline{\varphi}isf(r,\theta)\psi\:\dee\nu\\
    &=T\left(\int_{S(u,\ubar)}a\sin\theta\overline{\varphi}\psi\:\dee\nu\right)-\int_{S(u,\ubar)}\overline{\left(\frac{1}{\sin\theta}\Phi+a\sin\theta T\right)\varphi}\psi\:\dee\nu\\
    &\quad\quad\quad-\int_{S(u,\ubar)}\overline{isf(r,\theta)\varphi}\psi\:\dee\nu\\
    &=T\left(\int_{S(u,\ubar)}a\sin\theta\overline{\varphi}\psi\:\dee\nu\right)-\int_{S(u,\ubar)}\overline{\mcu\varphi}\psi\:\dee\nu,
\end{align*}
as stated.
\end{proof}
\begin{rem}
    Notice that we have the expression
    \begin{align}\label{eqn:expreT}
        T&=\frac{r^2+a^2}{\Sigma}(\mu e_3+e_4)-\frac{a\sin\theta}{\Sigma}\mcu+\frac{ias\cos\theta}{r^2+a^2},\\        \Phi&=\frac{\sin\theta(r^2+a^2)}{\Sigma}(\mcu-a\sin\theta e_4-a\sin\theta\mu e_3)-is\cos\theta.\label{eqn:exprephi}
    \end{align}
    In particular, $T=O(\mu)e_3+O(1)e_4+O(1)\mcu+O(1)$, and $\Phi=O(\mu)e_3+O(1)e_4+O(1)\mcu+O(1)$.
\end{rem}
\subsubsection{Poincaré and Sobolev inequalities for spin-weighted scalars}
We first recall the standard Poincaré inequality for spin-weighted scalars, see for example \cite[Eq. (34)]{dhrteuk}.
\begin{prop}
    Let $\psi$ be a spin $\pm 2$ scalar on $\sph^2$. Then 
    \begin{align}\label{eqn:dhrpoinc}
    2\int_{\sph^2}|\psi|^2\dee\nu&\leq\int_{\sph^2}\left(|\partial_\theta\psi|^2+\left|\frac{1}{\sin\theta}\partial_\phi\psi+is\cot\theta\psi\right|^2\right)\dee\nu.
\end{align}
\end{prop}
Next, we define the energy for wich we will show decay for solutions of the Teukolsky equation. 
\begin{defi}
    Let $\psi$ be a spin-weighted scalar on $\mcm$. We define its energy density and degenerate energy density as the scalars 
\begin{align}
    \mathbf{e}[\psi]&:=|e_3\psi|^2+|e_4\psi|^2+|\mcu\psi|^2+|\partial_\theta\psi|^2,\label{eqn:defener}\\
    \mathbf{e}_{deg}[\psi]&:=\mu^2|e_3\psi|^2+|e_4\psi|^2+|\mcu\psi|^2+|\partial_\theta\psi|^2.\label{eqn:defenerdeg}
\end{align}
\end{defi}
The following result is a re-writing of the standard Poincaré inequality \eqref{eqn:dhrpoinc}, that will be useful in the energy estimates.
\begin{prop}
Let $\psi$ be a spin $\pm2$ scalar on $\mcm$. For any $(u,\ubar)$ we have the Poincaré inequality
\begin{align}\label{eqn:poincare}
    \|\psi\|^2_{L^2(S(u,\ubar))}{\lesssim}\int_{S(u,\ubar)}\mathbf{e}_{deg}[\psi]\dee\nu.
\end{align}
\end{prop}
\begin{proof}
By the standard Poincaré inequality \eqref{eqn:dhrpoinc}, we have 
\begin{align}\label{eqn:zzmm}
    \|\psi\|^2_{L^2(S(u,\ubar))}\lesssim\intS\left(|\partial_\theta\psi|^2+\left|\frac{1}{\sin\theta}\partial_\phi\psi+is\cot\theta\psi\right|^2\right)\dee\nu.
\end{align}
Moreover, we have in view of \eqref{eqn:exprephi},
$$\frac{1}{\sin\theta}\partial_\phi+is\cot\theta=\frac{r^2+a^2}{\Sigma}\left(\mcu-a\sin\theta e_4-a\sin\theta\mu e_3\right),$$
and reinjecting in \eqref{eqn:zzmm} concludes the proof.
\end{proof}
We now recall the standard Sobolev embedding for spin-weighted scalars, see for example Lemma 4.27 and Lemma 4.24 of \cite{mablueandersson}.
\begin{prop}
    Let $\psi$ be a spin $\pm 2$ scalar on $\sph^2$. We have 
\begin{align}\label{eqn:sobolevbase}
    \|\psi\|^2_{L^\infty(\sph^2)}\lesssim\int_{\sph^2}\left(|\psi|^2+|\drond'\drond\psi|^2\right)\dee\nu.
\end{align}
\end{prop}
By replacing the spin-weighted Laplacian $\drond'\drond$ with the Carter operator in equation \eqref{eqn:carter}, one obtains the following reformulation of the standard spherical Sobolev embedding \eqref{eqn:sobolevbase}:
\begin{prop}\label{prop:sobolev}
    Let $\psi$ be a spin $s=\pm 2$ scalar on $\mcm$. For any $(u,\ubar)$ we have
\begin{align}\label{eqn:sobolev}
    \|\psi\|^2_{L^\infty(S(u,\ubar))}{\lesssim}\int_{S(u,\ubar)}\left(|T^{\leq 2}\psi|^2+|\mcq_s\psi|^2\right)\:\dee\nu.
\end{align}
\end{prop}
\subsection{The Teukolsky equation and the Teukolsky-Starobinsky identities}\label{section:eqs}
\subsubsection{Different expressions of the Teukolsky operators}\label{subsection:teukoperators}
Recall from \eqref{eqn:premteuk} the expression of the Teukolsky operator obtained in a frame regular at $\ch$, originally found by Teukolsky \cite{teukolsky} :
$$\begin{aligned}
\teukhat_s:= & -\left[\frac{\left(r^2+a^2\right)^2}{\Delta}-a^2 \sin ^2 \theta\right] T^2 -\frac{4 M a r}{\Delta} T \Phi -\left[\frac{a^2}{\Delta}-\frac{1}{\sin ^2 \theta}\right] \Phi^2  \\
& +\Delta^{-s} \partial_r\left(\Delta^{s+1} \partial_r \right)+\frac{1}{\sin \theta} \partial_\theta\left(\sin \theta \partial_\theta \right)+2 s\left[\frac{a(r-M)}{\Delta}+\frac{i \cos \theta}{\sin ^2 \theta}\right] \Phi  \\
& +2 s\left[\frac{M\left(r^2-a^2\right)}{\Delta}-r-i a \cos \theta\right] T -\left[\frac{s^2 \cos ^2 \theta}{\sin ^2 \theta}-s\right],
\end{aligned}$$
such that for $s=\pm 2$, the spin $s$ Teukolsky equation writes
$$\teukhat_s\Psihat_s=0.$$
The rescaled Teukolsky operator, obtained in the rescaled frame that is regular on $\mch_+$ is $$\teuk_s:=\Delta^s\teukhat_s(\Delta^{-s}\:\cdot\:),$$
such that for $s=\pm 2$, recalling $\psi_s=\Delta^s\Psihat_s$,
$$\teuk_s\psi_s=0.$$
The expression of $\teuk_s$ in B-L coordinates is 
$$\begin{aligned}
\teuk_s= & -\left[\frac{\left(r^2+a^2\right)^2}{\Delta}-a^2 \sin ^2 \theta\right] T^2 -\frac{4 M a r}{\Delta} T \Phi -\left[\frac{a^2}{\Delta}-\frac{1}{\sin ^2 \theta}\right] \Phi^2  \\
& +\Delta^{-s} \partial_r\left(\Delta^{s+1} \partial_r \right)+\frac{1}{\sin \theta} \partial_\theta\left(\sin \theta \partial_\theta \right)+2 s\left[\frac{a(r-M)}{\Delta}+\frac{i \cos \theta}{\sin ^2 \theta}\right] \Phi  \\
& +2 s\left[\frac{M\left(r^2-a^2\right)}{\Delta}-r-i a \cos \theta\right] T -\left[\frac{s^2 \cos ^2 \theta}{\sin ^2 \theta}+s\right] -4 s(r-M) \partial_r.
\end{aligned}$$
Notice that 
\begin{align}\label{eqn:dereqnn}
    \teukhat_s=\teuk_s+2s+4s(r-M)\partial_r=\teuk_s+2s+4s(r-M)(\mu^{-1}e_4-e_3).
\end{align}
To do the energy estimates, it is convenient to write the Teukolsky operators in terms of $e_3$, $e_4$, and $\mcu$.
\begin{prop}\label{prop:}
We have :
    \begin{align}\label{eqn:teuke3e4}
    \teuk_s=&-4(r^2+a^2)e_3e_4+ \mcu^2+\dfrac{1}{\sin\theta}\partial_\theta(\sin\theta\partial_\theta)-4 ias\cos\theta T+2is\dfrac{a^2\sin\theta\cos\theta}{(r^2+a^2)}\mcu\nonumber\\
    &+2 r (e_4-\mu e_3)+4 s[(r-M)e_3- rT]+\dfrac{2ar }{r^2+a^2}\Phi- s-s^2\dfrac{a^4\sin^2\theta\cos^2\theta}{(r^2+a^2)^2},
\end{align}
    \begin{align}\label{teukhat}
    \teukhat_s=&-4(r^2+a^2)e_3e_4+ \mcu^2+\frac{1}{\sin\theta}\partial_\theta(\sin\theta\partial_\theta)-4 ias\cos\theta T+2is\frac{a^2\sin\theta\cos\theta}{(r^2+a^2)}\mcu\nonumber\\
    &+2 r (e_4-\mu e_3)+4 s[(r-M)\mu^{-1}e_4- rT]+\frac{2ar }{r^2+a^2}\Phi+ s-s^2\frac{a^4\sin^2\theta\cos^2\theta}{(r^2+a^2)^2}.
\end{align}
\end{prop}
\begin{proof}
    The computation is done in the easiest way by rewriting
    $$\mcu^2+2is\frac{a^2\cos\theta\sin\theta}{(r^2+a^2)}\mcu-s^2\frac{a^4\sin^2\theta\cos^2\theta}{(r^2+a^2)^2}=\widetilde{\mcu}^2,\quad\widetilde{\mcu}:=\frac{1}{\sin\theta}\Phi+a\sin\theta T+is\cot\theta,$$
    and by using \eqref{eqn:dereqnn} to deduce \eqref{teukhat} from \eqref{eqn:teuke3e4}.
\end{proof}
The Teukolsky operators can also be written \cite[(3.3)]{pricelaw} as
\begin{align}
    &\teuk_s=-\frac{(r^2+a^2)^2-a^2\Delta\sin^2\theta}{\Delta}T^2+\partial_r(\Delta\partial_r)-\frac{4aMr}{\Delta}T\Phi-\frac{a^2}{\Delta}\Phi^2\nonumber\\
    &\quad\quad\quad+\drond_{s-1}\drond'_s-2ias\cos\theta T+4s[(r-M)e_3-rT],\label{eqn:opteuk}\\
    &\teukhat_s=-\frac{(r^2+a^2)^2-a^2\Delta\sin^2\theta}{\Delta}T^2+\partial_r(\Delta\partial_r)-\frac{4aMr}{\Delta}T\Phi-\frac{a^2}{\Delta}\Phi^2\nonumber\\
    &\quad\quad\quad+\drond'_{s+1}\drond_s-2ias\cos\theta T+4s[(r-M)\mu^{-1}e_4-rT].\label{eqn:opteukhat}
\end{align}
Note that these expressions show the crucial fact that $T$, $\Phi$ and the Carter operator commute with the Teukolsky equation :
\begin{align}\label{eqn:commutateurs}
    [\teuk_s,T]=[\teuk_s,\mcq_s]=[\teuk_s,\Phi]=[\teukhat_s,T]=[\teukhat_s,\mcq_s]=[\teukhat_s,\Phi]=0.
\end{align}
To compute the commutator of the Teukolsky operator with $e_3$, we will rewrite $\teuk_s$ in terms of $e_3$, $T$, and the angular operators.
\begin{prop}\label{prop:teukavece3}
    We have 
\begin{align*}
    \mathbf{T}_s=4\Delta e_3^2+&a^2\sin^2\theta T^2-4a\Phi e_3-4(r^2+a^2)Te_3+2aT\Phi \\
    &+\drond_{s-1}\drond'_{s}+4(r-M)(s-1)e_3+(2r(1-2s)-2ais\cos\theta)T.
\end{align*}
\end{prop}
\begin{proof}
    
Using \eqref{eqn:opteuk} and
$$\partial_r=-2e_3+\frac{a}{\Delta}\Phi+\frac{r^2+a^2}{\Delta}T,$$
we get 
\begin{align*}
    \teuk_s&=-\frac{(r^2+a^2)^2-a^2\Delta\sin^2\theta}{\Delta}T^2+\left(-2e_3+\frac{a}{\Delta}\Phi+\frac{r^2+a^2}{\Delta}T\right)(-2\Delta e_3+a\Phi+(r^2+a^2)T)\\
    &\quad\quad\quad-\frac{4aMr}{\Delta}T\Phi-\frac{a^2}{\Delta}\Phi^2+\drond_{s-1}\drond'_s-2ias\cos\theta T+4s[(r-M)e_3-rT]\\
    &=4\Delta e_3^2+a^2\sin^2\theta T^2-4(r-M)e_3-4a\Phi e_3-4(r^2+a^2)Te_3+\frac{a^2}{\Delta}\Phi^2+\frac{2a(r^2+a^2)}{\Delta}T\Phi+2r T\\
    &\quad\quad\quad-\frac{4aMr}{\Delta}T\Phi-\frac{a^2}{\Delta}\Phi^2+\drond_{s-1}\drond'_s-2ias\cos\theta T+4s[(r-M)e_3-rT],
\end{align*}
simplifying the coefficients gives the result.
\end{proof}
\begin{prop}\label{prop:comme3}
    We have 
    $$[\teuk_s,e_3]=4(r-M)e_3^2-4rTe_3+2(s-1)e_3+(1-2s)T.$$
\end{prop}
\begin{proof}
    We simply use the above expression for $\teuk_s$ and the fact that $[\drond_{s-1},e_3]=[\drond'_s,e_3]=0$.
\end{proof}
\subsubsection{Teukolsky-Starobinsky identities}\label{section:tsi}
On top of the Teukolsky equations, we also assume that the Teukolsky-Starobinsky identities (TSI) hold. The Teukolsky-Starobinsky identities are a PDE system relating the 4th order angular and radial derivatives of $\psi_{\pm 2}$. Like the Teukolsky equations, they are obtained from the linearisation of a gravitational perturbation of the Einstein vacuum equations around Kerr spacetime. Their differential form was first derived in \cite{tsi1,tsi2} in frequency space, while their covariant form is derived in \cite{tsicovariant}. Recalling from \eqref{eqn:coordEFin}, \eqref{eqn:coordEFout} the coordinate vector fields 
$$\partial_{r_\mathrm{out}}=\partial_r+\frac{r^2+a^2}{\Delta}T+\frac{a}{\Delta}\Phi=-2\eqhat,\quad \partial_{r_\mathrm{in}}=\partial_r-\frac{r^2+a^2}{\Delta}T-\frac{a}{\Delta}\Phi=-2e_3,$$
the TSI for the spin $\pm 2$ write \cite[Lemma 3.21]{pricelaw}, in Kerr spacetime,
\begin{align}
    (\drond'-ia\sin\theta T)^4\psip-12M\overline{T\psip}&=\Delta^2\partial_{r_\mathrm{out}}^4(\Delta^2\psim),\label{eqn:tsi1}\\
    (\drond+ia\sin\theta T)^4\psim+12M\overline{T\psim}&=\partial_{r_\mathrm{in}}^4(\psip).\label{eqn:tsi2}
\end{align}
As mentioned in \cite{pricelaw}, \eqref{eqn:tsi1} and \eqref{eqn:tsi2} are physical space versions of the frequency space TSI's obtained in \cite{tsi1,tsi2}, and it is also possible to obtain \eqref{eqn:tsi1} and \eqref{eqn:tsi2} from the covariant TSI in \cite{tsicovariant}. We will actually only use \eqref{eqn:tsi2} in this work, close to $\mch_+$, and not \eqref{eqn:tsi1}.
\section{Statement of the main theorem}
\label{section:eq+assumptions}
Recall that we denote by $\psim$, $\psip$ the spin $-2$ and $+2$ scalars that are solutions of the spin $\pm2$ Teukolsky equations :
\begin{align}
    \teuk_{+2}\psip=0,\quad\teuk_{-2}\psim=0.
\end{align}
As before, we denote $\psihatp=\Delta^{-2}\psip$ that satisfies $\teukhat_{+2}\psihatp=0$. In this section, we state our main result on the precise asymptotics of $\psihatp$ at $\ch$. To this end, we first introduce the different subregions of the Kerr interior that we will consider.

\subsection{The different regions of the Kerr black hole interior}\label{section:regions}

Fix $r_\mathfrak{b} \in (r_-,r_+)$ close to $r_-$, $\gamma>0$ small, that will both be chosen later in the energy estimates. More precisely, $\rb$ is chosen in Appendix \ref{appendix:bulkII} and $\gamma$ is chosen in Section \ref{subsection:energyII}. We define the following subregions of the Kerr black hole interior, see Figure \ref{fig:regions2} :
\begin{align*}
    &\mathbf{I}:=\{ \rb\leq r\leq r_+\}\cap\{\ubar\geq 1\},\\
    &\mathbf{II}:=\{r_-\leq r\leq \rb\}\cap\{2r^*\leq\ubar^\gamma\}\cap\{w\leq w_{\rb,\gamma}\},\\
    &\mathbf{III}:=\{r_-\leq r\leq \rb\}\cap\{2r^*\geq\ubar^\gamma\}\cap\{w\leq w_{\rb,\gamma}\},\\
    &\quatre:=\{\ubar\geq\ubar_{\rb,\gamma}\}\cap\{w\geq w_{\rb,\gamma}\},
\end{align*}
where $w_{\rb,\gamma}:=2\rb^*-(2\rb^*)^{1/\gamma}-\rb+r_-$ and $\ubar_{\rb,\gamma}:=(2\rb^*)^{1/\gamma}$ are such that $\{w=w_{\rb,\gamma}\}$ and $\{\ubar=\ubar_{\rb,\gamma}\}$ intersect $\Gamma$ and $\{r=\rb\}$, where the hypersurface $\Gamma$ is defined by $$\Gamma:=\{2r^*=\ubar^\gamma\}.$$ 

  \begin{figure}[h!]
        \centering
        \includegraphics[scale=0.47]{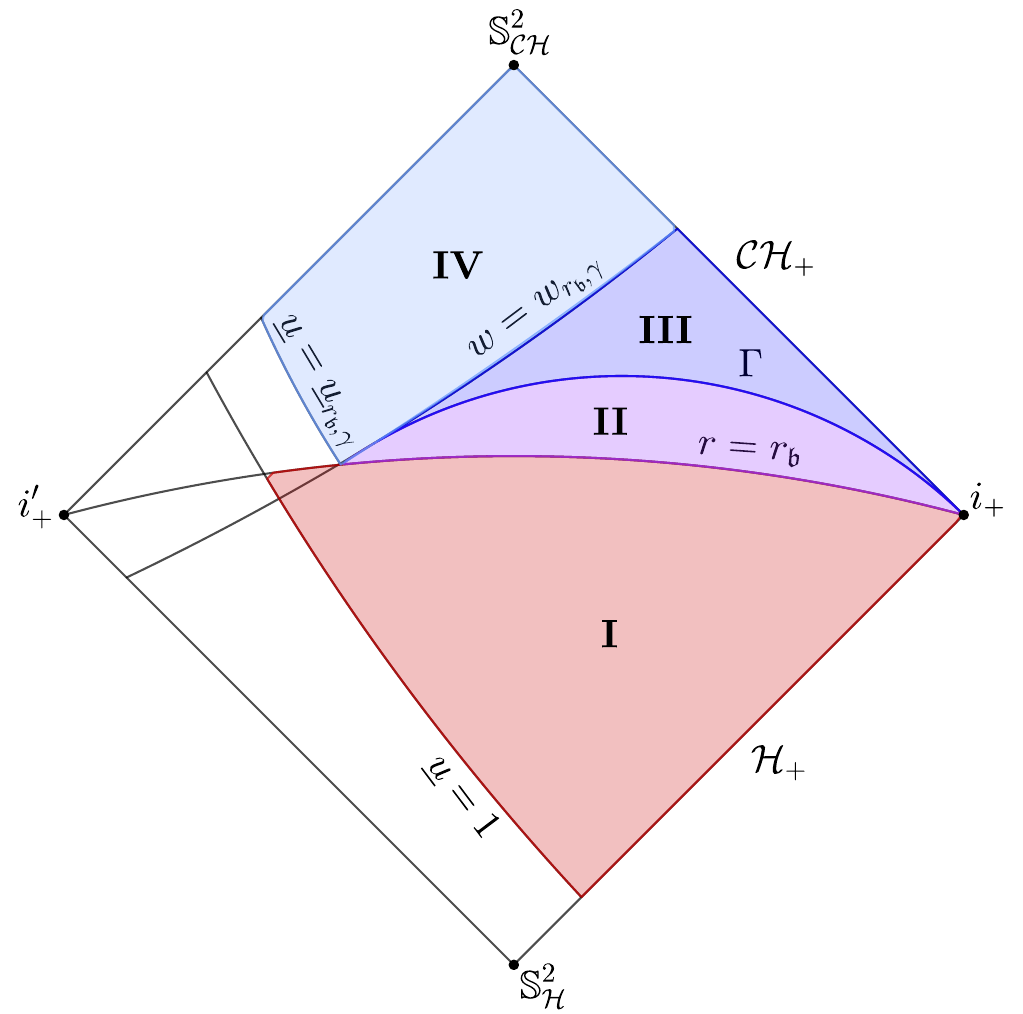}
        \caption{The different subregions of the Kerr black hole interior.}
        \label{fig:regions2}
    \end{figure}
    
Region $\un$ is the redshift region that contains $\mch_+$, region $\trois\cup\quatre$ is the blueshift region very close to $\ch$ where the scalar $\Delta$ decays exponentially\footnote{More precisely, $\Delta$ decays exponentially in $\ubar^\gamma$ in $\trois$, and $\Delta$ decays exponentially both in $u$ and $\ubar$ in $\quatre$.} towards the Cauchy horizon, and region $\deux$ is an intermediate region, where the blueshift effect is already effective.

In region $\un$ we will obtain redshift energy estimates for the Teukolsky equation $\teuk_s\psi_s=0$ for $s=-2$, while in $\deux\cup\trois\cup\quatre$ we will derive effective blueshift estimates for $s=+2$. For the energy estimates, we use the coordinate system $(u,\ubar,\theta,\phi_{+})$ in $\mathbf{I}$, and the coordinate system $(u,\ubar,\theta,\phi_{-})$ in $\deux\cup\trois\cup\quatre$. 

In the next section, we provide the assumptions on $\mch_+$ on which all the analysis is based. The goal of the analysis will be to successively propagate the polynomial bounds on $\mch_+$ to regions $\un$, $\deux$, $\trois$ and $\quatre$. 
\subsection{Main assumptions on the event horizon}\label{section:assumptions}
We first define our initial spacelike hypersurface $\Sigma_{_0}$ as the union of three spacelike hypersurfaces :
$$\Sigma_0:=\Sigma_{\tau_0}\cup\Sigma_{int}\cup\Sigma_{\tau'_0},$$
similarly as in \cite{scalarMZ}. More precisely, we 
define $\Sigma_{\tau_0}$ as the hypersurface defined in \cite{pricelaw}, i.e. a constant $\tau$ hypersurface, where $(\tau,r,\theta,\phi_+)$ is a hyperboloidal coordinate system on the right part of the exterior of Kerr spacetime. Similarly, $\Sigma_{\tau'_0}$ is a constant $\tau'$ hypersurface, where $(\tau',r,\theta,\phi_-)$ is a hyperboloidal coordinate system on the left part of the exterior of Kerr spacetime. We chose $\Sigma_{int}$ as any spacelike hypersurface inside the Kerr interior that joins $\Sigma_{\tau_0}$ and $\Sigma_{\tau'_0}$  such that the union of the three hypersurfaces is spacelike, see Figure \ref{fig:sigma_0union}. The Cauchy problem for the Teukolsky equation with initial data on $\Sigma_0$ is well-posed on the future maximal Cauchy development $\mcd^+(\Sigma_0)$ of $\Sigma_0$. We will prove the precise asymptotics of the Teukosky field on $\mcm\cap\mcd^+(\Sigma_0)$, that is the part of $\mcm$ (thus in the grey shaded area) that is located above $\Sigma_{int}$ in Figure \ref{fig:sigma_0union}. Without loss of generality we can assume that at the intersection of $\Sigma_0$ and $\mch_+$, we have $\ubar\leq 1$, and symmetrically that $u\leq 1$ at $\Sigma_0\cap\mch'_+$.
\begin{figure}[h!]
    \centering
    \includegraphics[scale=0.45]{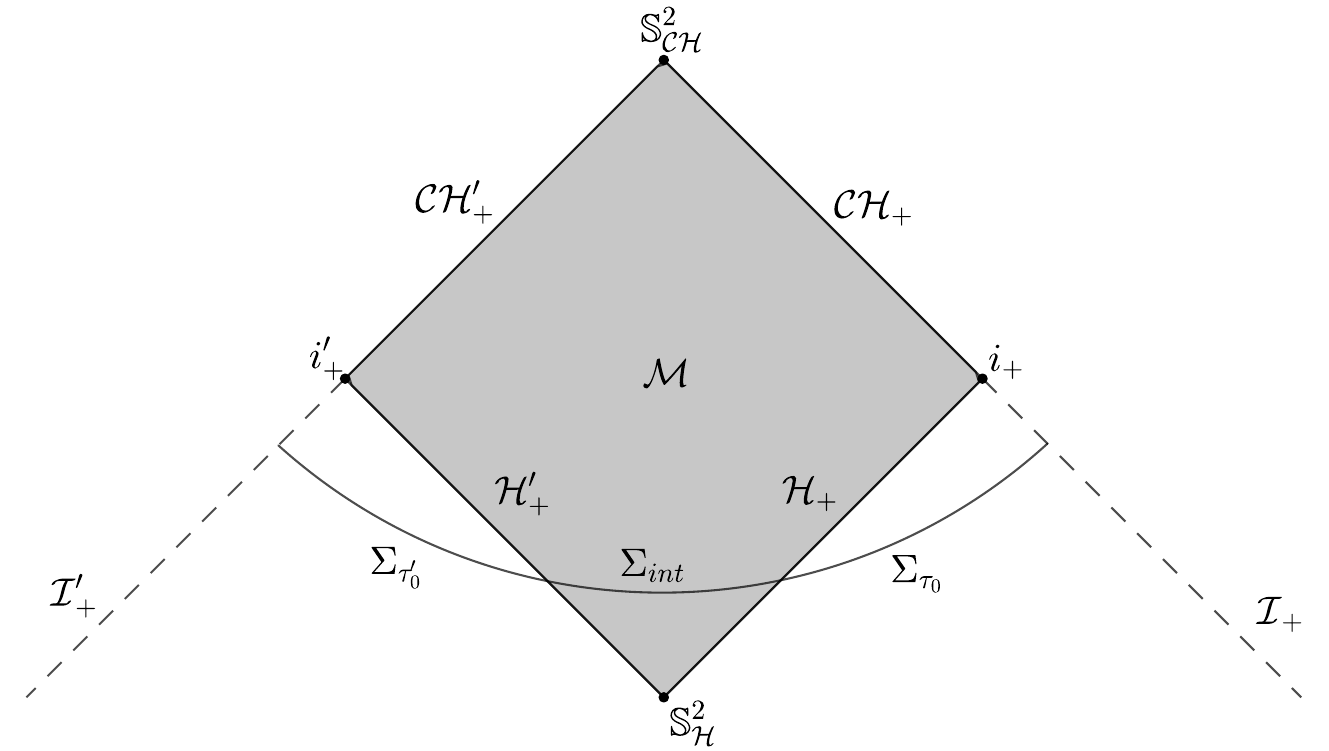}
    \caption{The initial hypersurface $\Sigma_0=\Sigma_{\tau_0}\cup\Sigma_{int}\cup\Sigma_{\tau'_0}$.}
    \label{fig:sigma_0union}
\end{figure}

The conclusion of the main theorem will be applicable to solutions of the Teukolsky equations that arise from compactly supported initial data on $\Sigma_{0}$, but we will actually only rely on the Price's law results of \cite{pricelaw}, that we write down now.

We denote by $\nablabar$ the set of tangential derivatives on the event horizon. More precisely, we use
$$\overline{\nabla}:=\{T,\Phi, \drond, \drond'\},$$
and for $k=(k_1,k_2,k_3,k_4)\in\mathbb{N}^4$, we define $|k|=k_1+k_2+k_3+k_4$ and $\nablabar^k=T^{k_1}\Phi^{k_2}\drond^{k_3}(\drond')^{k_4}$.

Let $N_j^\pm, N_k^\pm, (N_j^+)', (N_k^+)'\geq 1$ be sufficiently large integers that we will choose later, see Remark \ref{rem:apresthm} for their precise values. 

We consider $\psi_{\pm 2}$ such that there is $(Q_{m,2})_{|m|\leq 2}\in\mathbb{C}^5$, and $\delta>0$ such that for $j\leq N_j^-$, $|k|\leq N_k^-$, 
    \begin{align}
        \left|T^j\nablabar^k\errm\right|\lesssim\ubar^{-7-j-\delta}\text{   on   }\mch_+\cap\{\ubar\geq 1\},\label{eqn:hypm}
    \end{align}
    and such that for $j\leq N_j^+$, $|k|\leq N_k^+$, and $l\leq 3$, 
    \begin{align}   \left|T^j\nablabar^ke_3^l\errp\right|\lesssim\ubar^{-7-j-\delta}\text{   on   }\mch_+\cap\{\ubar\geq 1\}\label{eqn:hypp},
    \end{align}
where the \app error' quantities $\mathrm{Err}[\psi_{\pm 2}]$ are defined by
\begin{align*}
    \errm&:=\psim-\ansatzm,\\
    \errp&:=\psip-\ansatzp.
\end{align*}
We also assume the less precise assumption on $\mch'_+$ : for $j\leq (N_j^+)'$, $|k|\leq (N_k^+)'$, and $l\leq 1$,
\begin{align}\label{eqn:hypp'}
    |T^j\nablabar^k\eqhat^{l}\psihatp|\lesssim u^{-7-j}\text{   on   }\mch'_+\cap\{u\geq 1\}.
\end{align}
\begin{rem}
Assumptions \eqref{eqn:hypm}, \eqref{eqn:hypp}, \eqref{eqn:hypp'} with $|k|=l=0$ correspond to the so-called Price's law, and were recently shown to hold true in subextremal Kerr by Ma and Zhang \cite{pricelaw}. More precisely, they have shown that this holds for initially smooth and compactly supported solutions on $\Sigma_0$ of the Teukolsky equations $\teuk_{\pm 2}\psi_{\pm 2}=0$, where :
\begin{itemize}
    \item  The constants $Q_{m,2}$ are defined as $2^7\mathbb{Q}_{m,2}/5$ where the constants $\mathbb{Q}_{m,2}$ are  defined in \cite[Lemma 5.7]{pricelaw}, depend on the values of $\psi_{\pm 2}$ on the initial hypersurface $\Sigma_{0}$, and are generically non-zero.
    \item The functions $A_m(r)$ are defined as $(r^2+a^2)^2\mathfrak{f}_{+2,m}(r)$ where $\mathfrak{f}_{+2,m}(r)$ is precisely defined in \cite[Eq. (5.82c)]{pricelaw}. The explicit computation of $A_m(r)$ gives
    \begin{align*}
    A_m(r)=\frac{1}{3}\Big[3\Delta^2&+(r-M)(4(a^2-M^2)+6\Delta)iam-(2\Delta+6(a^2-M^2)+4(r-M)^2)a^2m^2\nonumber\\
    &-4(r-M)ia^3m^3+2a^4m^4\Big],
\end{align*} 
see Appendix \ref{appendix:am(r)}.
\end{itemize}
    Statements \eqref{eqn:hypm}, \eqref{eqn:hypp}, \eqref{eqn:hypp'} with non-zero $|k|,l$ can be deduced\footnote{The statements about the tangential derivatives on the event horizon can be obtained easily using the fact that $\Phi$, $T$ and the Carter operator commute with the Teukolsky equation and directly applying the main theorem of \cite{pricelaw}. Statement \eqref{eqn:hypp} for the $e_3^{\leq 3}$ derivatives can be deduced as follows. Define $\errp$ as in \eqref{eqn:deferrp} and let $r_0>r_+$. Differentiating \cite[Eq. (5.87)]{pricelaw} by $e_3^n$ for $n=1,2,3$, restricting on $\{r=r_0\}$, and using the bounds \cite[Eq. (4.101), (4.99)]{pricelaw} as well as $e_3\sim \mu(r_0)\partial_\rho+O(1)T$ on $\{r=r_0\}$ gives $|e_3^{n}\errp|\lesssim\ubar^{-7-\delta}$ on $\{r=r_0\}$. Then the TSI \eqref{eqn:ineq3} gives $|e_3^4\errp|\lesssim\ubar^{-7-\delta}$ on $\{r_0\leq r\leq r_+\}$, and all is left to do is integrate this bound from $r=r_0$ to $r=r_+$.} directly from the results in \cite{pricelaw}, for solutions of the Teukolsky equations arising from smooth and compactly supported initial data on $\Sigma_0$.
\end{rem}

\subsection{Precise version of the main theorem}
We state the main result of this paper.
\begin{thm}[Main theorem, precise version]\label{thm:main}
    Let $\psi_{\pm 2}$ be solutions of the Teukolsky equations $\teuk_{\pm 2}\psi_{\pm2}=0$ on the interior of Kerr spacetime with $0<|a|<M$. Assume that the Teukolsky-Starobinsky identity \eqref{eqn:tsi2} holds in $\un$, as well as the assumptions \eqref{eqn:hypm}, \eqref{eqn:hypp}, \eqref{eqn:hypp'} on the event horizon $\mch_+\cup\mch_+'$. Then, denoting $\psihatp=\Delta^{-2}\psip$, in region $\trois\cup\quatre$ close to $\ch$ we have the following asymptotic behaviour :
\begin{align}\label{eqn:holdssss}
    \psihatp(u,\ubar,\theta,\phi_{-})=\frac{\Delta^{-2}(u,\ubar)}{\ubar^7}\sum_{|m|\leq 2}A_m(r_-)Q_{m,2}e^{2imr_{mod}(u,\ubar)}Y_{m,2}^{+2}(\cos\theta)e^{im\phi_{-}}+O(\Delta^{-2}\ubar^{-7-\delta}).
\end{align} 
\end{thm}
\begin{rem}\label{rem:apresthm}
    A few remarks are in order.
    \begin{enumerate}
        \item As discussed in Section \ref{section:teukolskyeq}, $\psihatp$ is the Teukolsky scalar obtained from the linearisation of a gravitational perturbation obtained in the Newman-Penrose tetrad that is regular at $\ch$. Thus Theorem \ref{thm:main} should be interpreted as a linear curvature instability statement for the Kerr Cauchy horizon, as $\Delta^{-2}/\ubar^7$ blows up exponentially in $\ubar$ in $\trois\cup\quatre$ towards $\ch$. Moreover, since the function $r_{mod}$ blows up on $\ch$ as $r_{mod}\sim a\log(r-r_-)$, \eqref{eqn:holdssss} implies large oscillations of $\psihatp$ close to $\ch$, as predicted by Ori \cite{ori}.
        \item Notice that the expression of $A_m(r)$ gives
        \begin{align}\label{eq:am(r-)}
            A_m(r_-)=2am\left(am+2i\sqrt{M^2-a^2}\right)\left(M^2+a^2(m^2-1)\right).
        \end{align}
        Thus for any subextremal Kerr black hole parameters $(a,M)$, for $m=\pm 1$ and $m=\pm 2$, we have $A_m(r_-)\neq 0$. Note that the identity $A_0(r_-)=0$ confirms the heuristics arguments of Ori \cite{ori} according to which the $m=0$ azimuthal mode of $\Psihatp$ decays better than $\Delta^{-2}\ubar^{-7}$. 
        \item We also prove, see Proposition \ref{prop:ansatzpII}, in  intermediate region $\deux$,
\begin{align}\label{eqn:holdssssII}
    \psihatp(u,\ubar,\theta,\phi_{-})=\frac{\Delta^{-2}(u,\ubar)}{\ubar^7}\sum_{|m|\leq 2}A_m(r)Q_{m,2}e^{2imr_{mod}(u,\ubar)}Y_{m,2}^{+2}(\cos\theta)e^{im\phi_{-}}+O(\Delta^{-2}\ubar^{-7-\delta}).
\end{align} 
        \item In view of the asymptotic behaviors \eqref{eqn:holdssss} and  \eqref{eqn:holdssssII}, there exists $C>1$ large enough such that the following uniform lower bound
$$\|\psihatp\|_{L^2(S(u,\ubar))}\gtrsim \frac{\Delta^{-2}(u,\ubar)}{\ubar^7},\text{   in   }\{\ubar\geq C\}\cap \{r_-< r\leq \rb\},$$
        holds for generic initial data, in accordance with the results of \cite{sbierski}.
        \item Inspecting the proof, we find that the minimal values of $N_j^\pm$, $N_k^\pm$ and $(N_j^+)'$, $(N_k^+)'$ for which we prove \eqref{eqn:holdssss} are $N_j^+=5$, $N_k^+=5$, $N_j^-=15$, $N_k^-=14$, and $(N_j^+)'=2$, $(N_k^+)'=3$. We did not try to optimize this loss of derivatives. We mainly lose derivatives when we apply the Sobolev embedding on the spheres (which loses each time two $T$ and angular derivatives) and when we integrate the Teukolsky-Starobinsky identity from $\mch_+$. Assuming that higher order 
        derivatives decay on the event horizon leads to an asymptotic \eqref{eqn:holdssss} that holds for higher order $T$, $\carterp$, and $\Phi$ derivatives. 
        \item By assuming the equivalent statement of assumptions \eqref{eqn:hypp}, \eqref{eqn:hypm} for $\psihat_{\pm 2}$ on $\mch'_+$, we can get an equivalent statement in the symmetric region $\{u\geq 1\}$ at any point in the analysis, by replacing $\ubar$ with $u$, $\psi_s$ with $\psihat_{-s}$, $e_3$ with $\eqhat$ and $e_4$ with $\ethat$. As in \cite{scalarMZ} for the scalar wave equation, this will not be useful until we try to get the asymptotics on the upper part of $\ch$, i.e. in region $\quatre$. There, we will need the boundedness of the energy density $\ener[e_3^{\leq 1}\psip]$ on $\{u\geq 1\}\cap\{r=\rb\}$. This is why we also assume \eqref{eqn:hypp'}, which ensures this energy density bound, and which holds for physical initial data compactly supported near the right part of the Kerr exterior. We note that to prove the final result in $\deux\cup\trois$, that contains the lower part of $\ch$, the analysis does not require any assumptions on $\mathcal{H}'_+$.
        \item We also prove a result on the precise asymptotics of $\psi_{- 2}$ in redshift region $\un$, see Proposition \ref{prop:ansatzI}, and Theorem \ref{thm:psim} for a more general result on the spin $-2$ Teukolsky equation in $\un$. See Proposition \ref{prop:ansatzIp} for the precise asymptotics of $\Psip$ in $\un$.

    \end{enumerate}
\end{rem}
The rest of the paper is devoted to the proof of Theorem \ref{thm:main}. In Section \ref{section:redshift} we prove the precise asymptotics of $\psi_{\pm 2}$ in region $\un$, see Propositions \ref{prop:ansatzI} and \ref{prop:ansatzIp}. In Section \ref{section:blueshiftII}, we prove the precise asymptotics of $\Psip$ in region $\deux$, see Proposition \ref{prop:ansatzpII}. Finally, in Section \ref{section:blueshiftIII}, we prove the precise asymptotic behavior \eqref{eqn:holdssss} of $\Psihatp$ in regions $\trois$ and $\quatre$, see Theorems \ref{thm:psip} and \ref{thm:psipIV}, hence concluding the proof of Theorem \ref{thm:main}.

\section{Precise asymptotics in redshift region $\mathbf{I}$ near $\mch_+$}\label{section:redshift}
We begin the proof of Theorem \ref{thm:main}, with the description of the precise asymptotics in region $\un$. In Section \ref{subsection:energyI}, we show a redshift energy estimate, that we will eventually apply to $\errm$ in Section \ref{subsection:applyI} to propagate the ansatz for $\psim$ from $\mch_+$ to $\un$. Finally, in Section \ref{subsection:usetsi}, using the Teukolsky-Starobinsky identity \eqref{eqn:tsi2} we propagate the ansatz for $\psip$ from $\mch_+$ to $\un$.

\subsection{Energy method for the spin $-2$ Teukolsky equation in $\un$}
\label{subsection:energyI}
We begin this section with the following definition.
\begin{defi}
    Let $V$ and $c$ be real numbers. We define the following spin-weighted operator :
    \begin{align}\label{eqn:teukmodifI}
    \teuk_s^{(c,V)}:=&-4(r^2+a^2)e_3e_4+ \mcu^2+\dfrac{1}{\sin\theta}\partial_\theta(\sin\theta\partial_\theta)-4 ias\cos\theta T+2is\dfrac{a^2\sin\theta\cos\theta}{(r^2+a^2)}\mcu\nonumber\\
    &+2 r (e_4-\mu e_3)+4cs[(r-M)e_3- rT]+\dfrac{2ar }{r^2+a^2}\Phi- s-s^2\dfrac{a^4\sin^2\theta\cos^2\theta}{(r^2+a^2)^2}+V.
\end{align}
\end{defi}
\begin{rem}We remark the following :
\begin{itemize}
    \item We introduce the modified Teukolsky operator $\teuk_s^{(c,V)}$ in order to do energy estimates for a more general class of operators, which will be useful after commuting the Teukolsky operator with $e_3$.
    \item We will use the estimates of this section for a finite number of explicit constants $(c,V)$ thus we still write the bounds that depend on $(c,V)$ with notations $O$, $\lesssim$.
    \item Using \eqref{eqn:teuke3e4} we get $\teuk_s=\teuk_s^{(1,0)}$. More generally, 
    \begin{align}\label{eqn:recupc}
        \teuk_s^{(c,V)}=\teuk_s+4(c-1)s[(r-M)e_3-rT]+V,
    \end{align}
    which proves, using \eqref{eqn:commutateurs},
    $$[\teuk_s^{(c,V)},T]=[\teuk_s^{(c,V)},\carterm]=0.$$
\end{itemize}

\end{rem}
In all this section, we denote $\psi$ a spin $-2$ scalar such that there are constants 
$$V\in\mathbb{R},\quad c>0,\quad\beta>1,$$ 
such that for $0\leq j\leq N_j^-$, $0\leq |k|\leq N_k^-$, and $0\leq 2 k_1+k_2\leq N_k^-$,
\begin{align}
    \bullet\:&T^j\nablabar^k\psi=O(\ubar^{-\beta-j})\text{   on   }\mch_+\cap\{\ubar\geq 1\},\label{eqn:hyppsiborne}\\
    \bullet\:&\teuk_{-2}^{(c,V)}T^j\mcq_{-2}^{k_1}\Phi^{k_2}\psi=O(\ubar^{-\beta-j})\text{   in   }\un.\label{eqn:hyppsiteuk}
\end{align}
{\begin{rem}Note that the motivation for studying the inhomogeneous Teukolsky equation \eqref{eqn:hyppsiteuk} is that $\errm$ satisfies \eqref{eqn:hyppsiteuk} with $\beta=8$, see \eqref{eq:premiereacit} and \eqref{eq:deuxiemeacit}.
\end{rem}
}
The goal of this section is to propagate the upper bound \eqref{eqn:hyppsiborne} for $\psi$ on the event horizon to region $\un$ using a redshift energy estimate, see Proposition \ref{prop:propagI}. Recall the definition \eqref{eqn:defener} of the energy density $\ener[\psi]=|e_3\psi|^2+|e_4\psi|^2+|\mathcal{U}\psi|^2+|\partial_\theta\psi|^2$.
\begin{prop}\label{prop:ener}
    Assume that $\psi$ is a spin $-2$ scalar that satisfies \eqref{eqn:hyppsiborne}, \eqref{eqn:hyppsiteuk} with $c>0$, $\beta>1$. Then for $0\leq j\leq N_j^-$, $0\leq 2 k_1+k_2\leq N_k^--1$, and for $\wbar_1\geq 1$,
    $$\iint_{\left\{\underline{w}=\underline{w}_1\right\} \cap \mathbf{I}} \mathbf{e}[T^j\mcq_{-2}^{k_1}\Phi^{k_2}\psi](-\mu) \dee\nu\:\dee u\lesssim\wbar_1^{-2\beta-2j}.$$
\end{prop}
\begin{rem}
    The choice of the negative $-2$ spin is the right one in the redshift region $\mathbf{I}$ to be able to obtain a positive bulk term in the energy estimate. Actually, we will see that the sign that matters is the one of $s(r-M)$ where $s$ is the spin, thus we already anticipate that we will only be able to control solutions with positive $+2$ spin in the blueshift region $\deux\cup\trois\cup\quatre$.
\end{rem}
\begin{proof}[Proof of Proposition \ref{prop:ener}]
In what follows, we denote $s=-2$. The computations are done for a general spin $s$, but the bulk term will only be positive for $s< 0$. In view of the assumptions \eqref{eqn:hyppsiborne}, \eqref{eqn:hyppsiteuk}, it suffices to treat the case $j=k_1=k_2=0$. Recall that in $\mathbf{I}$, in coordinates $(u,\ubar,\theta,\phi_{+})$, we have $\partial_u=-\mu e_3$ and $\partial_\ubar=e_4-a/(r^2+a^2)\Phi$. Thus multiplying \eqref{eqn:teukmodifI} by $\mu$ we get the following expressions :
\begin{align}
    \mu\teuk_s^{(c,V)}&=4(r^2+a^2)\partial_ue_4+\mu \mathcal{U}^2+\dfrac{\mu}{\sin\theta}\partial_\theta(\sin\theta\partial_\theta)-4\mu ias\cos\theta T+2is\mu\dfrac{a^2\sin\theta\cos\theta}{(r^2+a^2)}\mathcal{U}\nonumber\\
    &+2 r\mu (e_4+\partial_u)-4cs[(r-M)\partial_u+\mu rT]+\dfrac{2ar\mu }{r^2+a^2}\Phi-\mu s-s^2\mu\frac{a^4\sin^2\theta\cos^2\theta}{(r^2+a^2)^2}+\mu V\label{eqn:muteuk1}\\
    &=4(r^2+a^2)e_4\partial_u+\mu \mathcal{U}^2+\dfrac{\mu}{\sin\theta}\partial_\theta(\sin\theta\partial_\theta)-4\mu ias\cos\theta T+2is\mu\dfrac{a^2\sin\theta\cos\theta}{(r^2+a^2)}\mathcal{U}\nonumber\\
    &\quad+2 r\mu (e_4+\partial_u)-4cs[(r-M)\partial_u+\mu rT]-\dfrac{2ar\mu }{r^2+a^2}\Phi-\mu s-s^2\mu\frac{a^4\sin^2\theta\cos^2\theta}{(r^2+a^2)^2}+\mu V\label{eqn:muteuk2},
\end{align}
where we used \eqref{eqn:coordDNLin} to get 
$$4(r^2+a^2)[\partial_u,e_4]=-\dfrac{4ar\mu }{r^2+a^2}\Phi.$$
As in \cite{scalarMZ} we now multiply \eqref{eqn:hyppsiteuk} by $\mu$ and by the complex conjugate of $$X(\psi):=f(r)\partial_u\psi+g(r)e_4\psi,$$
where we choose 
$$f(r)=(-\mu)^{-1}(r^2+a^2)^p,\quad g(r)=(r^2+a^2)^p$$
with a real number $p=p(a,M,c,V)\gg1$ large enough that will be chosen in Appendix \ref{appendix:bulkI}. We then integrate over $S(u,\ubar)$ against $\dee\nu$, and take the real part, to get
\begin{align}\label{eqn:teukmult}
    \int_{S(u,\ubar)}\Real(f(r)\overline{\partial_u\psi}\mu\teuk_s^{(c,V)}{\psi})+\Real(g(r)\overline{e_4\psi}\mu\teuk_s^{(c,V)}\psi)\dee\nu=\int_{S(u,\ubar)}\mu\Real(\overline{X(\psi)}O(\ubar^{-\beta}))\dee\nu.
\end{align}
By Lemma \ref{lem:integS} in Appendix \ref{appendix:computationsI}, this implies
\begin{align}\label{eqn:conservationener2}
    \partial_\ubar\left(\int_{S(u,\ubar)}\mathbf{F}_\ubar[\psi]\dee\nu\right)+\partial_u\left(\int_{S(u,\ubar)}\mathbf{F}_u[\psi]\dee\nu\right)+\int_{S(u,\ubar)}\mathbf{B}[\psi]\dee\nu=\int_{S(u,\ubar)}\mu\Real(\overline{X(\psi)}O(\ubar^{-\beta}))\dee\nu,
\end{align}
where $\mathbf{F}_\ubar[\psi]$, $\mathbf{F}_u[\psi]$, and $\mathbf{B}[\psi]$ are defined in Lemma \ref{lem:integS}. 
\begin{figure}[h!]
    \centering
    \includegraphics[scale=0.5]{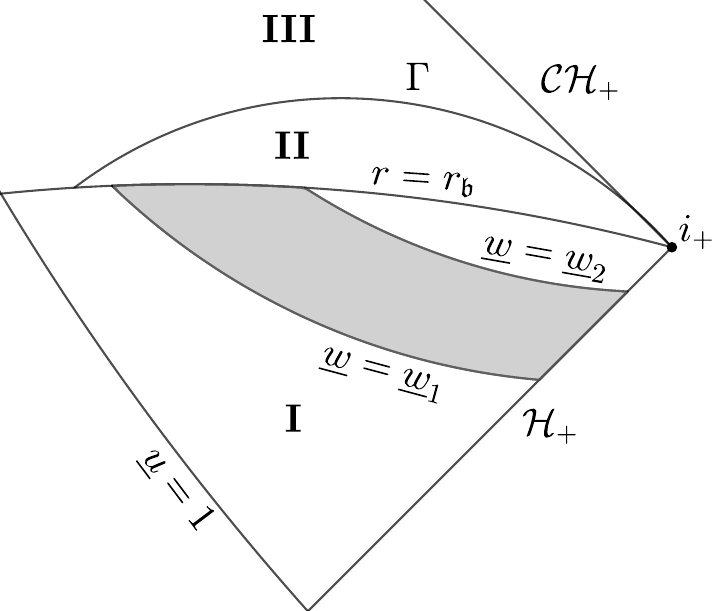}
    \caption{The redshift energy method. The grey region is $\un\cap\{\wbar_1\leq\wbar\leq\wbar_2\}$, where we integrate \eqref{eqn:conservationener2}.}
    \label{fig:I}
\end{figure}
We now follow \cite{scalarMZ} and integrate \eqref{eqn:conservationener2} in $\mathbf{I}\cap\{\wbar_1\leq\wbar\leq\wbar_2\}$ (see Figure \ref{fig:I}) with respect to $\dee u\dee\ubar$ to get 
\begin{align}\label{eqn:combi}
& \iint_{\left\{\underline{w}=\underline{w}_2\right\} \cap \mathbf{I}} \mathcal{T}_{\underline{w}}[\psi] \dee\nu \mathrm{d} u+\iint_{\left\{r=r_{\mathfrak{b}}\right\} \cap\left\{\underline{w}_1 \leq \underline{w} \leq \underline{w}_2\right\}} \mathcal{T}_r[\psi] \dee\nu \mathrm{d} \underline{u}+\iiint_{\left\{\underline{w}_1 \leq \underline{w} \leq \underline{w}_2\right\} \cap \mathbf{I}} \mathbf{B}[\psi] \dee\nu \mathrm{d} u \mathrm{d} \underline{u} \nonumber\\
= & \iint_{\left\{\underline{w}=\underline{w}_1\right\} \cap \mathbf{I}}\hspace{-0.2cm}\mathcal{T}_{\underline{w}}[\psi] \dee\nu \mathrm{d} u+\iint_{\mathcal{H}_{+} \cap\left\{\underline{w}_1 \leq \underline{w} \leq \underline{w}_2\right\}} \hspace{-0.3cm}\mathcal{T}_u[\psi] \dee\nu \mathrm{d} \underline{u}+\iiint_{\left\{\underline{w}_1 \leq \underline{w} \leq \underline{w}_2\right\} \cap \mathbf{I}} \hspace{-1cm}\mu\Real(\overline{X(\psi)}O(\ubar^{-\beta})) \dee\nu \mathrm{d} u \mathrm{d} \underline{u},
\end{align}
where $$\mathcal{T}_{\underline{w}}[\psi]=\left(1-\mu/2\right)\mathbf{F}_{\underline{u}}[\psi]-\dfrac{1}{2} \mu \mathbf{F}_u[\psi],\quad\mathcal{T}_r[\psi]=-\dfrac{1}{2} \mu\left(\mathbf{F}_{\underline{u}}[\psi]+\mathbf{F}_u[\psi]\right), \text{  and  }\mathcal{T}_u[\psi]=\mathbf{F}_u[\psi].$$
We now estimate the different quantities using the choice of $f$ and $g$.

\noindent\textbf{Control of the bulk terms.}
First, using Lemma \ref{lem:bulkpos} which is a manifestation of the redshift effect, we get that for $p=p(a,M,c,V)$ chosen large enough and for $s=-2$, we have in $\un$
    \begin{align}\label{eq:bulkposI2}
    \int_{S(u,\ubar)}\mathbf{B}[\psi]\dee\nu\gtrsim (-\mu)\int_{S(u,\ubar)}\ener[\psi]\dee\nu.
\end{align}
To control the other bulk term on the RHS of \eqref{eqn:combi}, we write, for $\varepsilon>0$,
\begin{align}
    \Bigg|\iiint&_{\left\{\underline{w}_1 \leq \underline{w} \leq \underline{w}_2\right\} \cap \mathbf{I}} \mu\Real(\overline{X(\psi)}O(\ubar^{-\beta})) \dee\nu \mathrm{d} u \mathrm{d} \underline{u}\Bigg|\lesssim \nonumber\\
    &\varepsilon\iiint_{\left\{\underline{w}_1 \leq \underline{w} \leq \underline{w}_2\right\} \cap \mathbf{I}} |X(\psi)|^2(-\mu)\dee\nu \mathrm{d} u \mathrm{d} \underline{u}+\varepsilon^{-1}\iiint_{\left\{\underline{w}_1 \leq \underline{w} \leq \underline{w}_2\right\} \cap \mathbf{I}} \ubar^{-2\beta} (-\mu)\dee\nu \mathrm{d} u \mathrm{d} \underline{u}.\label{eqn:abso}
\end{align}
Moreover, changing variables\footnote{Note that on $\wbar=cst$, we have $\mu\dee u=(2-\mu)\dee r$.\label{footnote}} from $u$ to $r$, and the fact that $r$ is bounded, we get
$$\iiint_{\left\{\underline{w}_1 \leq \underline{w} \leq \underline{w}_2\right\} \cap \mathbf{I}} \ubar^{-2\beta} (-\mu)\dee\nu \mathrm{d} u \mathrm{d} \underline{u}\lesssim\iint_{\left\{\underline{w}_1 \leq \underline{w} \leq \underline{w}_2\right\} \cap \mathbf{I}} \ubar^{-2\beta} \mathrm{d} r \mathrm{d} \underline{u}\lesssim\int_{\wbar_1}^{\wbar_2}\wbar^{-2\beta}\dee\wbar.$$
We also have $|X(§\psi)|^2\lesssim |e_3\psi|^2+|e_4\psi|^2\leq\ener[\psi]$. Thus, choosing $\varepsilon>0$ small enough such that the first term on the RHS of \eqref{eqn:abso} is absorbed in the LHS of \eqref{eqn:combi}, we get 
\begin{align}\label{eqn:combi2}
& \iint_{\left\{\underline{w}=\underline{w}_2\right\} \cap \mathbf{I}} \mathcal{T}_{\underline{w}}[\psi] \dee\nu \mathrm{d} u+\iint_{\left\{r=r_{\mathfrak{b}}\right\} \cap\left\{\underline{w}_1 \leq \underline{w} \leq \underline{w}_2\right\}} \mathcal{T}_r[\psi] \dee\nu \mathrm{d} \underline{u}+\iiint_{\left\{\underline{w}_1 \leq \underline{w} \leq \underline{w}_2\right\} \cap \mathbf{I}} \mathbf{e}[\psi] (-\mu)\dee\nu \mathrm{d} u \mathrm{d} \underline{u} \nonumber\\
\lesssim & \iint_{\left\{\underline{w}=\underline{w}_1\right\} \cap \mathbf{I}}\hspace{-0.2cm}\mathcal{T}_{\underline{w}}[\psi] \dee\nu \mathrm{d} u+\iint_{\mathcal{H}_{+} \cap\left\{\underline{w}_1 \leq \underline{w} \leq \underline{w}_2\right\}} \hspace{-0.3cm}\mathcal{T}_u[\psi] \dee\nu \mathrm{d} \underline{u}+\int_{\wbar_1}^{\wbar_2}\wbar^{-2\beta}\dee\wbar.
\end{align}
\noindent\textbf{Control of the boundary terms.} We first deal with the boundary terms on $\wbar=cst$ : 
\begin{align}
     \mathcal{T}_{\underline{w}}[\psi]&=\left(1-\mu/2\right)\mathbf{F}_{\underline{u}}[\psi]-\dfrac{1}{2} \mu \mathbf{F}_u[\psi]\nonumber\\
     &=2(r^2+a^2)f(r)|\partial_u\psi|^2-\dfrac{1}{2}\mu g(r)(|\partial_\theta\psi|^2+|\mcu\psi|^2)+a\sin\theta\mu \mathfrak{R}(\overline{X(\psi)}\mcu\psi)\nonumber\\
     &\quad\quad-\frac{\mu}{2}\Bigg[2(r^2+a^2)f(r)|\partial_u\psi|^2+2(r^2+a^2)g(r)|e_4\psi|^2-\dfrac{1}{2}\mu (f(r)+g(r))(|\partial_\theta\psi|^2+|\mcu\psi|^2)\Bigg]\nonumber\\
     &\sim (-\mu)\ener[\psi],\label{eqn:estimebordwbar}
\end{align}
where we absorbed the term $a\sin\theta\mu \mathfrak{R}(X(\psi)\mcu\psi)$ using 
$$|a\sin\theta \mu f(r) \mathfrak{R}(\overline{\partial_u\psi}\mcu\psi)|\leq a^2f(r)|\partial_u\psi|^2+\frac{1}{4}\mu^2f(r)|\mcu\psi|^2$$
and
$$|a\sin\theta\mu g(r) \mathfrak{R}(\overline{e_4\psi}\mcu\psi)|\leq -\frac{3}{4}\mu a^2g(r)|e_4\psi|^2-\frac{1}{3}\mu g(r)|\mcu\psi|^2.$$
Then, for the term on $\{r=\rb\}$, we have 
\begin{align}\label{eqn:estimbordrb}
    \mathcal{T}_r[\psi]=-\dfrac{1}{2} \mu\left(\mathbf{F}_{\underline{u}}[\psi]+\mathbf{F}_u[\psi]\right)\sim \ener[\psi],\text{   on   }\{r=\rb\}.
\end{align}
To control the term on the event horizon, we use\footnote{Notice that $e_4=-\mu e_3+O(1)T+O(1)\Phi$.}  {$|\mathcal{T}_u[\psi]|\lesssim |e_4\psi|^2+|\overline{\nabla}\psi|^2\lesssim |\overline{\nabla}\psi|^2$ on $\mch_+$. }Thus, using \eqref{eqn:hyppsiborne}  we get on $\mch_+\cap\{\ubar\geq 1\}$,
$$\int_{S(-\infty,\ubar)}\mathcal{T}_u[\psi]\dee\nu\lesssim\int_{S(-\infty,\ubar)}|T\psi|^2+|\Phi\psi|^2\dee\nu\lesssim\ubar^{-2\beta}.$$
This concludes the estimates of the boundary terms.  Together with \eqref{eqn:combi2}, this yields 
\begin{align}
 \iint&_{\left\{\underline{w}=\underline{w}_2\right\} \cap \mathbf{I}} \mathbf{e}[\psi](-\mu) \dee\nu \mathrm{d} u+\iiint_{\left\{\underline{w}_1 \leq \underline{w} \leq \underline{w}_2\right\} \cap \mathbf{I}} \mathbf{e}[\psi](-\mu) \dee\nu \mathrm{d} u \mathrm{d} \underline{u} \label{eqn:pourapres} \\
&+\iint_{\left\{\underline{w}_1 \leq \underline{w} \leq \underline{w}_2\right\} \cap \{r=\rb\}} \mathbf{e}[\psi]\dee\nu\mathrm{d} \underline{u} \lesssim \iint_{\left\{\underline{w}=\underline{w}_1\right\} \cap \mathbf{I}} \mathbf{e}[\psi](-\mu) \dee\nu \mathrm{d} u+\int_{\underline{w}_1}^{\underline{w}_2} \underline{w}^{-2\beta} \mathrm{d} \underline{w}.\nonumber
\end{align}
{Dropping the term on $\{r=\rb\}$ which is non-negative, and using the fact that for any $\wbar_1\geq 1$, the initial energy 
$$\iint_{\left\{\underline{w}=\underline{w}_1\right\} \cap \mathbf{I}} \mathbf{e}[\psi](-\mu) \dee\nu \mathrm{d} u\lesssim 1$$
is finite\footnote{In the notation of Lemma \ref{lem:decay}, this guarantees that $f(1)$ is finite. This is clear by standard existence results for linear wave equations with smooth initial data, and because $\left\{\underline{w}=\underline{w}_1\right\} \cap \mathbf{I}$ is a compact region inside a globally hyperbolic spacetime.}, we conclude the proof by applying Lemma \ref{lem:decay} with $p=2\beta$.}
\end{proof}
The following result will be useful to deduce decay in $L^2(S(u,\ubar))$ norm from energy decay.
\begin{lem}\label{lem:integ}
    Let $\psi$ be a $\pm 2$ spin-weighted scalar on $\mcm$. We define the scalar function $$f(\wbar,r):=\|\psi\|_{L^2(S(u(r,\wbar),\ubar(r,\wbar)))}.$$
    Then for $\wbar\geq 1$ and $r_2\leq r_1$,
    $$f(\wbar,r_2)\lesssim f(\wbar,r_1)+\left(\int_{r_2}^{r_1}\int_{S(\wbar,r')}\ener[\psi]\dee\nu\dee r'\right)^{1/2}.$$
\end{lem}
\begin{proof}
    In ingoing Eddington-Finkelstein coordinates $(\ubar,r,\theta,\phi_+)$ we have $\partial_{r_\mathrm{in}}=-2e_3$, $\partial_{\ubar_{\mathrm{in}}}=T$. We deduce, changing of variables from $(\wbar,r)$ to $(\ubar,r)$,
    \begin{align}\label{eqn:intapres}
        \frac{\partial}{\partial r}(f(\wbar,r))=\left[\frac{\partial}{\partial r}(\wbar+r-r_+)\partial_{\ubar_\mathrm{in}}+\frac{\partial r}{\partial r}\partial_{r_\mathrm{in}}\right]f(\ubar,r)=(\partial_{\ubar_\mathrm{in}}+\partial_{r_\mathrm{in}})f(\ubar,r).
    \end{align}
    Moreover, by a Cauchy-Schwarz inequality\footnote{Notice that even in $\{r_-\leq r\leq \rb\}$, we can rewrite the $L^2(S(u,\ubar))$ norm as an integration with respect to $\phi_+$, doing the change of variables $\phi_+=\phi_-+2r_{mod}$, as shown in Section \ref{section:spinscalars}.},
    \begin{align*}
        |(\partial_{r_\mathrm{in}}+\partial_{\ubar})(f^2)|&=\left|(\partial_{r_\mathrm{in}}+\partial_{\ubar})\left(\int_0^\pi\int_0^{2\pi}|\psi|^2(\ubar,r,\theta,\phi_+)\sin\theta\dee\theta\dee\phi_+\right)\right|\\
        &=\left|2\mathfrak{R}\left(\int_0^\pi\int_0^{2\pi}(\overline{\psi} (\partial_{r_\mathrm{in}}+\partial_{\ubar})\psi)(\ubar,r,\theta,\phi_+)\sin\theta\dee\theta\dee\phi_+\right)\right|\\
        &\lesssim 2f\|(\partial_{r_\mathrm{in}}+\partial_{\ubar})\psi\|_{L^2(S(\wbar,r))}=2f\|(T-2e_3)\psi\|_{L^2(S(\wbar,r))},
    \end{align*}
    which yields
    \begin{align}\label{eqn:pourint1}
        |(\partial_{r_\mathrm{in}}+\partial_{\ubar})f|=\frac{1}{2f}|(\partial_{r_\mathrm{in}}+\partial_{\ubar})(f^2)|\lesssim\|(T-2e_3)\psi\|_{L^2(S(\wbar,r))}.
    \end{align}
    Integrating \eqref{eqn:intapres}, together with the bound \eqref{eqn:pourint1} gives
    $$f(\wbar,r_2)\lesssim f(\wbar,r_1)+\int_r^{r_+}\left(\int_{S(\wbar,r')}|T\psi|^2+|e_3\psi|^2\dee\nu\right)^{1/2}\dee r',$$
    and we conclude the proof of Lemma \ref{lem:integ} using the decomposition \eqref{eqn:expreT} that gives $T=O(1)e_3+O(1)e_4+O(1)\mcu+O(1)$, the Poincaré inequality \eqref{eqn:poincare}, and a Cauchy-Schwarz inequality.
\end{proof}
\begin{prop}\label{prop:decrT}
    Assume that $\psi$ is a spin $-2$ scalar that satisfies \eqref{eqn:hyppsiborne}, \eqref{eqn:hyppsiteuk} with $c>0$, $\beta>1$. Then we have in $\mathbf{I}$, for $0\leq j\leq N_j^-$ and $0\leq 2 k_1+k_2\leq N_k^--1$,
    $$\|T^j\carterm^{k_1}\Phi^{k_2}\psi\|_{L^2(S(u,\ubar))}\lesssim \ubar^{-\beta-j}.$$
\end{prop}
\begin{proof}
As $T$, $\Phi$, and $\carterm$ commute with $\teuk_{-2}^{(c,V)}$, it suffices to prove the case $j=k_1=k_2=0$. Let $(u,\ubar)\in\un$ and denote $\wbar_1=\wbar(u,\ubar)$, $r=r(u,\ubar)$. Using Lemma \ref{lem:integ} with $r_1=r_+$ and $r_2=r$ we get 
$$\|\psi\|_{L^2(S(u,\ubar))}\lesssim\|\psi\|_{L^2(S(-\infty,\wbar_1))}+\left(\int_{r}^{r_+}\int_{S(\wbar,r')}\ener[\psi]\dee\nu\dee r'\right)^{1/2}.$$
Changing variables from $\dee r$ to $\dee u$, taking into account Footnote \ref{footnote}, we get using Proposition \ref{prop:ener} :
$$\left(\int_{r}^{r_+}\int_{S(\wbar,r')}\ener[\psi]\dee\nu\dee r'\right)^{1/2}\lesssim \left(\iint_{\left\{\underline{w}=\underline{w}_1\right\} \cap \mathbf{I}} \mathbf{e}[\psi](-\mu) \dee\nu\:\dee u\right)^{1/2}\lesssim\wbar_1^{-\beta}.$$
Moreover, $\|\psi\|_{L^2(S(-\infty,\wbar_1))}\lesssim\wbar_1^{-\beta}$ by \eqref{eqn:hyppsiborne}, thus using $\wbar_1^{-\beta}\lesssim\ubar^{-\beta}$ we get
$$\|\psi\|_{L^2(S(u,\ubar))}\lesssim\wbar_1^{-\beta}\lesssim\ubar^{-\beta}.$$
\end{proof}
The previous energy and $L^2(S(u,\ubar))$ decay estimates, combined with the Sobolev embedding \eqref{eqn:sobolev} give the following polynomial bound propagation result.
\begin{prop}\label{prop:propagI}
    Assume that $\psi$ is a spin $-2$ scalar that satisfies \eqref{eqn:hyppsiborne}, \eqref{eqn:hyppsiteuk} with $c>0$, $\beta>1$. Then we have in $\un$, for $0\leq j\leq N_j^--2$ and $0\leq 2 k_1+k_2\leq N_k^--3$,
    $$|T^j\mcq_{-2}^{k_1}\Phi^{k_2}\psi|\lesssim\ubar^{-\beta-j}.$$
\end{prop}
\begin{proof}
It suffices to treat the case $j=k=0$. Using the Sobolev embedding \eqref{eqn:sobolev} and Proposition \ref{prop:decrT},  we get
$$|\psi|^2\lesssim \int_{S(u,\ubar)}|T^{\leq 2}\psi|^2\dee\nu+\int_{S(u,\ubar)}|\mcq_{-2}\psi|^2\dee\nu\lesssim\ubar^{-2\beta},$$
which is the stated estimate for $j=k=0$.
\end{proof}

We conclude this section with the following result, that proves decay of $e_3\psi$ under additionnal assumptions on $\mch_+$. We also prove a precise energy bound on $\{r=\rb\}$. 
\begin{thm}\label{thm:psim} Assume that $\psi$ is a spin $-2$ scalar that satisfies, for $0\leq j\leq N_j^-$, $0\leq |k|\leq N_k^-$, and $0\leq 2k_1+k_2\leq N_k^-$,
\begin{align}
    \bullet\:&e_3^{\leq 1}T^j\nablabar^k\psi=O(\ubar^{-\beta-j})\text{   on   }\mch_+\cap\{\ubar\geq 1\},\label{eqn:hyppsiborne'}\\
    \bullet\:&e_3^{\leq 1}\teuk_{-2}T^j\mcq_{-2}^{k_1}\Phi^{k_2}\psi=O(\ubar^{-\beta-j})\text{   in   }\un,\label{eqn:hyppsiteuk'}
\end{align}
 where $\beta>1$. Then for $j\leq N_j^--2$ and $2k_1+k_2\leq N_k^--3$, we have in $\un$,
\begin{align}\label{eqn:provprov}
        |e_3^{\leq 1}T^j\carterm^{k_1}\Phi^{k_2}\psi|\lesssim\ubar^{-\beta-j},
    \end{align}
    and for $1\leq\wbar_1\leq\wbar_2$ we have the energy bound on $\{r=\rb\}$
    \begin{align}\label{eqn:enerprecise}
        \iint_{\left\{\underline{w}_1 \leq \underline{w} \leq \underline{w}_2\right\} \cap \{r=\rb\}} \mathbf{e}[T^j\carterm^{k_1}\Phi^{k_2}e_3^{\leq 1}\psi]\dee\nu\mathrm{d} \underline{u}\lesssim\wbar_1^{-2\beta-2j}+\int_{\underline{w}_1}^{\underline{w}_2} \underline{w}^{-2\beta-2j} \mathrm{d} \underline{w}.
    \end{align}
\end{thm}
\begin{rem}
    Assumptions \eqref{eqn:hyppsiborne'} and \eqref{eqn:hyppsiteuk'} are satisfied by $\errm$ with $\beta=7+\delta$, see the proof of Proposition \ref{prop:ansatzI}.
\end{rem}
\begin{proof}[Proof of Theorem \ref{thm:psim}]
    We already proved \eqref{eqn:provprov} in the case without the $e_3$ derivative in Proposition \ref{prop:propagI}. Thus it remains to prove the bound \eqref{eqn:provprov} with the $e_3$ derivative, as well as \eqref{eqn:enerprecise}. The proof is based Proposition \ref{prop:propagI}, and on a commutation of the Teukolsky operator with $e_3$. Using Proposition \ref{prop:comme3}, we find that $e_3\psi$ satisfies
\begin{align*}
    \Big(\teuk_{-2}-4[(r-M)e_3-rT]+6\Big)\Big[e_3\psi\Big]=5T\psi+e_3\teuk_{-2}\psi.
\end{align*}
Commuting with $T^j\carterm^{k_1}\Phi^{k_2}$ and using \eqref{eqn:recupc}, \eqref{eqn:provprov} without the $e_3$ derivative, and \eqref{eqn:hyppsiteuk'} yields\footnote{Note that the bound for the RHS holds for $j\leq N_j-3$ but after integration on the spheres (which is the only bound that we need in the energy estimates) it holds for $j\leq N_j-2$ by Proposition \ref{prop:decrT}.}
\begin{align}\label{eqn:inun}
    \teuk_{-2}^{(3/2,6)}\Big[T^j\carterm^{k_1}\Phi^{k_2}e_3\psi\Big]=O(\ubar^{-\beta-j}),\text{   in   }\un,
\end{align}
Thus using assumption \eqref{eqn:hyppsiborne'} on $\mch_+$ with the $e_3$ derivative, by Proposition \ref{prop:propagI} with the parameters $\beta>1$, $c=3/2>0$ and $V=6$, we get for $j\leq N_j^--2$ and $2k_1+k_2\leq N_k^--3$,
\begin{align*}
    |T^j\carterm^{k_1}\Phi^{k_2}e_3\psi|\lesssim\ubar^{-\beta-j},\text{   in   }\un,
\end{align*}
as stated. To get the energy bound \eqref{eqn:enerprecise}, notice that dropping the non-negative first and second terms on the LHS of \eqref{eqn:pourapres} gives in this context, for $\wbar_2\geq\wbar_1$,
\begin{align*}
 \iint_{\left\{\underline{w}_1 \leq \underline{w} \leq \underline{w}_2\right\} \cap \{r=\rb\}} &\mathbf{e}[T^j\carterm^{k_1}\Phi^{k_2}e_3^{\leq 1}\psi]\dee\nu\mathrm{d} \underline{u}\lesssim\\
 &\iint_{\left\{\underline{w}=\underline{w}_1\right\} \cap \mathbf{I}} \mathbf{e}[T^j\carterm^{k_1}\Phi^{k_2}e_3^{\leq 1}\psi](-\mu) \dee\nu \mathrm{d} u+\int_{\underline{w}_1}^{\underline{w}_2} \underline{w}^{-2\beta-2j} \mathrm{d} \underline{w}.\nonumber
\end{align*}
Using Proposition \ref{prop:ener} for $e_3^{\leq 1}\psi$ yields 
\begin{align*}
 \iint_{\left\{\underline{w}_1 \leq \underline{w} \leq \underline{w}_2\right\} \cap \{r=\rb\}} \mathbf{e}[T^j\carterm^{k_1}\Phi^{k_2}e_3^{\leq 1}\psi]\dee\nu\mathrm{d} \underline{u}\lesssim\wbar_1^{-2\beta-2j}+\int_{\underline{w}_1}^{\underline{w}_2} \underline{w}^{-2\beta-2j} \mathrm{d} \underline{w},
\end{align*}
as stated.
\end{proof}
\begin{rem}
    Actually, further commutations with $e_3$ only improve the redshift effect. More precisely, for $k\geq 0$, $e_3^k\psi$ satisfies
    $$\teuk_{-2}^{(c_k,V_k)}[e_3^k\psi]=e_3^k\teuk_{-2}\psi+T^{\leq 1}e_3^{\leq k-1}\psi,$$
    where $c_k=(k+2)/2$. Thus, assuming decay of $e_3^k\teuk_{-2}\psi$ in $\un$ and of $e_3^k\psi$ on $\mch_+$ allows one to use Proposition \ref{prop:propagI} to successively control all the derivatives $e_3^k\psi$, $k\geq 0$, in $\un$.
\end{rem}

\subsection{Precise asymptotics of $\psi_{-2}$ in region $\mathbf{I}$}\label{subsection:applyI}
We state the main result of this section.
\begin{prop}\label{prop:ansatzI}
    Assume that $\psim$ satisfies \eqref{eqn:hypm}. Then, we have in $\mathbf{I}$,
    \begin{align}\label{eqn:ansatzmI}
        \psim =\ansatzm+\errm,
    \end{align}
where for $j\leq N_j^--2$ and $2k_1+k_2\leq N_k^--3$,
$$|T^j\mcq_{-2}^{k_1}\Phi^{k_2}\errm|\lesssim\ubar^{-7-j-\delta},\:\:\text{ in }\un.$$
If, additionally, we assume that for $j\leq N_j^-$ and $|k|\leq N_k^-$,
$$\left|e_3T^j\nablabar^k\errm\right|\lesssim\ubar^{-7-j-\delta}\text{   on   }\mch_+\cap\{\ubar\geq 1\},$$
then for $j\leq N_j^--2$ and $2k_1+k_2\leq N_k^--3$,
\begin{align}\label{eqn:labeler}
    |e_3T^j\mcq_{-2}^{k_1}\Phi^{k_2}\errm|\lesssim\ubar^{-7-j-\delta},\:\:\text{ in }\un.
\end{align}
\end{prop}

\begin{proof}
    Let 
    $$\errm=\Psim-\ansatzm.$$
    Using $\teuk_{-2}\psi_{-2}=0$, as well as $[T,\teuk_{-2}]=[\mcq_{-2},\teuk_{-2}]=[\Phi,\teuk_{-2}]=0$, we get in $\un$
\begin{align}\label{eq:premiereacit}
    e_3^{\leq 1}\teuk_{-\fraks}T^j\mcq_{-2}^{k_1}\Phi^{k_2}\errm=-e_3^{\leq 1}T^j\mcq_{-2}^{k_1}\Phi^{k_2}\teuk_{-2}\left(\ansatzm\right).
\end{align}
Notice that $e_3(\theta)=e_3(\ubar)=e_3(\phi_+)=0$ so that 
\begin{align}\label{eqn:e3ansatz}
    e_3\left(\ansatzm\right)=0.
\end{align}
Thus using the expression of the Teukolsky operator given by Proposition \ref{prop:teukavece3}, we get 
\begin{align}
    \teuk_{-2}&\left(\ansatzm\right)\nonumber\\
    &=\Big(a^2\sin^2\theta T^2+2aT\Phi-(10r+4ia\cos\theta)T\Big)\left[\ansatzm\right],\label{eqn:nonum}
\end{align}
since we also have $\drond\drond' Y_{m,2}^{-2}(\cos\theta)=0$ by \eqref{eqn:annulemode}. The explicit computation of \eqref{eqn:nonum} yields 
\begin{align}\label{eq:deuxiemeacit}
    e_3^{\leq 1}T^j\mcq_{-2}^{k_1}\Phi^{k_2}\teuk_{-2}\left(\ansatzm\right)=O(\ubar^{-8-j}).
\end{align}
Thus $\errm$ satisfies \eqref{eqn:hyppsiborne} and \eqref{eqn:hyppsiteuk} with $\beta=7+\delta>1$, $c=1$, and $V=0$ thanks to the assumption \eqref{eqn:hypm} on $\mch_+$. Using Proposition \ref{prop:propagI}, we get in $\un$
\begin{align*}
    |T^j\mcq_{-2}^{k_1}\Phi^{k_2}\errm|\lesssim\ubar^{-7-j-\delta},
\end{align*}
as stated. The result \eqref{eqn:labeler} with the $e_3$ derivative is a direct application of Theorem \ref{thm:psim}.
\end{proof}

\begin{rem}
Notice that \eqref{eqn:enerprecise} implies in particular the energy bound on $\{r=\rb\}\cap\{\ubar\geq 1\}$ :
    \begin{align}\label{eqn:enerpourapres}
    \iint_{\{r=\rb\}\cap\{\ubar\geq 1\}}\ener[T^j\mcq_{-2}^{k_1}\Phi^{k_2}e_3^{\leq 1}\psim]\dee\nu\dee\ubar\lesssim 1.
\end{align}
We will use the symmetric version of this bound for $\psihatp$, in the region $\{u\geq 1\}$. This is where we will use the initial assumption \eqref{eqn:hypp'} on $\mch'_+$. Recall that it will only be used to deduce the precise asymptotics of $\psihatp$ on the upper part of the Cauchy horizon, i.e. in $\quatre$, see Remark \ref{rem:cellela}. The symmetric argument in the region $\{u\geq 1\}$, together with assumption \eqref{eqn:hypp'} gives the following analog of \eqref{eqn:enerpourapres} : 
\begin{align}\label{eqn:rienvraiment}
    \iint_{\{r=\rb\}\cap\{u\geq 1\}}\ener[T^j\carterp^{k_1}\Phi^{k_2}\eqhat^{\leq 1}\psihatp]\dee\nu\dee u\lesssim 1.
\end{align}
\end{rem}
\subsection{Precise asymptotics of $\psi_{+2}$ in region $\mathbf{I}$}\label{subsection:usetsi}
We have the ansatz for $\psip$ on the event horizon, given by \eqref{eqn:hypp}. We show that this ansatz propagates to region $\un$, using the asymptotics for $\psim$ in $\un$ derived in Section \ref{subsection:applyI}, and the Teukolsky-Starobinsky identity (TSI) \eqref{eqn:tsi2}. The idea is that the TSI can be rewritten as a relation between $\errm$ and $\errp$, with error terms that can be bounded using the fact that $T$ derivatives of $\Psim$ gain powers of $\ubar^{-1}$. This will imply that the $O(\ubar^{-7-\delta})$ bound for $\errm$ in $\un$ also holds for $\errp$.
\begin{prop}\label{prop:ansatzIp}
    Assume that $\Psip$ satisfies \eqref{eqn:hypp}. Then we have in $\un$,
    \begin{align}\label{eqn:ansatzpI}
        \psip=\ansatzp+\errp,
    \end{align}
where for $j\leq \min(N_j^{-}-10,N_j^{+})$ and $2k_1+k_2\leq \min(N_k^{-}-9,N_k^+)$, 
\begin{align}\label{eqn:aprover}
    |e_3^{\leq 3}T^j\mcq_{+2}^{k_1}\Phi^{k_2}\errp|\lesssim\ubar^{-7-j-\delta}.
\end{align}
\end{prop}
\begin{proof}
Let 
$$\errp:=\Psip-\ansatzp.$$
To highlight the important points of the argument, we begin with the case $k_1=0$. We commute with $T^j\Phi^{k_2}$ and compute the LHS of TSI \eqref{eqn:tsi2} to get 
\begin{align}\label{eqn:substract}
    (\drond)^4T^j\Phi^{k_2}\psim+12MT^{j+1}\Phi^{k_2}\overline{\psim}+\sum_{1\leq p\leq 4,0\leq q\leq 3}O(1)(\drond)^{q}T^{p+j}\Phi^{k_2}\psim=\partial_{r_\mathrm{in}}^4T^j\Phi^{k_2}\psip.
\end{align}
Next, using \eqref{eqn:drooo} four times, and the expression \eqref{eqn:am(r)} of $A_m(r)$ to get $\partial_{r_\mathrm{in}}^4(A_m(r))=A_m^{(4)}(r)=24$, we obtain the identity
\begin{align*}
    (\drond)^4T^j\Phi^{k_2}\left(\ansatzm\right)&=24T^j\Phi^{k_2}\left(\sum_{|m|\leq 2}Q_{m,2}Y_{m,2}^{+2}(\cos\theta)e^{im\phi_{+}}\right)\\
    &=\partial_{r_\mathrm{in}}^4T^j\Phi^{k_2}\left(\ansatzp\right).
\end{align*}
Subtracting this identity from both sides of \eqref{eqn:substract} yields 
\begin{align}\label{eqn:ineq3}
    (\drond)^4T^j\Phi^{k_2}\errm+12MT^{j+1}\Phi^{k_2}\overline{\psim}+\sum_{1\leq p\leq 4,0\leq q\leq 3}O(1)(\drond)^{q}T^{p+j}\Phi^{k_2}\psim=\partial_{r_\mathrm{in}}^4T^j\Phi^{k_2}\errp.
\end{align}
The crucial point is that the second and third terms on the LHS of \eqref{eqn:ineq3} contains at least one extra $T$ derivative\footnote{Also notice that we lose a lot of derivatives to control the third term.} compared to the other terms in \eqref{eqn:ineq3}. In order to estimate this term, we use the Sobolev embedding \eqref{eqn:sobolevbase} and {\eqref{eqn:drooo}, \eqref{eq:vpspin} three times to obtain\footnote{{Note that to obtain the second inequality, denoting $\psi=T^{p+j}\Phi^{k_2}\psim$, we use \eqref{eqn:drooo}, \eqref{eq:vpspin} three times, and the fact that $(Y_{m,2}^s(\cos\theta)e^{im\phi_+})_{\ell\geq 2,|m|\leq \ell}$ is an orthonormal basis of spin $s$ scalars on $S(u,\ubar)$, to get the elliptic estimate
$$\|(\drond'\drond)^{\leq 1}(\drond)^{\leq 3}\psi\|_{L^2(S(u,\ubar))}^2\lesssim \sum_{\ell\geq 2,|m|\leq\ell} |\psi_{m,\ell}|^2\ell^{10}\lesssim\sum_{\ell\geq 2,|m|\leq\ell} |\psi_{m,\ell}|^2((\ell+2)(\ell-1))^6=\|(\drond\drond')^3\psi\|_{L^2(S(u,\ubar))}^2,$$
where $\psi_{m,\ell}$ is the $(m,\ell)$ component of $\psi$ on the basis $(Y_{m,2}^{-2}(\cos\theta)e^{im\phi_+})_{\ell\geq 2,|m|\leq \ell}$.}} in $\un$}
\begin{align*}
    \Bigg|12M&T^{j+1}\Phi^{k_2}\overline{\psim}+\sum_{1\leq p\leq 4,0\leq q\leq 3}O(1)(\drond)^{q}T^{p+j}\Phi^{k_2}\psim\Bigg|\\
    &\lesssim\sum_{1\leq p\leq 4,0\leq q\leq 3}\|(\drond'\drond)^{\leq1}(\drond)^{q}T^{p+j}\Phi^{k_2}\psim\|_{L^2(S(u,\ubar))}\lesssim \sum_{p=1}^4\|(\drond'\drond)^{\leq 3}T^{j+p}\Phi^{k_2}\psim\|_{L^2(S(u,\ubar))}.
\end{align*}

Using the definition \eqref{eqn:carter} of the Carter operator we compute 
$$(\drond'\drond)^{\leq 3}T^{j+p}\Phi^{k_2}\psim=\sum_{0\leq n\leq 3,\: 0\leq m\leq 6 }O(1)\mcq_{-2}^nT^mT^{j+p}\Phi^{k_2}\psim.$$
Thus using Proposition \ref{prop:ansatzI} we get\footnote{Notice that the bound $|(\drond'\drond)^{\leq 3}T^{j+p}\Phi^{k_2}\psim|\lesssim \ubar^{-8-j}$ holds only for $j\leq N_j-12$ by Proposition \ref{prop:propagI}, but using instead Proposition \ref{prop:decrT} we get $\|(\drond'\drond)^{\leq 3}T^{j+p}\Phi^{k_2}\psim\|_{L^2(S(u,\ubar))}\lesssim \ubar^{-8-j}$ for $j\leq N_j-10$.} for $p\geq 1$, $|(\drond'\drond)^{\leq 3}T^{j+p}\Phi^{k_2}\psim|\lesssim \ubar^{-8-j}$, and hence 
\begin{align}\label{eqn:ineq1}
    \left|\sum_{1\leq p\leq 4,0\leq q\leq 3}O(1)(\drond)^{q}T^{p+j}\Phi^{k_2}\psim\right|\lesssim\ubar^{-8-j}.
\end{align}
Next, we use the Sobolev embedding again to get $$|(\drond)^4T^j\Phi^{k_2}\errm|\lesssim \|(\drond'\drond)^{\leq3}T^j\Phi^{k_2}\errm\|_{L^2(S(u,\ubar))} $$ and we can compute again the RHS to get 
$$(\drond'\drond)^{\leq3}T^j\Phi^{k_2}\errm=\sum_{0\leq n\leq 3,\:0\leq m\leq 6}O(1)\mcq_{-2}^nT^mT^j\Phi^{k_2}\errm $$
which gives, using Proposition \ref{prop:ansatzI}, 
\begin{align}\label{eqn:ineq2}
    |(\drond)^4T^j\Phi^{k_2}\errm|\lesssim\ubar^{-7-j-\delta}.
\end{align}
Combining \eqref{eqn:ineq3}, \eqref{eqn:ineq1} and \eqref{eqn:ineq2} yields 
\begin{align}\label{eqn:aint}
    |\partial_{r_\mathrm{in}}^4T^j\Phi^{k_2}\errp|\lesssim\ubar^{-7-j-\delta}\text{   in    }\un.
\end{align}
Now, using ingoing Eddington-Finkelstein coordinates $(\ubar,r,\theta,\phi_{+})$, and using the initial condition \eqref{eqn:hypp} on $\mch_+$, we easily integrate \eqref{eqn:aint} four times from $r_+$ to $r$ on $\ubar=cst$ and use $\drin=-2e_3$ to get $|e_3^{\leq 3}T^j\Phi^{k_2}\errp|\lesssim\ubar^{-7-j-\delta}$ in $\un$, which concludes the proof of \eqref{eqn:aprover} in the case $k_1=0$.

Finally, to treat the case $k_1\neq 0$, first notice that we can compute again
$$T^j \mcq_{+2}^{k_1}\Phi^{k_2}\errp=\sum_{0\leq n\leq k,0\leq m\leq 2k}O(1)(\drond'\drond)^nT^{m+j}\Phi^{k_2}\errp$$
so we only need to show the $O(\ubar^{-7-j-\delta})$ decay of $(\drond'\drond)^nT^{j}\Phi^{k_2}\errp$ for any $j,n$. Differentiating \eqref{eqn:ineq3} by $(\drond'\drond)^n$ gives
\begin{align*}
    (\drond'\drond)^n(\drond)^4T^j\Phi^{k_2}\errm+\sum_{1\leq p\leq 4,0\leq q\leq 3}O(1)(\drond'\drond)^n(\drond)^{q}T^{p+j}\Phi^{k_2}\psim=\partial_{r_\mathrm{in}}^4(\drond'\drond)^nT^j\Phi^{k_2}\errp,
\end{align*}
and applying the exact same techniques as in the case $k_1=0$, controlling any angular derivative $\drond,\drond'\drond$ using the Carter operator, gives 
\begin{align}\label{eqn:bienbien}
    |\partial_{r_\mathrm{in}}^4T^j \mcq_{+2}^{k_1}\Phi^{k_2}\errp|\lesssim\ubar^{-7-j-\delta},
\end{align}
and thus \eqref{eqn:aprover} by integrating from $\mch_+$ on 
$\ubar=cst$ as in the case $k_1=0$.
\end{proof}
\begin{rem}\label{rem:cellela}
Using the fact that the integral is taken on $\{r=\rb\}$, the control \eqref{eqn:ansatzpI} of $T^j\carterp^{k_1}\Phi^{k_2}\psip$ in $\un$ as well as the relation $\eqhat=e_3+O(1)T+O(1)\Phi$ on $\{r=\rb\}$, we can rewrite the energy bound \eqref{eqn:rienvraiment} as
\begin{align}\label{eqn:vraienerpourapres}
    \iint_{\{r=\rb\}\cap\{u\geq 1\}}\ener[T^j\carterp^ke_3^{\leq 1}\psip]\dee\nu\dee u\lesssim 1.
\end{align}
\end{rem}
We continue this section by proving an energy bound for $\errp$ on $\{r=\rb\}\cap\{\ubar\geq 1\}$. This will be useful information on the initial data when doing the energy estimates for the spin $+2$ Teukolsky equation in region $\deux$. The following result is a corollary of Proposition \ref{prop:ansatzIp}.
 \begin{cor}\label{cor:hypener}
     Assume that $\Psip$ satisfies \eqref{eqn:hypp}. Then we have, for $j\leq \min(N_j^--12,N_j^+-2)$ and $2k_1+k_2\leq \min(N_k^--11,N_k^+-2)$, 
     $$\int_{{S(u,\ubar)}}\ener[T^j\mcq_{+2}^{k_1}\Phi^{k_2}e_3^{\leq 1}\errp]\dee\nu=O(\ubar^{-14-2\delta-2j})\text{   on   }\{r=\rb\}\cap\{\ubar\geq 1\}.$$ 
 \end{cor}
 \begin{proof}
Using \eqref{eqn:aprover} we get 
\begin{align}\label{eqn:bienbienbien}
    |e_3^{\leq 2}T^j \mcq_{+2}^{k_1}\Phi^{k_2}\errp|\lesssim\ubar^{-7-j-\delta}, \text{   on   }\{r=\rb\}.
\end{align}
Next, we write 
$$e_4=-\mu e_3+T+\frac{a}{r^2+a^2}\Phi.$$ 
Using Proposition \ref{prop:ansatzIp}, we can bound
\begin{align*}
    \|\Phi T^j \mcq_{+2}^{k_1}\Phi^{k_2}e_3^{\leq 1}\errp\|_{L^2(S(u,\ubar))}\lesssim\ubar^{-7-j-\delta}
\end{align*}
on $\{r=\rb\}$, as well as $|TT^j \mcq_{+2}^{k_1}\Phi^{k_2}e_3^{\leq 1}\errp|\lesssim\ubar^{-8-j}$. Thus we get 
$$\intS|e_4T^j \mcq_{+2}^{k_1}\Phi^{k_2}e_3^{\leq 1}\errp|\dee\nu\lesssim\ubar^{-14-2\delta-2j}, \text{   on   }\{r=\rb\}.$$
Finally, denoting 
$$\psi:=T^j \mcq_{+2}^{k_1}\Phi^{k_2}e_3^{\leq 1}\errp,$$
to bound $|\partial_\theta\psi|^2+|\mathcal{U}\psi|^2$ we first write
$$|\partial_\theta\psi|^2+|\mathcal{U}\psi|^2\lesssim |\partial_\theta\psi|^2+\frac{1}{\sin^2\theta}|\Phi\psi+is\cos\theta\psi|^2+|\psi|^2.$$
Then we use an integration by parts formula (see for example \cite[Eq. (32)]{dhrteuk}) to get
\begin{align}
    \int_{{S(u,\ubar)}}|\partial_\theta\psi|^2+\frac{1}{\sin^2\theta}|\Phi\psi+2i\cos\theta\psi|^2\dee\nu&=\int_{S(u,\ubar)}(\drond'\drond+2)\psi\cdot\overline{\psi}\dee\nu\nonumber\\
    &\lesssim \|\psi\|^2_{L^2({S(u,\ubar)})}+\|(\drond'\drond)\psi\|^2_{L^2({S(u,\ubar)})}.\label{eqn:ipppppp}
\end{align}
This gives, reinjecting the definition of the Carter operator \eqref{eqn:carter}, and using Proposition \ref{prop:ansatzIp} ,
\begin{align*}
    \int_{{S(u,\ubar)}}|\partial_\theta T^j \mcq_{+2}^{k_1}\Phi^{k_2}e_3^{\leq 1}\errp|^2+|\mathcal{U}T^j &\mcq_{+2}^{k_1}\Phi^{k_2}e_3^{\leq 1}\errp|^2\dee\nu\\
    &\lesssim \|T^{\leq 2}\mcq_{+2}^{\leq k+1}T^je_3^{\leq 1}\errp\|^2_{L^2({S(u,\ubar)})}\lesssim\ubar^{-14-2\delta-2j},
\end{align*}
which concludes the proof of Corollary \ref{cor:hypener}.
 \end{proof}
 We continue with an energy boundedness result on $\{r=\rb\}$ for $\Psip$ that will be used only at the end of the paper to get the precise asymptotics in region $\quatre$.
\begin{cor}\label{cor:symmetric}
    Assume that $\Psip$ satisfies \eqref{eqn:hypp}. Then we have for $j\leq \min(N_j^--12,N_j^+-2)$ and $2k_1+k_2\leq \min(N_k^--11,N_k^+-2)$,
    $$\iint_{\{r=\rb\}}\ener[T^j\carterp^{k_1}\Phi^{k_2}e_3^{\leq 1}\psip]\lesssim 1.$$
\end{cor}
\begin{proof}
    The proof is made by combining Corollary \ref{cor:hypener} on $\{r=\rb\}\cap\{\ubar\geq 1\}$ and the energy bound \eqref{eqn:vraienerpourapres} on $\{r=\rb\}\cap\{u\geq 1\}$.
\end{proof}

\section{Precise asymptotics in blueshift intermediate region $\mathbf{II}$}
\label{section:blueshiftII}
\subsection{Energy method for the spin $+2$ Teukolsky equation in $\deux$}\label{subsection:energyII}
In this section, we consider $\psi$ a spin $+2$ scalar such that there are constants
$$V\in\mathbb{R},\quad c>1/4,\quad\beta>1,$$
such that for $0\leq j\leq N_j'$, and $0\leq 2k_1+k_2\leq N_k'$,
\begin{align}
    \bullet\:&T^j\mcq_{+2}^{k_1}\Phi^{k_2}\psi=O(\ubar^{-\beta-j})\text{   on   }\{r=r_\mathfrak{b}\}\cap\{\ubar\geq 1\}
    ,\label{eqn:hypborneII}\\
    \bullet\:&\intS\ener[T^j\mcq_{+2}^{k_1}\Phi^{k_2}\psi]\dee\nu=O(\ubar^{-2\beta-2j})\text{   on   }\{r=r_\mathfrak{b}\}\cap\{\ubar\geq 1\},\label{eqn:hypenerII}\\
    \bullet\:&\teuk_{+2}^{(c,V)}T^j\mcq_{+2}^{k_1}\Phi^{k_2}\psi=O(\ubar^{-\beta-j})\text{   in   }\deux.\label{eqn:hypteukII}
\end{align}
Recall that we have
\begin{align}
    \widehat{e_3}=-\mu e_3,\quad \widehat{e_4}=(-\mu)^{-1}e_4.
\end{align}
The goal of this section is to propagate the polynomial decay \eqref{eqn:hypborneII} on the spacelike hypersurface $\{r=\rb\}$ to region $\deux$, see Proposition \ref{prop:propagII}. In this region, we consider the positive $+2$ spin because it provides a positive bulk term in the energy estimate for the Teukolsky equation. We begin by proving the following energy estimate, that holds only for $c>1/4$ (more precisely, see \eqref{eqn:use1/4}).
\begin{prop}\label{prop:enerdeuxx}
    Assume that $\psi$ is a spin $+2$ scalar satisfying \eqref{eqn:hypenerII} and \eqref{eqn:hypteukII} with $c>1/4$, $\beta>1$. Then for $\gamma>0$ small enough, we have, for $j\leq N_j'$, $2k_1+k_2\leq N_k'$, and $\wbar_1\geq 1$,
    $$\iint_{\left\{\wbar={\wbar}_1\right\} \cap \mathbf{II}} \ener[T^j\mcq_{+2}^{k_1}\Phi^{k_2}\psi](-\mu) \dee\nu \mathrm{~d} u\lesssim \wbar_1^{-2\beta-2j}.$$
\end{prop}
\begin{rem}
    From now on, we assume that $\gamma>0$ is small enough such that Proposition \ref{prop:enerdeuxx} holds.
\end{rem}
\begin{proof}[Proof of Proposition \ref{prop:enerdeuxx}]
In what follows, we denote $s=+2$. Similarly as in $\un$, the computations work for a general spin $s$, but the bulk term will only be positive for $s>0$. Once again, it suffices to treat the case $j=k_1=k_2=0$. Recall that in the coordinates $(u,\ubar,\theta,\phi_{-})$, we have $\partial_u=-\mu e_3+a/(r^2+a^2)\Phi$, $\partial_\ubar=e_4$. Thus we compute, in region $\mathbf{II}$ :
\begin{align}
    \mu\teuk_s^{(c,V)}&=4(r^2+a^2)\widehat{e_3}\partial_\ubar+ \mu \mathcal{U}^2+\frac{\mu}{\sin\theta}\partial_\theta(\sin\theta\partial_\theta)-4 ias\mu\cos\theta T+2is\mu \frac{a^2\sin\theta\cos\theta}{(r^2+a^2)}\mathcal{U}\nonumber\\
    &\quad\quad+2 r\mu (\widehat{e_3}+\partial_\ubar)-4 s[(r-M)\widehat{e_3}+ r\mu T]+\frac{2ar\mu  }{r^2+a^2}\Phi+ s\mu -s^2\mu\frac{a^4\sin^2\theta\cos^2\theta}{(r^2+a^2)^2},\label{eqn:muteukII1}\\
    &=4(r^2+a^2)\partial_\ubar \widehat{e_3}+ \mu \mathcal{U}^2+\frac{\mu}{\sin\theta}\partial_\theta(\sin\theta\partial_\theta)-4 ias\mu\cos\theta T+2is\mu \frac{a^2\sin\theta\cos\theta}{(r^2+a^2)}\mathcal{U}\nonumber\\
    &\quad\quad+2 r\mu (\widehat{e_3}+\partial_\ubar)-4 s[(r-M)\widehat{e_3}+ r\mu T]-\frac{2ar\mu  }{r^2+a^2}\Phi+ s\mu -s^2\mu\frac{a^4\sin^2\theta\cos^2\theta}{(r^2+a^2)^2},\label{eqn:muteukII2}
\end{align}
where we used 
$$4(r^2+a^2)[\partial_\ubar,\ethat]=\frac{4ar\mu}{r^2+a^2}\Phi.$$
Next, similarly as in $\un$, we multiply \eqref{eqn:hypteukII} by $\mu$ and by the complex conjugate of 
$$X(\psi):=f(r)\widehat{e_3}\psi+g(r)\partial_\ubar\psi$$
where $$f(r):=(r^2+a^2)^p(-\mu)^{-1},\quad g(r)=(r^2+a^2)^p,$$ we take the real part, and we integrate on ${S(u,\ubar)}$ with respect to $\dee\nu$. We get 
$$\intS \Real(f(r)\overline{\widehat{e_3}\psi}\mu\teuk_s^{(c,V)}\psi)\dee\nu+\intS\Real(g(r)\overline{\partial_\ubar\psi}\mu\teuk_s^{(c,V)}\psi)\dee\nu=\intS\mu\Real({\overline{X(\psi)}O(\ubar^{-\beta})})\dee\nu.$$
Making the substitutions\footnote{Recall that our convention for the derivatives $\partial_u$, $\partial_\ubar$ is that we use the ingoing double null like coordinates in $\{r\geq \rb\}$ and the outgoing double null like coordinates in $\{r\leq \rb\}$.} 
$$\partial_u\rightarrow \ethat,\quad e_4\rightarrow \partial_\ubar,$$
the same computation as in the proof of Lemma \ref{lem:integS} in Appendix \ref{appendix:computationsI} for the energy method in the redshift region gives
\begin{align}\label{eqn:dudubarII}
    \partial_\ubar\left(\int_{{S(u,\ubar)}}\widehat{\mathbf{F}}_\ubar[\psi]\dee\nu\right)+\partial_u\left(\int_{{S(u,\ubar)}}\widehat{\mathbf{F}}_u[\psi]\dee\nu\right)+\int_{{S(u,\ubar)}}\widehat{\mathbf{B}}[\psi]\dee\nu=\intS\mu\Real({\overline{X(\psi)}O(\ubar^{-\beta})})\dee\nu,
\end{align}
where 
$$\widehat{\mathbf{F}}_\ubar[\psi]=2(r^2+a^2)f(r)|\widehat{e_3}\psi|^2-\frac{1}{2}\mu g(r)(|\partial_\theta\psi|^2+|\mathcal{U}\psi|^2)+a\sin\theta \mu \mathfrak{R}(\overline{X(\psi)}\mathcal{U}\psi),$$
$$\widehat{\mathbf{F}}_u[\psi]=2(r^2+a^2)g(r)|\partial_\ubar\psi|^2-\frac{1}{2}\mu f(r)(|\partial_\theta\psi|^2+|\mathcal{U}\psi|^2)-a\sin\theta \mu \mathfrak{R}(\overline{X(\psi)}\mathcal{U}\psi)),$$
and
\begin{align*}
    \widehat{\mathbf{B}}[\psi]&=2(r\mu f(r)-\partial_\ubar((r^2+a^2)f(r)))|\widehat{e_3}\psi|^2+2(r\mu g(r)-\partial_u((r^2+a^2)g(r)))|\partial_\ubar\psi|^2\\
    &\quad +\frac{1}{2}(\partial_u(\mu f(r))+\partial_\ubar(\mu g(r)))(|\partial_\theta\psi|^2+|\mathcal{U}\psi|^2)\\
    &\quad+\mu g(r)\frac{sra^2\mu\cos\theta\sin\theta}{(r^2+a^2)^2}\mathfrak{I}(\overline{\psi} \mathcal{U}\psi)\\
    &\quad+4\mu g(r)as\cos\theta\mathfrak{I}(\overline{\partial_\ubar\psi}T\psi)-4srg(r)\mu\mathfrak{R}(\overline{\partial_\ubar\psi}T\psi)\\
    &\quad+\frac{2g(r)ar\mu}{r^2+a^2}\mathfrak{R}(\overline{\partial_\ubar\psi}\Phi\psi)-g(r)\left(\mu s+s^2\mu\frac{a^4\sin^2\theta\cos^2\theta}{(r^2+a^2)^2}-\mu V\right)\mathfrak{R}(\overline{\partial_\ubar\psi}\psi)\\
    &\quad-2sg(r)\mu\frac{a^2\sin\theta\cos\theta}{(r^2+a^2)}\mathfrak{I}(\overline{\partial_\ubar\psi}\mathcal{U}\psi)+g(r)(2r\mu-4cs(r-M))\Real(\overline{\partial_\ubar\psi} \widehat{e_3}\psi)\\
    &\quad+\mu f(r)\frac{sra^2\mu\cos\theta\sin\theta}{(r^2+a^2)^2}\mathfrak{I}(\overline{\psi} \mathcal{U}\psi)\\
    &\quad+4\mu f(r)as\cos\theta\mathfrak{I}(\overline{\widehat{e_3}\psi}T\psi)-4sr\mu\mathfrak{R}(\overline{\widehat{e_3}\psi}T\psi)+2 r\mu f(r)\mathfrak{R}(\overline{\widehat{e_3}\psi}\partial_\ubar\psi)\\
    &\quad-\frac{2f(r)ar\mu}{r^2+a^2}\mathfrak{R}(\overline{ \widehat{e_3} \psi}\Phi\psi)-f(r)\left(\mu s+s^2\mu\frac{a^4\sin^2\theta\cos^2\theta}{(r^2+a^2)^2}-\mu V\right)\mathfrak{R}(\overline{ \widehat{e_3} \psi}\psi)\\
    &\quad-2sf(r)\mu\frac{a^2\sin\theta\cos\theta}{(r^2+a^2)}\mathfrak{I}(\overline{ \widehat{e_3} \psi}\mathcal{U}\psi)-4cs(r-M)f(r)|\widehat{e_3}\psi|^2.
\end{align*}

Integrating \eqref{eqn:dudubarII} on $\{\wbar_1\leq \wbar\leq \wbar_2\}\cap\deux$ (see Figure \ref{fig:II}) with respect to $\dee u\dee\ubar$ gives 
\begin{align}
& \iint_{\left\{\wbar=\wbar_2\right\} \cap \mathbf{II}} \mathcal{T}_{\wbar}[\psi] \dee\nu \mathrm{~d} u+\iint_{\Gamma\cap\left\{\wbar_1 \leq \wbar\leq \wbar_2\right\}} \mathcal{T}_\Gamma[\psi] \dee\nu \mathrm{~d}  u+\iiint_{\left\{\wbar_1 \leq \wbar\leq \wbar_2\right\} \cap \mathbf{II}} \widehat{\mathbf{B}}[\psi] \dee\nu\:\dee u\dee\ubar \nonumber\\
= & \iint_{\left\{\wbar=\wbar_1\right\} \cap \mathbf{II}} \hspace{-0.5cm}\mathcal{T}_{\wbar}[\psi] \dee\nu \mathrm{~d} u+\iint_{\{r=r_{\mathfrak{b}}\}  \cap\left\{\wbar_1 \leq \wbar \leq \wbar_2\right\}} \hspace{-0.7cm}\mathcal{T}_r[\psi] \dee\nu \mathrm{~d}  u+\iiint_{\left\{\underline{w}_1 \leq \underline{w} \leq \underline{w}_2\right\} \cap \mathbf{I}} \hspace{-1cm}\mu\Real(\overline{X(\psi)}O(\ubar^{-\beta})) \dee\nu \mathrm{d} u \mathrm{d} \underline{u},\label{eqn:conserenerII}
\end{align}
where  $$\mathcal{T}_{\wbar}[\psi]=\left(1-\frac{1}{2} \mu\right) \widehat{\mathbf{F}}_{\ubar}[\psi]-\frac{1}{2} \mu \widehat{\mathbf{F}}_{{u}}[\psi],\quad\mathcal{T}_r[\psi]=-\frac{1}{2} \mu\left(\widehat{\mathbf{F}}_{\underline{u}}[\psi]+\widehat{\mathbf{F}}_u[\psi]\right),$$
and $$\mathcal{T}_\Gamma[\psi]=\widehat{\mathbf{F}}_u[\psi]+(1-\gamma\ubar^{\gamma-1})\widehat{\mathbf{F}}_{{\ubar}}[\psi].$$
\begin{figure}[h!]
    \centering
    \includegraphics[scale=0.47]{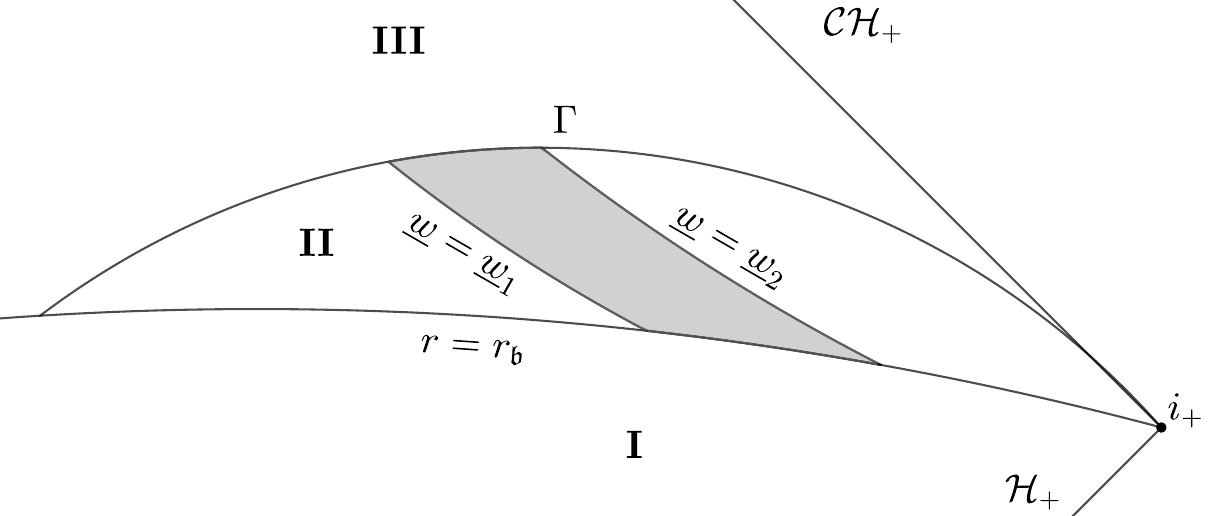}
    \caption{The blueshift energy estimate. The grey region is $\deux\cap\{\wbar_1\leq\wbar\leq\wbar_2\}$ where we integrate \eqref{eqn:dudubarII}.}
    \label{fig:II}
\end{figure}

Now we estimate the different quantities involved, with the previous choice of $f$ and $g$.

\noindent\textbf{Control of the bulk terms.} Using Lemma \ref{lem:posbulkII}, which holds thanks to an effective blueshift effect, we get that for $s=+2$, $c>1/4$, $p=p(a,M,c,V)\gg 1$ large enough, and $\rb=\rb(a,M,c,V)$ sufficiently close to $r_-$, we have in $\{r_-\leq r\leq\rb\}$,
    \begin{align}\label{eqn:bulkIIpos2}
    \intS\widehat{\mathbf{B}}[\psi]\dee\nu\gtrsim(-\mu)\intS{\ener}[\psi]\dee\nu.
\end{align}
To control the other bulk term on the RHS of \eqref{eqn:conserenerII}, we write, for $\varepsilon>0$, similarly as in $\un$,
\begin{align}
    \Bigg|\iiint_{\left\{\wbar_1 \leq \wbar\leq \wbar_2\right\} \cap \mathbf{II}}\mu\Real(\overline{X(\psi)}&O(\ubar^{-\beta}))\dee\nu\:\dee u\dee\ubar\Bigg|\lesssim\nonumber\\
    &\varepsilon \iiint_{\left\{\wbar_1 \leq \wbar\leq \wbar_2\right\} \cap \mathbf{II}}\ener[\psi](-\mu)\dee\nu\:\dee u\dee\ubar+\varepsilon^{-1}\int_{\wbar_1}^{\wbar_2}\wbar^{-2\beta}\dee\wbar.\label{eqn:eqnabs}
\end{align}
Thus, choosing $\varepsilon>0$ small enough such that the first term on the RHS of \eqref{eqn:eqnabs} is absorbed in the LHS of \eqref{eqn:conserenerIIb}, we get 
\begin{align}
& \iint_{\left\{\wbar=\wbar_2\right\} \cap \mathbf{II}} \mathcal{T}_{\wbar}[\psi] \dee\nu \mathrm{~d} u+\iint_{\Gamma\cap\left\{\wbar_1 \leq \wbar\leq \wbar_2\right\}} \mathcal{T}_\Gamma[\psi] \dee\nu \mathrm{~d}  u+\iiint_{\left\{\wbar_1 \leq \wbar\leq \wbar_2\right\} \cap \mathbf{II}} \ener[\psi] \dee\nu\:\dee u\dee\ubar \nonumber\\
= & \iint_{\left\{\wbar=\wbar_1\right\} \cap \mathbf{II}} \mathcal{T}_{\wbar}[\psi] \dee\nu \mathrm{~d} u+\iint_{\{r=r_{\mathfrak{b}}\}  \cap\left\{\wbar_1 \leq \wbar \leq \wbar_2\right\}} \mathcal{T}_r[\psi] \dee\nu \mathrm{~d}  u+\int_{\wbar_1}^{\wbar_2}\wbar^{-2\beta}\dee\wbar.\label{eqn:conserenerIIb}
\end{align}

\noindent\textbf{Control of the boundary terms.} As in region $\un$, we have
$$\mathcal{T}_{\wbar}[\psi]\sim(-\mu)\ener[\psi], \text{   in   }\deux,$$
and
$$\mathcal{T}_{r}[\psi]\sim \ener[\psi],\text{   on   }\{r=\rb\}.$$
We also have, as in \cite{scalarMZ}
$$\mathcal{T}_\Gamma[\psi]=\widehat{\mathbf{F}}_u[\psi]+(1-\gamma\ubar^{\gamma-1})\widehat{\mathbf{F}}_{{\ubar}}[\psi]\geq f(r)|\ethat\psi|^2+g(r)|e_4\psi|^2-\mu(f(r)+g(r))(\modnab)\geq 0$$
for $\gamma>0$ small enough. Thus combining these boundary terms estimates with \eqref{eqn:conserenerIIb} yields
\begin{align}
 &\iint_{\left\{\wbar=\wbar_2\right\} \cap \mathbf{II}} \ener[\psi]  (-\mu)\dee\nu\mathrm{~d} u+\iiint_{\left\{\wbar_1 \leq \wbar\leq \wbar_2\right\} \cap \mathbf{II}} \ener[\psi] \dee\nu\:\dee u\dee\ubar\nonumber \\
 &\lesssim \iint_{\left\{\wbar=\wbar_1\right\} \cap \mathbf{II}} \ener[\psi] (-\mu)\dee\nu \mathrm{~d} u+\iint_{\{r=r_{\mathfrak{b}}\}  \cap\left\{\wbar_1 \leq \wbar \leq \wbar_2\right\}} \ener[\psi] \dee\nu \mathrm{~d} \underline{u}+\int_{\wbar_1}^{\wbar_2}\wbar^{-2\beta}\dee\wbar.\label{eqn:avantlemII}
 \end{align}
 Using the energy assumption \eqref{eqn:hypenerII} on $\{r=\rb\}$, we infer 
\begin{align}
 &\iint_{\left\{\wbar=\wbar_2\right\} \cap \mathbf{II}} \ener[\psi]  (-\mu)\dee\nu\mathrm{~d} u+\iiint_{\left\{\wbar_1 \leq \wbar\leq \wbar_2\right\} \cap \mathbf{II}} \ener[\psi] \dee\nu\:\dee u\dee\ubar\nonumber \\
 &\lesssim \iint_{\left\{\wbar=\wbar_1\right\} \cap \mathbf{II}} \ener[\psi] (-\mu)\dee\nu \mathrm{~d} u+\int_{\wbar_1}^{\wbar_2}\wbar^{-2\beta}\dee\wbar.\label{eqn:avantlemIIprim}
 \end{align}
{As in the end of the proof of Proposition \ref{prop:ener}, using the fact that for any $\wbar_1\geq 1$, the initial energy 
$$\iint_{\left\{\underline{w}=\underline{w}_1\right\} \cap \mathbf{II}} \mathbf{e}[\psi](-\mu) \dee\nu \mathrm{d} u\lesssim 1$$
is finite, we conclude the proof by applying Lemma \ref{lem:decay} with $p=2\beta$.}
\end{proof}

\begin{prop}\label{prop:decrTII}
    Assume that $\psi$ is a spin $+2$ scalar satisfying \eqref{eqn:hypborneII}, \eqref{eqn:hypenerII}, and \eqref{eqn:hypteukII} with $c>1/4$, $\beta>1$. Then we have in $\deux$, for $j\leq N_j'$ and $2k_1+k_2\leq N_k'$,
    $$\|T^j\mcq_{+2}^{k_1}\Phi^{k_2}\psi\|_{L^2({S(u,\ubar)})}\lesssim\ubar^{-\beta-j}.$$
\end{prop}
\begin{proof}
The proof is exactly the same as the one of Proposition \ref{prop:decrT}, i.e. we use Lemma \ref{lem:integ} and use the energy decay given by Proposition \ref{prop:enerdeuxx} to bound the integrated term.
\end{proof}
\begin{prop}\label{prop:propagII}
    Assume that $\psi$ is a spin $+2$ scalar satisfying \eqref{eqn:hypborneII}, \eqref{eqn:hypenerII}, and \eqref{eqn:hypteukII} with $c>1/4$, $\beta>1$. Then we have in $\deux$, for $j\leq N_j'-2$ and $2k_1+k_2\leq N_k'-2$,
    $$|T^j\mcq_{+2}^{k_1}\Phi^{k_2}\psi|\lesssim\ubar^{-\beta-j}.$$
\end{prop}
\begin{proof}
It suffices to treat the case $j=k_1=k_2=0$. Using the Sobolev embedding \eqref{eqn:sobolev} and Proposition \ref{prop:decrTII},  we get
$$|\psi|^2\lesssim \int_{S(u,\ubar)}|T^{\leq 2}\psi|^2\dee\nu+\int_{S(u,\ubar)}|\mcq_{-2}\psi|^2\dee\nu\lesssim\ubar^{-2\beta},$$
which concludes the case $j=k_1=k_2=0$.
\end{proof}

\subsection{Precise asymptotics of $\psi_{+2}$ in region $\mathbf{II}$}
We use the energy method of Section \ref{subsection:energyII} to get the following results.
\begin{prop}\label{prop:ansatzpII}
    Assume that $\Psip$ satisfies \eqref{eqn:hypp}. Then we have in $\deux$, 
    $$\psi_{+2}=\ansatzp+\errp,$$
    where for $j\leq \min(N_j^--14,N_j^+-4)$ and $2k_1+k_2\leq \min(N_k^--13,N_k^+-4)$, 
    $$|T^j\mcq_{+2}^{k_1}\Phi^{k_2}\errp|\lesssim\ubar^{-7-\delta-j}.$$
\end{prop}
The proof of Proposition \ref{prop:ansatzpII} requires the following lemma. 
\begin{lem}\label{prop:teukansatz}
    In $\deux$, for $j,k\geq 0$ we have
$$e_3^{\leq 1}\teuk_{+\fraks}T^j\mcq_{+2}^{k_1}\Phi^{k_2}\left(\frac{1}{\ubar^{7}}\sum_{|m|\leq2}A_m(r)Q_{m,\sfrak}Y_{m,\sfrak}^{+\sfrak}(\cos\theta)e^{im\phi_{+}}\right)=O(\ubar^{-8-j}). $$
\end{lem}
\begin{proof}[Proof of Lemma \ref{prop:teukansatz}]
    The computation requires the precise expression of $A_m(r)$, and the expression \eqref{prop:teukavece3} of the Teukolsky operator $\teuk_{+2}$, see Appendix \ref{appendix:am(r)} for the complete proof.
\end{proof}
\begin{proof}[Proof of Proposition \ref{prop:ansatzpII}]
By Proposition \ref{prop:ansatzIp}, we get that $\errp$ satisfies \eqref{eqn:hypborneII} with $\beta=7+\delta$. Moreover, by Corollary \ref{cor:hypener}, $\errp$ satisfies \eqref{eqn:hypenerII}. In order to use Proposition \ref{prop:propagII}, it remains to check that $\errp$ satisfies \eqref{eqn:hypteukII}. Using Lemma \ref{prop:teukansatz} and $\teuk_{+2}\Psip=0$, we get, in $\deux$, 
$$\teuk_{+\fraks}T^j\mcq_{+2}^{k_1}\Phi^{k_2}\errp=O(\ubar^{-8-j}).$$
{As $\teuk_{+2}=\teuk_{+2}^{(1,0)}$}, this proves that $\errp$ satisfies \eqref{eqn:hypteukII} with $\beta=7+\delta>1$, $c=1>1/4$ and $V=0$, thus by Proposition \ref{prop:propagII} we get in $\deux$
$$|T^j\mcq_{+2}^{k_1}\Phi^{k_2}\errp|\lesssim\ubar^{-7-\delta-j},$$
as stated.
\end{proof}
The rest of this section is devoted to obtaining $O(\ubar^{-7-\delta-j})$ pointwise decay for $e_3 T^j\mcq_{+2}^{k_1}\Phi^{k_2}\errp$ and $e_4\carterp^{k_1}\Phi^{k_2}T^j\errp$ in $\deux$. This will be used in Section \ref{section:blueshiftIII} as initial data on $\Gamma$ to integrate a $1+1$ wave equation that will eventually lead to the blow-up of $\psihatp$ on $\ch$. Unlike in region $\un$ where we already controled $\errm$ and could use TSI to get decay of $e_3\psip$, the proof in $\deux$ is done by commuting the Teukolsky equation with $e_3$ and applying the energy method of Section \ref{subsection:energyII}.

\begin{prop}\label{prop:decaynulldere3}
    Assume that $\Psip$ satisfies \eqref{eqn:hypp}. Then we have, in $\deux$, for $j\leq \min(N_j^--14,N_j^+-4)$ and $2k_1+k_2\leq \min(N_k^--13,N_k^+-4)$,
$$|T^j\mcq_{+2}^{k_1}\Phi^{k_2}e_3\errp|\lesssim\ubar^{-7-j-\delta}.$$
\end{prop}
\begin{proof}
By Proposition \ref{prop:ansatzIp}, we get that $e_3\errp$ satisfies \eqref{eqn:hypborneII} with $\beta=7+\delta$. Moreover, by Corollary \ref{cor:hypener}, $e_3\errp$ satisfies \eqref{eqn:hypenerII}. In order to use Proposition \ref{prop:propagII}, it remains to check that $e_3\errp$ satisfies \eqref{eqn:hypteukII}. Using the commutator between $\teuk_{+2}$ and $e_3$ given by Proposition \ref{prop:comme3}, and Lemma \ref{prop:teukansatz} to write
$$T^j\carterp^{k_1}\Phi^{k_2}e_3\teuk_{+2}\left(\ansatzp\right)=O(\ubar^{-8-j}),$$
we get
\begin{align*}
    \Big(\teuk_{+2}-4[(r-M)e_3&+rT]-2\Big)[T^j\carterp^{k_1}\Phi^{k_2}e_3\errp]\\
    &=T^j\carterp^{k_1}\Phi^{k_2}\left(-3T\errp-e_3\teuk_{+2}\left(\ansatzp\right)\right)\\
    &=O(\ubar^{-8-j}),
\end{align*}
where we also used\footnote{Once again, the reader interested in the count of the loss of derivatives will notice that the bound for the RHS holds only for $j\leq N_j-15$, but also for $j\leq N_j-14$ after integration on the sphere, which is the only bound that we use in the energy estimates.} Proposition \ref{prop:ansatzpII}. Using \eqref{eqn:recupc} gives
\begin{align*}
    \teuk_{+2}^{(1/2,-2)}T^j\carterp^{k_1}\Phi^{k_2}e_3\errp=O(\ubar^{-8-j}).
\end{align*}
This proves that $e_3\errp$ satisfies \eqref{eqn:hypteukII} with $\beta=7+\delta>1$, $c=1/2>1/4$ and $V=-2$, thus by Proposition \ref{prop:propagII} we get in $\deux$
$$|T^j\mcq_{+2}^{k_1}\Phi^{k_2}e_3\errp|\lesssim\ubar^{-7-j-\delta},$$
as stated.
\end{proof}
\begin{cor}\label{prop:decaynullder}
    Assume that $\Psip$ satisfies \eqref{eqn:hypp}. Then we have, in $\deux$, for $j\leq \min(N_j^--15,N_j^+-5)$ and $2k_1+k_2\leq \min(N_k^--14,N_k^+-5)$,
$$|e_4 T^j\mcq_{+2}^{k_1}\Phi^{k_2}\errp|\lesssim\ubar^{-7-\delta-j}.$$
\end{cor}
\begin{proof}
    We have $e_4=-\mu e_3+O(1)T+O(1)\Phi$, which gives the stated bound by Propositions \ref{prop:ansatzpII} and \ref{prop:decaynulldere3}. 
\end{proof}
\section{Precise asymptotics in blueshift region $\mathbf{III}\cup\quatre$ near $\ch$}\label{section:blueshiftIII}
In all this section, we assume that $\psi_{\pm 2}$ satisfy the assumptions of Theorem \ref{thm:main}. Note that we use the \emph{degenerate} Teukolsky field $\psi_{+ 2}=\Delta^2\Psihatp$ here.
\subsection{Energy estimate and pointwise bounds for $\psip$ in $\trois$}\label{subsection:boundedener}
\begin{prop}\label{prop:enerborne}
We have, for $(w_1,\wbar_1)\in\trois$, and for $j\leq \min(N_j^--12,N_j^+-2)$ and $2k_1+k_2\leq \min(N_k^--11,N_k^+-2)$, 
\begin{align*}
        \iint_{\{\wbar=\wbar_1, w\leq w_1\}\cap(\deux\cup\trois)}\ener[T^j\mcq_{+2}^{k_1}\Phi^{k_2}\psip]\dee\nu\dee r\lesssim 1.
\end{align*}
\end{prop}
\begin{proof}
As $T$ and $\mcq_{+2}$ commute with the Teukolsky equation, it suffices to treat the case $j=k=0$. We implement the same energy method as in region $\deux$ (see Section \ref{subsection:energyII}), noticing that the computations of the energy method work also in $\trois$, and we integrate \eqref{eqn:dudubarII} with\footnote{In this case, $\teuk_{+2}^{(c,V)}\psi=\teuk_{+2}\Psip=0$, so the RHS in \eqref{eqn:dudubarII} is exactly zero.} $\psi=\psip$ on $\{w\leq w_1,\:\wbar\leq\wbar_1\}\cap(\deux\cup\trois)$, see Figure \ref{fig:III}.
\begin{figure}[h!]
    \centering
    \includegraphics[scale=0.5]{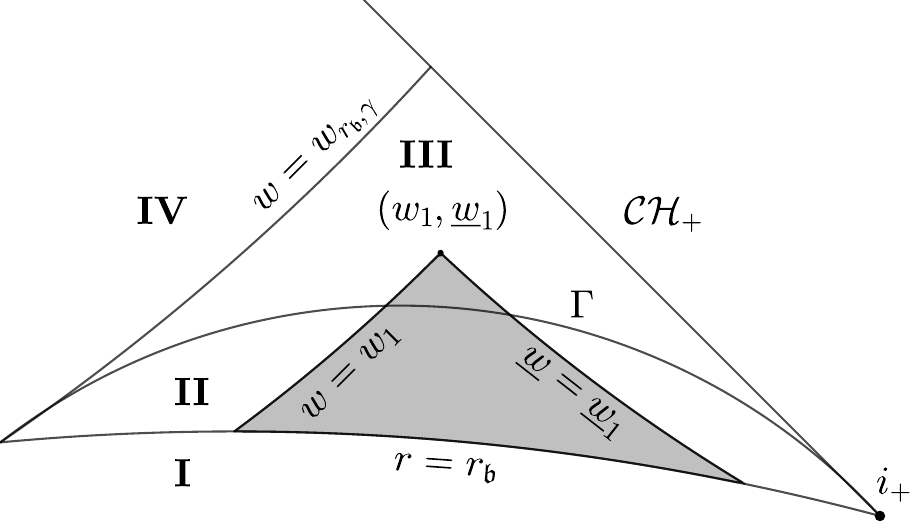}
    \caption{Region $\{w\leq w_1,\:\wbar\leq\wbar_1\}\cap(\deux\cup\trois)$ in grey.}
    \label{fig:III}
\end{figure}
We get
\begin{align}
& \iint_{\{w\leq w_1,\:\wbar=\wbar_1\}\cap(\deux\cup\trois)} \mathcal{T}_{\wbar}[\Psip] \dee\nu \mathrm{~d} u+\iint_{\{w= w_1,\:\wbar\leq\wbar_1\}\cap(\deux\cup\trois)} \mathcal{T}_{w}[\Psip] \dee\nu \mathrm{~d} u\nonumber\\
&+\iiint_{\{w\leq w_1,\:\wbar\leq\wbar_1\}\cap(\deux\cup\trois)} \widehat{\mathbf{B}}[\Psip] \dee\nu\:\dee u\dee\ubar= \iint_{\{r=\rb\}\cap\{\wbar_{\rb}(w_1)\leq\wbar\leq\wbar_1\}}\mathcal{T}_{r}[\Psip] \dee\nu \mathrm{~d} u,\label{eqn:conserenerIII}
\end{align}
where the boundary term on $\{w=w_1\}$ is
\begin{align*}
     {\mathcal{T}}_{w}[\Psip]&:=(1-\mu/2)\widehat{\mathbf{F}}_u[\Psip]-(\mu/2)\widehat{\mathbf{F}}_\ubar[\Psip]\\
     &=2(r^2+a^2)g(r)|\partial_\ubar\psi|^2-\dfrac{1}{2}\mu f(r)(|\partial_\theta\psi|^2+|\mcu\psi|^2)-a\sin\theta\mu \mathfrak{R}(\overline{X(\psi)}\mcu\psi)\\
     &\quad\quad-\frac{\mu}{2}\Bigg[2(r^2+a^2)f(r)|\ethat\psi|^2+2(r^2+a^2)g(r)|\partial_\ubar\psi|^2-\dfrac{1}{2}\mu (f(r)+g(r))(|\partial_\theta\psi|^2+|\mcu\psi|^2)\Bigg]\\
     &\sim (-\mu)\ener[\psi],
\end{align*}
where we absorbed the term $a\sin\theta\mu \mathfrak{R}(X(\psi)\mcu\psi)$ using 
$$|a\sin\theta \mu g(r) \mathfrak{R}(\overline{\partial_\ubar\psi}\mcu\psi)|\leq a^2g(r)|\partial_\ubar\psi|^2+\frac{1}{4}\mu^2g(r)|\mcu\psi|^2$$
and
$$|a\sin\theta\mu f(r) \mathfrak{R}(\overline{\ethat\psi}\mcu\psi)|\leq -\frac{3}{4}\mu a^2f(r)|\ethat\psi|^2-\frac{1}{3}\mu f(r)|\mcu\psi|^2.$$
Using \eqref{eqn:estimebordwbar} and \eqref{eqn:estimbordrb}, as well as Lemma \ref{lem:posbulkII}, this yields
\begin{align}
        &\iint_{\{w\leq w_1,\:\wbar=\wbar_1\}\cap(\deux\cup\trois)}{\ener}[\psip](-\mu)\dee\nu\dee u+\iint_{\{\wbar\leq \wbar_1,\:w=w_1\}\cap(\deux\cup\trois)}\ener[\psip](-\mu)\dee\nu\dee u\nonumber\\
        &+\iiint_{\{w\leq w_1,\:\wbar\leq\wbar_1\}\cap(\deux\cup\trois)}\ener[\psip](-\mu)\dee\nu\dee u\dee\ubar\lesssim \iint_{\{r=\rb\}\cap\{\wbar_{\rb}(w_1)\leq\wbar\leq\wbar_1\}}\ener[\psip]\dee\nu\dee\ubar.\label{eqn:similarly}
\end{align}

Using Corollary \ref{cor:hypener}, and the trivial bound
$$\ener\left[\ansatzp\right]\lesssim\ubar^{-14},$$
we get 
$$\iint_{\{r=\rb\}\cap\{\wbar_{\rb}(w_1)\leq\wbar\leq\wbar_1\}}\ener[\psip]\dee\nu\dee\ubar\lesssim \int_{\{\wbar_{\rb}(w_1)\leq\wbar\leq\wbar_1\}}\wbar^{-14}\dee\wbar\lesssim 1.$$
Since we have by changing variable :
\begin{align}\label{eqn:jrm}
    \iint_{\{w\leq w_1,\:\wbar=\wbar_1\}\cap(\deux\cup\trois)}{\ener}[\psip](-\mu)\dee\nu\dee u\sim\iint_{\{w\leq w_1,\:\wbar=\wbar_1\}\cap(\deux\cup\trois)}{\ener}[\psip]\dee\nu\dee r,
\end{align}
the conclusion of the proof of Proposition \ref{prop:enerborne} follows from \eqref{eqn:similarly} and \eqref{eqn:jrm}.
\end{proof}
\begin{prop}\label{prop:Linf}
We have in $\trois$, for $j\leq \min(N_j^--14,N_j^+-4)$ and $2k_1+k_2\leq \min(N_k^--13,N_k^+-4)$, 
$$|T^j\mcq_{+2}^{k_1}\Phi^{k_2}\psip|\lesssim 1.$$
\end{prop}
\begin{proof}
As in the proof of Proposition \ref{prop:decrT}, we use Lemma \ref{lem:integ} and use the energy decay given by Proposition \ref{prop:enerborne} to bound the integrated term, which yields
$$\|T^j\mcq_{+2}^{k_1}\Phi^{k_2}\psip\|_{L^2(S(u,\ubar))}\lesssim 1,\text{   in   }\trois.$$
Using the Sobolev embedding \eqref{eqn:sobolev}, we infer 
$$|T^j\carterp^{k_1}\Phi^{k_2}\psip|^2\lesssim \int_{S(u,\ubar)}|T^{\leq j+2}\carterp^{k_1}\Phi^{k_2}\psip|^2\dee\nu+\int_{S(u,\ubar)}|T^j\carterp^{k+1}\psip|^2\dee\nu\lesssim 1,$$
as stated.
\end{proof}
\subsection{The coupling of $\psihat_s$ to $\psi_s$ through a $1+1$ wave equation}
The goal of this section is to reformulate the Teukolsky equation $\teukhat_{+2}\psihatp=0$ as a $1+1$ wave equation for $\psihatp$ with a right-hand-side that we can control, so that we can solve explicitely for $\psihatp$ by integrating twice the equation.
\begin{prop}\label{prop:1+1}
    We have, for $j,k_1,k_2\geq 0$ and $s=\pm 2$,
\begin{align}\label{eqn:1+1}
    \ethat((r^2+a^2)\Delta^s e_4T^j\mcq_s^{k_1}\Phi^{k_2}\psihat_s)=\frac{1}{4}(-\mu)\left(\mcq_{s}+2aT\Phi-2(1+2s)rT\right)[T^j\mcq_s^{k_1}\Phi^{k_2}\psi_s].
\end{align}
\end{prop}
\begin{rem}
    Notice that for $s=+2$ we control the right-hand-side of \eqref{eqn:1+1} thanks to the $L^\infty$ bounds of Proposition \ref{prop:Linf}. The blow-up of $\psihatp$ on $\ch$ will come from the integration of \eqref{eqn:1+1}, as the inverse of the factor $\Delta^\fraks$ blows up exponentially on $\ch$.
\end{rem}
\begin{proof}[Proof of Proposition \ref{prop:1+1}]
    We have 
    \begin{align*}
        \ethat((r^2+a^2)\Delta^se_4\psihat_s)&=r\mu\Delta^s e_4\psihat_s+s\Delta^s(r-M)e_4\psihat_s+(r^2+a^2)\Delta^s \ethat e_4\psihat_s\\
        &=\frac{1}{4}\mu\Delta^s(-4(r^2+a^2)e_3e_4+4s(r-M)\mu^{-1}e_4+4re_4)[\psihat_s].
    \end{align*}
Using $e_4-\mu e_3=\mu\partial_r=2e_4-T-a/(r^2+a^2)\Phi$ we get
    \begin{align*}
        \ethat((r^2+a^2)&\Delta^s e_4\psihat_s)=\\
        &\frac{1}{4}\mu\Delta^s\Bigg(-4(r^2+a^2)e_3e_4+4s(r-M)\mu^{-1}e_4+2r(e_4-\mu e_3)+2r T+\frac{2ra}{r^2+a^2}\Phi\Bigg)[\psihat_s].
    \end{align*}
Next, using the expression \eqref{teukhat} we get
    \begin{align*}
    \teukhat_{s}=&-4(r^2+a^2)e_3e_4+ \widetilde{\mathcal{U}}^2+\frac{1}{\sin\theta}\partial_\theta(\sin\theta\partial_\theta)-4ias\cos\theta T\\
    &+2 r (e_4-\mu e_3)+4s[(r-M)\mu^{-1}e_4- rT]+\frac{2ar }{r^2+a^2}\Phi+ s,
\end{align*}
where 
$$\widetilde{\mathcal{U}}^2=\frac{1}{\sin^2\theta}\Phi^2+a^2\sin^2 T^2+2a T\Phi+2ais\cos\theta T+\frac{2is\cos\theta}{\sin^2\theta}\Phi-s^2\cot^2\theta.$$
We infer, using the definition of the Carter operator \eqref{eqn:carter},
    \begin{align*}
    \teukhat_{s}=&-4(r^2+a^2)e_3e_4+2 r (e_4-\mu e_3)+4s(r-M)\mu^{-1}e_4+\frac{2ar }{r^2+a^2}\Phi\\
    &\quad\quad+\drond'\drond+a^2\sin^2 T^2+2aT\Phi-2ais\cos\theta T-4srT\\
    &=-4(r^2+a^2)e_3e_4+2 r (e_4-\mu e_3)+4s(r-M)\mu^{-1}e_4+\frac{2ar }{r^2+a^2}\Phi\\
    &\quad\quad+\mcq_{s}+2aT\Phi-4srT.
\end{align*}
Thus combining this with the Teukolsky equation $\teukhat_{s}\psihat_s=0$ gives \eqref{eqn:1+1} with $j=k_1=k_2=0$, where we use the fact that $\Delta^s$ commutes with $\mcq_{s}$, $T$ and $\Phi$. We then extend this to general non zero $j,k_1,k_2$ by commuting with $T^j\mcq_s^{k_1}\Phi^{k_2}$.
\end{proof}

\subsection{Precise asymptotics of $\psihatp$ in $\trois$ and blow-up at $\trois\cap\ch$}
In this section we will use the crucial fact that $2r^*\geq\ubar^\gamma$ in $\trois$, thus by \eqref{eqn:rstar},
$$-\Delta\sim\exp\left(-2|\kappa_-|r^*\right)\leq\exp(-|\kappa_-|\ubar^\gamma)\quad\text{ in }\trois.$$
\begin{thm}\label{thm:psip}
    We have in $\trois$,
    \begin{align*}
    \psihatp(u,\ubar,\theta,\phi_{-})=\frac{\Delta^{-2}(u,\ubar)}{\ubar^7}\sum_{|m|\leq2}A_m(r_-)e^{2imr_{mod}(u,\ubar)}Q_{m,\sfrak}Y_{m,\sfrak}^{+\sfrak}(\cos\theta)e^{im\phi_{-}}
    +\mathrm{Err}[\psihatp],
\end{align*}
where for $j\leq \min(N_j^--15,N_j^+-5)$ and $2k_1+k_2\leq \min(N_k^--14,N_k^+-5)$,
$$|T^j\carterp^{k_1}\Phi^{k_2}\mathrm{Err}[\psihatp]|\lesssim\Delta^{-2}\ubar^{-7-j-\delta}\text{   in   }\trois.$$
\end{thm}
\begin{proof}
The proof is basically done by integrating the $1+1$ wave equation \eqref{eqn:1+1}. A bit of work is necessary at the end of the proof to get rid of the dependance in $u$ that comes from the boundary terms on $\Gamma$, i.e. to prove that the upper bound for $\mathrm{Err}[\psihatp]$ is uniform in $u$. Recall that we have shown in Proposition \ref{prop:ansatzpII}:
$$\psip=\frac{1}{\ubar^{7}}\sum_{|m|\leq2}A_m(r)Q_{m,\sfrak}Y_{m,\sfrak}^{+\sfrak}(\cos\theta)e^{im\phi_{+}}+\errp,\quad\text{   on   }\Gamma$$
where $$|e_4^{\leq 1}T^j\carterp^{k_1}\Phi^{k_2}\errp|\lesssim\ubar^{-7-j-\delta}\quad\text{  on  }\Gamma\subseteq\deux,$$
using Proposition \ref{prop:ansatzpII} and Corollary \ref{prop:decaynullder}. This implies 
\begin{align}
    e_4\psihatp=&\frac{\mu\partial_r(\Delta^{-2})}{2\ubar^7}\sum_{|m|\leq2}A_m(r)Q_{m,\sfrak}Y_{m,\sfrak}^{+\sfrak}(\cos\theta)e^{im\phi_{+}}+\frac{\Delta^{-2}}{2\ubar^7}\sum_{|m|\leq2}\mu A_m'(r)Q_{m,\sfrak}Y_{m,\sfrak}^{+\sfrak}(\cos\theta)e^{im\phi_{+}}\nonumber\\
    &+\frac{a\Delta^{-2}}{(r^2+a^2)\ubar^7}\sum_{|m|\leq2}imA_m(r)Q_{m,\sfrak}Y_{m,\sfrak}^{+\sfrak}(\cos\theta)e^{im\phi_{+}}+{\mathrm{Err}[e_4\psihatp]},\label{eqn:cpf}
\end{align}
with 
\begin{align}\label{eqn:erre4}
    T^j\carterp^{k_1}\Phi^{k_2}{\mathrm{Err}[e_4\psihatp]}=O(\Delta^{-2}\ubar^{-7-j-\delta})\quad\text{   on   }\Gamma.
\end{align}
Also, notice that the term 
$$S:=\frac{\Delta^{-2}\mu}{2\ubar^7}\sum_{|m|\leq2} A_m'(r)Q_{m,\sfrak}Y_{m,\sfrak}^{+\sfrak}(\cos\theta)e^{im\phi_{+}}$$
has an extra factor $\mu$ that decays exponentially on $\Gamma$, such that $$|\carterp^{k_1}\Phi^{k_2}T^jS|\lesssim-\Delta^{-1}\ubar^{-7-j}\lesssim \Delta^{-2}\ubar^{-7-j-\delta}.$$
We can thus add $S$ to the error term while still satisfying \eqref{eqn:erre4}, to get on $\Gamma$,
$$\Delta^2e_4\psihatp=\frac{1}{\ubar^7}\sum_{|m|\leq2}A_m(r)\left(\frac{ima-2(r-M)}{r^2+a^2}\right)Q_{m,\sfrak}Y_{m,\sfrak}^{+\sfrak}(\cos\theta)e^{im\phi_{+}}+\Delta^{2}{\mathrm{Err}[e_4\psihatp]}.$$
Next, we define 
$$Z(r,\theta,\phi_{+}):=\sum_{|m|\leq2}A_m(r)\left(\frac{ima-2(r-M)}{r^2+a^2}\right)Q_{m,\sfrak}Y_{m,\sfrak}^{+\sfrak}(\cos\theta)e^{im\phi_{+}}.$$
Then we integrate the $1+1$ wave equation \eqref{eqn:1+1} on a curve from $\Gamma$ to $(u,\ubar,\theta,\phi_{+})$ with  constant $\ubar,\theta,\phi_{+}$, using $\ethat=\partial_u$ in these coordinates. Using Proposition \ref{prop:Linf} to bound the RHS, and denoting $u_\Gamma(\ubar)=\ubar^\gamma-\ubar$, we get in $\trois$
\begin{align}
    \Bigg|(r^2+a^2)\Delta^2T^j\carterp^{k_1}\Phi^{k_2}e_4\psihatp\nonumber-T^j\carterp^{k_1}\Phi^{k_2}\frac{(r_\Gamma(\ubar)^2+a^2)}{\ubar^7}&Z(r_\Gamma(\ubar),\theta,\phi_{+})+O(\ubar^{-7-\delta-j})\Bigg|\\
    &\lesssim\int_{u_\Gamma(\ubar)}^u(-\mu)(u',\ubar)\dee u'\nonumber\\
    &\lesssim(u-u_\Gamma(\ubar))e^{-|\kappa_-|\ubar^\gamma}\lesssim\ubar^{-7-\delta-j},\label{eqn:aintegr0}
\end{align}
where we used the definition of $\trois$ to write the exponential decay of $\mu$ with $\ubar$ in $\trois$, and the fact that $u-u_\Gamma(\ubar)=u+\ubar-\ubar^\gamma\lesssim\ubar$ in $\trois$, as $u\lesssim 1$ in $\trois$. Moreover, 
\begin{align}
    &\left|\carterp^{k_1}\Phi^{k_2}\left(\frac{r_\Gamma(\ubar)^2+a^2}{r^2+a^2}\right)Z(r_\Gamma(\ubar),\theta,\phi_{+})-\carterp^{k_1}\Phi^{k_2}\sum_{|m|\leq2}A_m(r_-)\frac{ima-2(r-M)}{r^2+a^2}Q_{m,\sfrak}Y_{m,\sfrak}^{+\sfrak}(\cos\theta)e^{im\phi_{+}}\right|\nonumber\\
    &\quad\quad\quad\lesssim\sum_{|m|\leq2}|A_m(r_\Gamma(\ubar))(ima-2(r_\Gamma(\ubar)-M))-A_m(r_-)(ima-2(r-M))|\nonumber\\
    &\quad\quad\quad\lesssim r_\Gamma(\ubar)-r_-\lesssim|\mu(r_\Gamma(\ubar))|=O(\exp(-|\kappa_-|\ubar^\gamma)),\nonumber
\end{align}
in $\trois$, which together with \eqref{eqn:aintegr0} yields
\begin{align}\label{eqn:aintegr}
    T^j\carterp^{k_1}\Phi^{k_2}e_4\psihatp(u,\ubar,\theta,\phi_{+})=T^j\carterp^{k_1}\Phi^{k_2}&\frac{\Delta^{-2}(u,\ubar)}{\ubar^7}\sum_{|m|\leq2}A_m(r_-)\frac{ima-2(r-M)}{r^2+a^2}Q_{m,\sfrak}Y_{m,\sfrak}^{+\sfrak}(\cos\theta)e^{im\phi_{+}}\nonumber\\
    &+O(\Delta^{-2}\ubar^{-7-j-\delta}),\text{   in   }\trois.
\end{align}
We finally integrate \eqref{eqn:aintegr} on a curve from $\Gamma$ to $(u,\ubar,\theta,\phi_{-})$ with  constant $u,\theta,\phi_{-}$, using $e_4=\partial_\ubar$ in these coordinates, and we get using Proposition \ref{prop:ansatzpII},
\begin{align}
    &T^j\carterp^{k_1}\Phi^{k_2}\psihatp(u,\ubar,\theta,\phi_{-})\nonumber\\
    &=T^j\carterp^{k_1}\Phi^{k_2}\frac{\Delta^{-2}(r_\Gamma(u))}{\ubar_\Gamma(u)^7}\sum_{|m|\leq2}A_m(r_\Gamma(u))Q_{m,\sfrak}Y_{m,\sfrak}^{+\sfrak}(\cos\theta)e^{im\phi_{-}+2imr_{mod}(r_\Gamma(u))}+O\Bigg(\frac{\Delta^{-2}(r_\Gamma(u))}{\ubar_\Gamma(u)^{7+j+\delta}}\Bigg)\nonumber\\
    &+\int_{\ubar_\Gamma(u)}^\ubar T^j\carterp^{k_1}\Phi^{k_2}\frac{\Delta^{-2}(u,\ubar')}{(\ubar')^7}\sum_{|m|\leq2}\Big[A_m(r_-)\frac{ima-2(r'-M)}{(r')^2+a^2}Q_{m,\sfrak}Y_{m,\sfrak}^{+\sfrak}(\cos\theta)e^{im\phi_{-}+2imr_{mod}(u,\ubar')}\Big]\dee\ubar'\nonumber\\
    &+ \int_{\ubar_\Gamma(u)}^\ubar O(\Delta^{-2}(u,\ubar')(\ubar')^{-7-j-\delta})\dee\ubar',\label{eqn:yyyy}
\end{align}
{where $\ubar_{\Gamma}(u)$ is the solution to $\ubar_{\Gamma}(u)^\gamma=\ubar_{\Gamma}(u)+u$, and $r_\Gamma(u)^*=u+\ubar_{\Gamma}(u)$}. We now simplify the terms on the RHS of \eqref{eqn:yyyy}. First, {using $\Delta^{-2}=-\frac{r^2+a^2}{2(r-M)}\frac{\mu}{2}\partial_r\Delta^{-2}$, we obtain }
\begin{align*}
    &\int_{\ubar_\Gamma(u)}^\ubar\Delta^{-2}(u,\ubar')(\ubar')^{-7-j-\delta}\dee\ubar'=\int_{\ubar_\Gamma(u)}^\ubar-\frac{r^2+a^2}{2(r-M)}\frac{\dee}{\dee\ubar}(\Delta^{-2}(u,\ubar'))(\ubar')^{-7-j-\delta}\dee\ubar'\\
    &\sim\int_{\ubar_\Gamma(u)}^\ubar\frac{\dee}{\dee\ubar}(\Delta^{-2}(u,\ubar'))(\ubar')^{-7-j-\delta}\dee\ubar'\sim \Delta^{-2}(u,\ubar)\ubar^{-7-j-\delta}-\Delta^{-2}(u,\ubar_\Gamma(u))\ubar_\Gamma(u)^{-7-j-\delta}\\
    &\quad\quad\quad\quad\quad\quad\quad\quad\quad\quad\quad\quad\quad\quad\quad\quad\quad\quad\quad\quad+(7+j+\delta)\int_{\ubar_\Gamma(u)}^\ubar\Delta^{-2}(u,\ubar')(\ubar')^{-8-j-\delta}\dee\ubar',
\end{align*}
where we used the fact that $\partial_\ubar\Delta^{-2}(u,\ubar)=e_4\Delta^{-2}=(\mu/2)\partial_r\Delta^{-2}(r)$. This implies 
$$\int_{\ubar_\Gamma(u)}^\ubar\Delta^{-2}(u,\ubar')(\ubar')^{-7-j-\delta}\dee\ubar'\sim\Delta^{-2}(u,\ubar)\ubar^{-7-j-\delta}-\Delta^{-2}(u,\ubar_\Gamma(u))\ubar_\Gamma(u)^{-7-j-\delta}$$
and thus
\begin{align}\label{eqn:delta1}
    \int_{\ubar_\Gamma(u)}^\ubar O(\Delta^{-2}(u,\ubar')(\ubar')^{-7-j-\delta})\dee\ubar'=O(\Delta^{-2}(u,\ubar)\ubar^{-7-j-\delta})+O(\Delta^{-2}(u,\ubar_\Gamma(u))\ubar_\Gamma(u)^{-7-j-\delta}).
\end{align}
We now compute the second line on the RHS of \eqref{eqn:yyyy}. We have 
\begin{align}
    \int_{\ubar_\Gamma(u)}^\ubar &T^j\carterp^{k_1}\Phi^{k_2}\frac{\Delta^{-2}(u,\ubar')}{(\ubar')^7}\frac{ima-2(r-M)}{r^2+a^2}e^{2imr_{mod}(u,\ubar')}Y_{m,\sfrak}^{+\sfrak}(\cos\theta)e^{im\phi_{-}}\dee\ubar'\nonumber\\
    &=\int_{\ubar_\Gamma(u)}^\ubar\frac{\dee}{\dee\ubar}\left(\Delta^{-2}(u,\ubar')e^{2imr_{mod}(u,\ubar')}\right)T^j\carterp^{k_1}\Phi^{k_2}\frac{1}{(\ubar')^7}Y_{m,\sfrak}^{+\sfrak}(\cos\theta)e^{im\phi_{-}}\dee\ubar'\nonumber\\
    &=T^j\Bigg(\frac{\Delta^{-2}(u,\ubar)}{\ubar^7}e^{2imr_{mod}(u,\ubar)}-\frac{\Delta^{-2}(u,\ubar_\Gamma(u))}{\ubar_\Gamma(u)^{7}}e^{2imr_{mod}(u,\ubar_\Gamma(u))}\Bigg)\carterp^{k_1}\Phi^{k_2}Y_{m,\sfrak}^{+\sfrak}(\cos\theta)e^{im\phi_{-}}\nonumber\\
    &\quad\quad\quad+7\int_{\ubar_\Gamma(u)}^\ubar T^j\carterp^{k_1}\Phi^{k_2}\frac{\Delta^{-2}(u,\ubar')}{(\ubar')^8}e^{2imr_{mod}(u,\ubar')}Y_{m,\sfrak}^{+\sfrak}(\cos\theta)e^{im\phi_{-}}\dee\ubar'\nonumber\\
    &=T^j\carterp^{k_1}\Phi^{k_2}\frac{\Delta^{-2}(u,\ubar)}{\ubar^7}e^{2imr_{mod}(u,\ubar)}Y_{m,\sfrak}^{+\sfrak}(\cos\theta)e^{im\phi_{-}}\nonumber\\
    &\quad\quad\quad-T^j\carterp^{k_1}\Phi^{k_2}\frac{\Delta^{-2}(u,\ubar_\Gamma(u))}{\ubar_\Gamma(u)^{7}}e^{2imr_{mod}(r_\Gamma(u))}Y_{m,\sfrak}^{+\sfrak}(\cos\theta)e^{im\phi_{-}}\nonumber\\
    &\quad\quad\quad+O(\Delta^{-2}(u,\ubar)\ubar^{-7-j-\delta})+O(\Delta^{-2}(u,\ubar_\Gamma(u))\ubar_\Gamma(u)^{-7-j-\delta}),\label{eqn:utilehein2}
\end{align}
using the previous computation \eqref{eqn:delta1} with $\delta=1$. Next, notice that the first term on the RHS of \eqref{eqn:yyyy} can be rewritten
\begin{align}
    &T^j\carterp^{k_1}\Phi^{k_2}\frac{\Delta^{-2}(r_\Gamma(u))}{\ubar_\Gamma(u)^7}\sum_{|m|\leq2}A_m(r_\Gamma(u))Q_{m,\sfrak}Y_{m,\sfrak}^{+\sfrak}(\cos\theta)e^{im\phi_{-}+2imr_{mod}(r_\Gamma(u))}\nonumber\\
    &=T^j\carterp^{k_1}\Phi^{k_2}\frac{\Delta^{-2}(r_\Gamma(u))}{\ubar_\Gamma(u)^7}\sum_{|m|\leq2}A_m(r_-)Q_{m,\sfrak}Y_{m,\sfrak}^{+\sfrak}(\cos\theta)e^{im\phi_{-}+2imr_{mod}(r_\Gamma(u))}+O(\Delta^{-1}(r_\Gamma(u))\ubar_\Gamma(u)^{-7-j})\nonumber\\
    &=T^j\carterp^{k_1}\Phi^{k_2}\frac{\Delta^{-2}(r_\Gamma(u))}{\ubar_\Gamma(u)^7}\sum_{|m|\leq2}A_m(r_-)Q_{m,\sfrak}Y_{m,\sfrak}^{+\sfrak}(\cos\theta)e^{im\phi_{-}+2imr_{mod}(r_\Gamma(u))}+O(\Delta^{-2}(r_\Gamma(u))\ubar_\Gamma(u)^{-7-j-\delta}).\label{eqn:utilehein1}
\end{align}
Thus \eqref{eqn:yyyy} rewrites, using \eqref{eqn:delta1}, \eqref{eqn:utilehein2} and \eqref{eqn:utilehein1},
\begin{align*}
    T^j\carterp^{k_1}\Phi^{k_2}\psihat(u,\ubar,\theta,\phi_{-})&=T^j\carterp^{k_1}\Phi^{k_2}\frac{\Delta^{-2}(u,\ubar)}{\ubar^7}\sum_{|m|\leq 2}A_m(r_-)Q_{m,2}e^{2imr_{mod}(u,\ubar)}Y_{m,\sfrak}^{+\sfrak}(\cos\theta)e^{im\phi_{-}}\\
    &\quad\quad\quad+O(\Delta^{-2}(u,\ubar)\ubar^{-7-j-\delta})+O(\Delta^{-2}(u,\ubar_\Gamma(u))\ubar_\Gamma(u)^{-7-j-\delta}).
\end{align*}
We now finish the proof of Theorem \ref{thm:psip} by showing 
$$\Delta^{-2}(u,\ubar_\Gamma(u))\ubar_\Gamma(u)^{-7-j-\delta}=O(\Delta^{-2}(u,\ubar)\ubar^{-7-j-\delta})$$
in $\trois$. This requires some care. Note that in a region $\{-1\leq u\lesssim 1\}$ where $u$ is bounded from below and from above, the term that we are trying to bound is $O(1)$ which is controlled by $\Delta^{-2}(u,\ubar)\ubar^{-7-j-\delta}$. Thus we restrict the remaining part of the analysis to $\{u\leq -1\}\cap\trois$, where we want to prove
$$\frac{\Delta^{-2}(r_\Gamma(u))}{\ubar_\Gamma(u)^{7+\delta+j}}\lesssim\frac{\Delta^{-2}(u,\ubar)}{\ubar^{7+\delta+j}}.$$
Notice that $\Delta^{-2}(r_\Gamma(u))\leq\Delta^{-2}(u,\ubar)$, but that we cannot directly control $\ubar_\Gamma(u)^{-7-\delta-j}$ by $\ubar^{-7-\delta-j}$ in $\trois$, as a priori we only have $\ubar_\Gamma(u)\leq\ubar$. Let $\gamma'\in(\gamma,1)$. We write $\trois=\mathbf{A}\cup\mathbf{B}$ where 
$$\mathbf{A}:=\{2r^*\geq\ubar^{\gamma'}\}\cap\trois,\quad\mathbf{B}:=\{2r^*\leq\ubar^{\gamma'}\}\cap\trois.$$ 
\begin{itemize}
    \item In $\mathbf{A}$, we have $\Delta^{-2}(r_\Gamma(u))\sim\exp(2|\kappa_-|\ubar_\Gamma(u)^\gamma)\leq \exp(2|\kappa_-|\ubar^\gamma)$. Moreover by the definition of $\mathbf{A}$, $\Delta^{-2}(u,\ubar)\gtrsim\exp(2|\kappa_-|\ubar^{\gamma'}),$
    thus 
    $$\frac{\Delta^{-2}(r_\Gamma(u))\ubar^{7+\delta+j}}{\Delta^{-2}(u,\ubar)\ubar_\Gamma(u)^{7+\delta+j}}\lesssim\exp(2|\kappa_-|(\ubar^\gamma-\ubar^{\gamma'}))\ubar^{7+j+\delta}\lesssim 1.$$
    \item In $\mathbf{B}$, we use
    \begin{align}\label{eqn:derrr}
        \frac{\Delta^{-2}(r_\Gamma(u))}{\ubar_\Gamma(u)^{7+\delta+j}}\lesssim\frac{\Delta^{-2}(u,\ubar)}{\ubar_\Gamma(u)^{7+\delta+j}},
    \end{align}
    thus we only need to show that we can control $\ubar_\Gamma(u)^{-7-\delta-j}$ by $\ubar^{-7-\delta-j}$ there. We recall $u+\ubar_\Gamma(u)=\ubar_\Gamma(u)^\gamma$ thus $\ubar_\Gamma(u)\sim -u$ as $u\to-\infty$. Moreover in $\mathbf{B}$ we have $u+\ubar\leq\ubar^{\gamma'}$ thus $\ubar\leq\ubar^{\gamma'}-u$ and hence 
    $$\frac{-u}{\ubar}\geq 1-\ubar^{\gamma'-1}.$$
    As we are interested in the asymptotics on $\ch$, we can restrict the analysis to $\{\ubar\geq 2\}$, and we get in $\mathbf{B}$ : 
    $$\frac{-u}{\ubar}\geq 1-2^{\gamma'-1}>0.$$
    We finally reinject this back into \eqref{eqn:derrr} to get in $\mathbf{B}$ :
    \begin{align*}
        \frac{\Delta^{-2}(r_\Gamma(u))}{\ubar_\Gamma(u)^{7+\delta+j}}\lesssim\frac{\Delta^{-2}(u,\ubar)}{\ubar_\Gamma(u)^{7+\delta+j}}\lesssim\frac{\Delta^{-2}(u,\ubar)}{(-u)^{7+\delta+j}}\lesssim\frac{\Delta^{-2}(u,\ubar)}{\ubar^{7+\delta+j}}.
    \end{align*}
\end{itemize}
This concludes the proof of Theorem \ref{thm:psip}.
\end{proof}

\subsection{End of the proof of Theorem \ref{thm:main}.}
It remains to prove the asymptotic behavior \eqref{eqn:holdssss} in region $\quatre=\{\ubar\geq\ubar_{\rb,\gamma}\}\cap\{w\geq w_{\rb,\gamma}\}$, where $w_{\rb,\gamma}=2\rb^*-(2\rb^*)^{1/\gamma}-\rb+r_-$ and $\ubar_{\rb,\gamma}=(2\rb^*)^{1/\gamma}$. The first step is to notice that we can extend the bounded energy method of Section \ref{subsection:boundedener} to get non-sharp $L^\infty$ bounds for $\psip$ in $\quatre$.
\begin{lem}\label{lem:ansatzmIV}
    We have in $\quatre$, for $j\leq \min(N_j^--12,N_j^+-2)$ and $2k_1+k_2\leq \min(N_k^--11,N_k^+-2)$,
    $$|T^j\carterp^{k_1}\Phi^{k_2}\psip|\lesssim 1.$$
\end{lem}
\begin{proof}
    The proof is an extension of the argument in Section \ref{subsection:boundedener} taking into account the symmetric bounds on $\{u\geq 1\}\cap\{r=\rb\}$. Integrating \eqref{eqn:dudubarII} with\footnote{Notice that in this case, the RHS is exactly zero.} $\psi=T^j\carterp^{k_1}\Phi^{k_2}\psip$ on a triangle-shaped region $\mathcal{R}:=\{r_-<r<\rb\}\cap\{w\leq w_1\}\cap\{\wbar\leq\wbar_1\}$, with $(w_1,\wbar_1)\in\quatre$, gives similarly as in \eqref{eqn:similarly},
    \begin{align}\label{eqn:a1T}
        \iint_{\{w\leq w_1,\:\wbar=\wbar_1\}\cap\mathcal{R}}{\ener}[T^j\carterp^{k_1}\Phi^{k_2}\psip](-\mu)\dee\nu\dee u\lesssim \iint_{\partial\mathcal{R}\cap\{r=\rb\}}\ener[T^j\carterp^{k_1}\Phi^{k_2}\psip]\dee\nu\dee\ubar\lesssim 1,
    \end{align}
where we bounded the energy term on $\{r=\rb\}$ using Corollary \ref{cor:symmetric}. As before, using \eqref{eqn:a1T} together with Lemma \ref{lem:integ}, and using the initial $\ubar^{-7-j}$ bound on $\{r=\rb\}$ given by Proposition \ref{prop:ansatzIp}, we obtain $\|T^j\carterp^{k_1}\Phi^{k_2}\Psip\|_{L^2(S(u,\ubar))}\lesssim 1$ in $\quatre$. We conclude using the Sobolev embedding \eqref{eqn:sobolev}.
\end{proof}
The following result, together with Theorem \ref{thm:psip}, concludes the proof of Theorem \ref{thm:main}.
\begin{thm}\label{thm:psipIV}
    We have in $\quatre$,   
        \begin{align*}
    \psihatp(u,\ubar,\theta,\phi_{-})=\frac{\Delta^{-2}(u,\ubar)}{\ubar^7}\sum_{|m|\leq2}A_m(r_-)e^{2imr_{mod}(u,\ubar)}Q_{m,\sfrak}Y_{m,\sfrak}^{+\sfrak}(\cos\theta)e^{im\phi_{-}}
    +\mathrm{Err}[\psihatp],
\end{align*}
where for $j\leq \min(N_j^--15,N_j^+-5)$ and $2k_1+k_2\leq \min(N_k^--14,N_k^+-5)$, $$|T^j\carterp^{k_1}\Phi^{k_2}\mathrm{Err}[\psihatp]|\lesssim\Delta^{-2}\ubar^{-7-j-\delta}.
$$
\end{thm}

\begin{proof}
    Note that Theorem \ref{thm:psip} shows that this result holds on $\{w=  w_{\rb,\gamma}\}\cap\{\ubar\geq\ubar_{\rb,\gamma}\}\subset\trois$. {We will prove strong decay of the derivative $\partial_u\psip$ (see \eqref{eqn:goall}) of the \textit{degenerate} field $\psip=\Delta^2\Psihatp$. As $\partial_u$ is transverse to the hypersurfaces $\{w=cst\}$, this will infer the result in $\quatre=\{w\geq  w_{\rb,\gamma}\}\cap\{\ubar\geq\ubar_{\rb,\gamma}\}$ from the result on $\{w=  w_{\rb,\gamma}\}\cap\{\ubar\geq\ubar_{\rb,\gamma}\}$ by integration.} Using Proposition \ref{prop:comme3} 
    we find that $e_3\psip$ satisfies the PDE 
\begin{align}\label{eqn:inundd}
    \Big(\teuk_{+2}-4[(r-M)e_3-rT]-2\Big)\Big[e_3\psip\Big]=-3T\psip.
\end{align}
Thus using Lemma \ref{lem:ansatzmIV} and commuting with $T^j\carterp^{k_1}\Phi^{k_2}$ we obtain 
\begin{align}\label{eqn:inun2}
    \teuk_{+2}^{(1/2,-2)}T^j\carterp^{k_1}\Phi^{k_2}e_3\psip=O(1),\quad\text{  in  }\quatre.
\end{align}
This together with the computations of the energy method in $\deux\cup\trois$, that holds also in $\quatre$, shows that \eqref{eqn:dudubarII} holds in $\{\ubar\geq\ubar_{\rb,\gamma}\}\cap\{r_-<r\leq\rb\}$ for $\psi=T^j\carterp^{k_1}\Phi^{k_2}e_3\psip$ and $\beta=0$. Let $(u,\ubar)\in\quatre$ and denote the corresponding values $(w_1,\wbar_1)=(w_1(u,\ubar),\wbar_1(u,\ubar))$. Then, integrating \eqref{eqn:dudubarII} on $\mcr:=\{r_-<r<\rb\}\cap \{w\leq w_1\}\cap\{\wbar\leq\wbar_1\}$, with $\psi=T^j\carterp^{k_1}\Phi^{k_2}e_3\psip$, and with $\beta=0$ for the RHS gives :
\begin{align}
        &\iint_{\{w\leq w_1,\:\wbar=\wbar_1\}\cap\mcr}{\ener}[\psi](-\mu)\dee\nu\dee u+\iint_{\{\wbar\leq \wbar_1,\:w=w_1\}\cap\mcr}\ener[\psi](-\mu)\dee\nu\dee u\nonumber\\
        &+\iiint_{\{w\leq w_1,\:\wbar\leq\wbar_1\}\cap\mcr}\ener[\psi](-\mu)\dee\nu\dee u\dee\ubar\lesssim\nonumber\\
        &\iint_{\{r=\rb\}\cap\{\wbar_{\rb}(w_1)\leq\wbar\leq\wbar_1\}}\ener[\psi]\dee\nu\dee\ubar+\iiint_{\{w\leq w_1,\:\wbar\leq\wbar_1\}\cap\mcr}|X(\psi)|O(-\mu)\dee\nu\dee u\dee\ubar,\label{eqn:derlabel}
\end{align}
where we also used \eqref{eqn:estimebordwbar} and \eqref{eqn:estimbordrb}, as well as Lemma \ref{lem:posbulkII}.
Recall that, $|X(\psi)|^2\lesssim\ener[\psi]$ thus using a Cauchy-Schwarz inequality we can bound the last term on the RHS of \eqref{eqn:derlabel} by
\begin{align}\label{eqn:firstterm}
    \varepsilon\iiint_{\{w\leq w_1,\:\wbar\leq\wbar_1\}\cap\mcr}\ener[\psi](-\mu)\dee\nu\dee u\dee\ubar+\varepsilon^{-1}\iiint_{\{w\leq w_1,\:\wbar\leq\wbar_1\}\cap\mcr}O(-\mu)\dee\nu\dee u\dee\ubar.
\end{align}
Choosing $\varepsilon>0$ small enough so that the first term of \eqref{eqn:firstterm} gets absorbed on the LHS of \eqref{eqn:derlabel}. Moreover, 
\begin{align*}
    \iiint_{\{w\leq w_1,\:\wbar\leq\wbar_1\}\cap\mcr}(-\mu)\dee\nu\dee u\dee\ubar&=4\pi\int_{\ubar_{\rb}(w_1)}^{\ubar_{\rb}(\wbar_1)}\int_{u_{\rb}(\ubar)}^{u(\ubar,w_1)}(-\mu)\dee u\dee \ubar\\
    &\sim4\pi\int_{\ubar_{\rb}(w_1)}^{\ubar_{\rb}(\wbar_1)}\int_{\rb}^{r(\ubar,w_1)}\dee r\dee \ubar\\
    &\lesssim \ubar_{\rb}(\wbar_1)-\ubar_{\rb}(w_1)\\
    &\lesssim w_1+\wbar_1+K,
\end{align*}
where $K(a,M)>0$ is a constant, and where we used 
\begin{align*}
    \ubar_{\rb}(\wbar_1)&=\wbar_1+\rb-r_+,\\
    \ubar_{\rb}(w_1)&=2\rb^*-(w_1+\rb-r_-)=-w_1+cst.
\end{align*}
Using Corollary \ref{cor:symmetric} we also have the bound 
$$\iint_{\{r=\rb\}\cap\{\wbar_{\rb}(w_1)\leq\wbar\leq\wbar_1\}}\ener[T^j\carterp^{k_1}\Phi^{k_2}e_3\psip]\dee\nu\dee\ubar\lesssim 1.$$
Thus we have shown, { using $w=u+O(1)$, $\wbar=\ubar+O(1)$, and $K\lesssim 1\lesssim 2r^*=u+\ubar$,}
\begin{align}\label{eqn:toitoi}
    \iint_{\{w\leq w_1,\:\wbar=\wbar_1\}\cap\mcr}{\ener}[T^j\carterp^{k_1}\Phi^{k_2}e_3\psip](-\mu)\dee\nu\dee u\lesssim u+\ubar.
\end{align}
Using Lemma \ref{lem:integ} yields\footnote{Notice that $u+\ubar=2r^*>0$ in $\{r_-\leq r\leq\rb\}$.}
\begin{align*}
    \|T^j\carterp^{k_1}\Phi^{k_2}e_3\psip\|_{L^2(S^2(u,\ubar))}\lesssim & \|T^j\carterp^{k_1}\Phi^{k_2}e_3\psip\|_{L^2(S^2(u(\rb,\wbar_1),\ubar(\rb,\wbar_1)))}\\
    &+\left(\iint_{\{w\leq w_1,\:\wbar=\wbar_1\}\cap\mcr}{\ener}[T^j\carterp^{k_1}\Phi^{k_2}e_3\psip](-\mu)\dee\nu\dee u\right)^{1/2}\\
    &\lesssim \sqrt{u+\ubar}\quad\text{   in   }\quatre,
\end{align*}
where we used \eqref{eqn:toitoi} and Proposition \ref{prop:ansatzIp} to bound the term on $\{r=\rb\}$. Using the Sobolev embedding \eqref{eqn:sobolev} finally gives
$$|T^j\carterp^{k_1}\Phi^{k_2}e_3\psip|(u,\wbar_1,\theta,\phi_{-})\lesssim \sqrt{u+\ubar}\lesssim u+\ubar,\quad (\quatre).$$
We will now use the fact that $\ethat=\partial_u$ in coordinates $(u,\ubar,\theta,\phi_{+})$, and integrate the previous estimate on $\ubar=cst$ from $\{w= w_{\rb,\gamma}\}$ to $w(u,\ubar)$ to get information in $\quatre$ from the lower bound that we have on $\{w= w_{\rb,\gamma}\}$ from the Theorem \ref{thm:psip}. We have the estimate 
\begin{align}\label{eqn:goall}
    |\ethat T^j\carterp^{k_1}\Phi^{k_2}\psip|(u,\ubar,\theta,\phi_{-})\lesssim -(u+\ubar)\Delta,\quad (\quatre).
\end{align}
Thus integrating on $\ubar=cst,\theta=cst,\phi_{+}=cst$, we get 
\begin{align*}
    T^j\carterp^{k_1}\Phi^{k_2}\psip(u,\ubar,\theta,\phi_{-})=T^j\carterp^{k_1}\Phi^{k_2}\psip(u(w_{\rb,\gamma},&\ubar),\ubar,\theta,\phi_{-}|_{w=w_{\rb,\gamma}})\\
    &+\int_{u(w_{\rb,\gamma},\ubar)}^uO(-\Delta(u',\ubar)(u'+\ubar))\dee u',
\end{align*}
where $\phi_{-}|_{w=w_{\rb,\gamma}}=\phi_{-}+2r_{mod}(\ubar,u(w_{\rb,\gamma},\ubar))-2r_{mod}(\ubar,u)$. Using Theorem \ref{thm:psip} on $\{w=  w_{\rb,\gamma}\}\cap\{\ubar\geq\ubar_{\rb,\gamma}\}$ we obtain 
\begin{align*}
    T^j\carterp^{k_1}\Phi^{k_2}\psihatp(u,\ubar,\theta,\phi_{-})&=T^j\carterp^{k_1}\Phi^{k_2}\frac{\Delta^{-2}(u,\ubar)}{\ubar^7}\sum_{|m|\leq2}A_m(r_-)e^{2imr_{mod}(u,\ubar)}Q_{m,\sfrak}Y_{m,\sfrak}^{+\sfrak}(\cos\theta)e^{im\phi_{-}}\\
    &+O(\Delta^{-2}(u,\ubar)\ubar^{-7-j-\delta})+\Delta^{-2}(u,\ubar)\int_{u(w_{\rb,\gamma},\ubar)}^uO(-\Delta(u',\ubar)(u'+\ubar))\dee u'.
\end{align*}
We conclude by proving that in $\quatre$, 
$$\int_{u(w_{\rb,\gamma},\ubar)}^uO(-\Delta(u',\ubar)(u'+\ubar))\dee u'=O(\ubar^{-7-j-\delta}).$$
We have in $\quatre$, $-\Delta(u,\ubar)\sim\exp(-|\kappa_-|(u+\ubar))$. Thus 
\begin{align*}
    \Bigg|\int_{u(w_{\rb,\gamma},\ubar)}^u&O(-\Delta(u',\ubar)(u'+\ubar))\dee u'\Bigg|\\
    &\lesssim \exp(-|\kappa_-|\ubar)\left[\int_{u(w_{\rb,\gamma},\ubar)}^u |u'|\exp(-|\kappa_-|u')\dee u'+\ubar\int_{u(w_{\rb,\gamma},\ubar)}^u \exp(-|\kappa_-|u')\dee u'\right]\\
    &\lesssim \exp(-|\kappa_-|\ubar)(C_1+\ubar C_2)\\
    &\lesssim\ubar^{-7-j-\delta}
\end{align*}
in $\quatre$, where we defined
$$C_1=\int_{w_{\rb,\gamma}-r_++r_-}^{+\infty} |u'|\exp(-|\kappa_-|u')\dee u', \quad C_2=\int_{w_{\rb,\gamma}-r_++r_-}^{+\infty} \exp(-|\kappa_-|u')\dee u'.$$
This concludes the proof of Theorem \ref{thm:psipIV}.
\end{proof}

\appendix

\section{Computations for the energy method}
\subsection{Integration by parts on the sphere}\label{appendix:computationsI}

\begin{lem}\label{lem:integS}
The computation of the left-hand side of \eqref{eqn:teukmult} gives : 
\begin{align}\label{eqn:conservationener}
    \partial_\ubar\left(\int_{S(u,\ubar)}\mathbf{F}_\ubar[\psi]\dee\nu\right)+\partial_u\left(\int_{S(u,\ubar)}\mathbf{F}_u[\psi]\dee\nu\right)+\int_{S(u,\ubar)}\mathbf{B}[\psi]\dee\nu=\int_{S(u,\ubar)}\mu\Real(\overline{X(\psi)}O(\ubar^{-\beta}))\dee\nu,
\end{align}
where $$\mathbf{F}_\ubar[\psi]:=2(r^2+a^2)f(r)|\partial_u\psi|^2-\dfrac{1}{2}\mu g(r)(|\partial_\theta\psi|^2+|\mathcal{U}\psi|^2)+a\sin\theta f(r)\mu \mathfrak{R}(\overline{\partial_u\psi}\mathcal{U}\psi)+a\sin\theta g(r)\mu \mathfrak{R}(\overline{e_4\psi}\mathcal{U}\psi),$$
$$\mathbf{F}_u[\psi]:=2(r^2+a^2)g(r)|e_4\psi|^2-\dfrac{1}{2}\mu f(r)(|\partial_\theta\psi|^2+|\mathcal{U}\psi|^2)-a\sin\theta f(r)\mu \mathfrak{R}(\overline{\partial_u\psi}\mathcal{U}\psi)-a\sin\theta g(r)\mu \mathfrak{R}(\overline{e_4\psi}\mathcal{U}\psi),$$
and the bulk term is
\begin{align*}
    \mathbf{B}[\psi]:&=2(r\mu g(r)-\partial_u((r^2+a^2)g(r)))|e_4\psi|^2+2(r\mu f(r)-e_4((r^2+a^2)f(r)))|\partial_u\psi|^2\\
    &\quad +\dfrac{1}{2}(\partial_u(\mu f(r))+e_4(\mu g(r)))(|\partial_\theta\psi|^2+|\mathcal{U}\psi|^2)\\
    &\quad+\mu g(r)\frac{sra^2\mu\cos\theta\sin\theta}{(r^2+a^2)^2}\mathfrak{I}(\overline{\psi} \mathcal{U}\psi)-4cs(r-M)f(r)|\partial_u\psi|^2\\
    &\quad+4\mu g(r)as\cos\theta\mathfrak{I}(\overline{e_4\psi}T\psi)-4srg(r)\mu\mathfrak{R}(\overline{e_4\psi}T\psi)+g(r)[2r\mu-4cs(r-M)]\mathfrak{R}(\overline{e_4\psi}\partial_u\psi)\\
    &\quad+\dfrac{2g(r)ar\mu}{r^2+a^2}\mathfrak{R}(\overline{e_4\psi}\Phi\psi)-g(r)\left(\mu s+s^2\mu\dfrac{a^4\sin^2\theta\cos^2\theta}{(r^2+a^2)^2}-\mu V\right)\mathfrak{R}(\overline{e_4\psi}\psi)\\
    &\quad-2sg(r)\mu\dfrac{a^2\sin\theta\cos\theta}{(r^2+a^2)}\mathfrak{I}(\overline{e_4\psi}\mathcal{U}\psi)+\mu f(r)\frac{sra^2\mu\cos\theta\sin\theta}{(r^2+a^2)^2}\mathfrak{I}(\overline{\psi} \mathcal{U}\psi)\\
    &\quad+4\mu f(r)as\cos\theta\mathfrak{I}(\overline{\partial_u\psi}T\psi)-4srf(r)\mu\mathfrak{R}(\overline{\partial_u\psi}T\psi)+2r\mu f(r)\mathfrak{R}(\overline{\partial_u\psi}e_4\psi)\\
    &\quad-\dfrac{2f(r)ar\mu}{r^2+a^2}\mathfrak{R}(\overline{\partial_u\psi}\Phi\psi)-f(r)\left(\mu s+s^2\mu\dfrac{a^4\sin^2\theta\cos^2\theta}{(r^2+a^2)^2}-\mu V\right)\mathfrak{R}(\overline{\partial_u\psi}\psi)\\
    &\quad-2sf(r)\mu\dfrac{a^2\sin\theta\cos\theta}{(r^2+a^2)}\mathfrak{I}(\overline{\partial_u\psi}\mathcal{U}\psi).
\end{align*}
\end{lem}
\begin{rem}
    Setting $s=0$ in the LHS of \eqref{eqn:conservationener}, we find the same expression as in \cite[(3.5c)]{scalarMZ} for the scalar wave.
\end{rem}
\begin{proof}[Proof of Lemma \ref{lem:integS}]
In $\un$, we compute 
$$\int_{{S(u,\ubar)}}\Real(f(r)\overline{\partial_u\psi}\mu\teuk_s^{(c,V)}{\psi})+\Real(g(r)\overline{e_4\psi}\mu\teuk_s^{(c,V)}\psi)\dee\nu$$
using integration by parts on the spheres. We have using \eqref{eqn:muteuk1} :
\begin{align*}
    \int_{{S(u,\ubar)}}\Real(g(r)\overline{e_4\psi}&\mu\teuk_s^{(c,V)}\psi)\:\dee\nu\\
    =&4(r^2+a^2)g(r)\int_{{S(u,\ubar)}}\mathfrak{R}(\overline{e_4\psi}\partial_u e_4\psi)\:\dee\nu+\mu g(r)\int_{{S(u,\ubar)}}\mathfrak{R}(\overline{e_4\psi} \mathcal{U}^2\psi)\:\dee\nu\\
    &+\mu g(r)\int_{{S(u,\ubar)}}\mathfrak{R}\left(\overline{e_4\psi}\dfrac{1}{\sin\theta}\partial_\theta(\sin\theta\partial_\theta\psi)\right)\dee\nu\\
    &+4\mu g(r) as \int_{{S(u,\ubar)}}\cos\theta\mathfrak{I}(\overline{e_4\psi} T\psi)\dee\nu\\
    &+2r\mu g(r)\int_{{S(u,\ubar)}}|e_4\psi|^2\dee\nu+ g(r) [2r\mu -4cs(r-M)]\int_{{S(u,\ubar)}}\mathfrak{R}(\overline{e_4\psi}\partial_u\psi)\dee\nu\\
    &-4sr\mu g(r)\int_{{S(u,\ubar)}}\mathfrak{R}(\overline{e_4\psi}T\psi)\dee\nu+\dfrac{2g(r)ar\mu }{r^2+a^2}\int_{{S(u,\ubar)}}\mathfrak{R}(\overline{e_4\psi}\Phi\psi)\dee\nu\\
    &-g(r)\left(\mu s+s^2\mu\dfrac{a^4\sin^2\theta\cos^2\theta}{(r^2+a^2)^2}-\mu V\right)\int_{{S(u,\ubar)}}\mathfrak{R}(\overline{e_4\psi}\psi)\dee\nu\\
    &-2sg(r)\mu\dfrac{a^2\sin\theta\cos\theta}{(r^2+a^2)}\int_{{S(u,\ubar)}}\mathfrak{I}(\overline{e_4\psi}\mathcal{U}\psi)\dee\nu.
\end{align*}
We begin with the term 
\begin{align*}
    4(r^2+a^2)g(r)&\int_{{S(u,\ubar)}}\mathfrak{R}(\overline{e_4\psi}\partial_u e_4\psi)\:\dee\nu=2(r^2+a^2)g(r)\partial_u\left(\int_{{S(u,\ubar)}} |e_4\psi|^2\:\dee\nu\right)\\
    &=\partial_u\left(2(r^2+a^2)g(r)\int_{{S(u,\ubar)}} |e_4\psi|^2\:\dee\nu\right)-2\partial_u((r^2+a^2)g(r))\int_{{S(u,\ubar)}} |e_4\psi|^2\:\dee\nu.
\end{align*}
Using Lemma \ref{lem:intpartU} and
$$[e_4,\mathcal{U}]=e_4\left(is\cot\theta\frac{\Sigma}{r^2+a^2}\right)=\frac{isra^2\mu\cos\theta\sin\theta}{(r^2+a^2)^2},$$
we get, in view of $T=\partial_\ubar-\partial_u$,
\begin{align*}
    \mu g(r)\int_{{S(u,\ubar)}}\mathfrak{R}(&\overline{e_4\psi} \mathcal{U}^2\psi)\:\dee\nu\\
    &=T\left(\mu g(r)\int_{{S(u,\ubar)}}a\sin\theta\mathfrak{R}(\overline{e_4\psi}\mathcal{U}\psi)\:\dee\nu\right)-\mu g(r)\int_{{S(u,\ubar)}}\mathfrak{R}(\overline{\mathcal{U}e_4\psi} \mathcal{U}\psi)\:\dee\nu\\
    &=(\partial_\ubar-\partial_u)\left(\mu g(r)\int_{{S(u,\ubar)}}a\sin\theta\mathfrak{R}(\overline{e_4\psi}\mathcal{U}\psi)\:\dee\nu\right)-\dfrac{1}{2}\partial_\ubar\left(\mu g(r)\int_{{S(u,\ubar)}}|\mathcal{U}\psi|^2\:\dee\nu\right)\\
    &\quad\quad+\dfrac{1}{2}e_4(\mu g(r))\int_{{S(u,\ubar)}}|\mathcal{U}\psi|^2\:\dee\nu+\mu g(r)\frac{sra^2\mu\cos\theta\sin\theta}{(r^2+a^2)^2}\int_{{S(u,\ubar)}}\mathfrak{I}(\overline{\psi} \mathcal{U}\psi)\:\dee\nu.
\end{align*}
We also have, using $[e_4,\partial_\theta]=0$,
\begin{align*}
    \mu g(r)\int_{{S(u,\ubar)}}\mathfrak{R}\Bigg(\overline{e_4\psi}\dfrac{1}{\sin\theta}\partial_\theta&(\sin\theta\partial_\theta\psi)\Bigg)\dee\nu\\
    &=-\mu g(r)\int_{{S(u,\ubar)}}\mathfrak{R}\left(\partial_\theta\overline{e_4\psi}\partial_\theta\psi\right)\dee\nu\\
    &=-\dfrac{1}{2}\mu g(r)\int_{{S(u,\ubar)}}e_4|\partial_\theta\psi|^2\dee\nu\\
    &=\partial_\ubar\left(-\dfrac{1}{2}\mu g(r)\int_{{S(u,\ubar)}}|\partial_\theta\psi|^2\dee\nu\right)+\dfrac{1}{2}e_4(\mu g(r))\int_{{S(u,\ubar)}}|\partial_\theta\psi|^2\dee\nu.
\end{align*}
All the remaing terms will be put in the bulk term $\mathbf{B}[\psi]$. Now we do the same computations for 
\begin{align*}
    \int_{{S(u,\ubar)}}\mathfrak{R}(f(r)\overline{\partial_u\psi}&\mu\teuk_s^{(c,V)}\psi)\:\dee\nu\\
    =&4(r^2+a^2)f(r)\int_{{S(u,\ubar)}}\mathfrak{R}(\overline{\partial_u\psi} e_4\partial_u\psi)\:\dee\nu+\mu f(r)\int_{{S(u,\ubar)}} \mathfrak{R}(\overline{\partial_u\psi} \mathcal{U}^2\psi)\:\dee\nu\\
    &+\mu f(r)\int_{{S(u,\ubar)}}\mathfrak{R}\left(\overline{\partial_u\psi}\dfrac{1}{\sin\theta}\partial_\theta(\sin\theta\partial_\theta\psi)\right)\dee\nu\\
    &+4\mu f(r) as \int_{{S(u,\ubar)}}\cos\theta\mathfrak{I}(\overline{\partial_u\psi} T\psi)\dee\nu\\
    &+f(r) [2r\mu -4cs(r-M)]\int_{{S(u,\ubar)}}|\partial_u\psi|^2\dee\nu+ 2r\mu f(r)\int_{{S(u,\ubar)}}\mathfrak{R}(\overline{\partial_u\psi}e_4\psi)\dee\nu\\
    &-4sr\mu f(r)\int_{{S(u,\ubar)}}\mathfrak{R}(\overline{\partial_u\psi}T\psi)\dee\nu-\dfrac{2f(r)ar\mu }{r^2+a^2}\int_{{S(u,\ubar)}}\mathfrak{R}(\overline{\partial_u\psi}\Phi\psi)\dee\nu\\
    &-f(r)\left(\mu s+s^2\mu\dfrac{a^4\sin^2\theta\cos^2\theta}{(r^2+a^2)^2}-\mu V\right)\int_{{S(u,\ubar)}}\mathfrak{R}(\overline{\partial_u\psi}\psi)\dee\nu\\
    &-2sf(r)\mu\dfrac{a^2\sin\theta\cos\theta}{(r^2+a^2)}\int_{{S(u,\ubar)}}\mathfrak{I}(\overline{\partial_u\psi}\mathcal{U}\psi)\dee\nu.
\end{align*}
We begin with the term 
\begin{align*}
    4(r^2+a^2)f(r)&\int_{{S(u,\ubar)}}\mathfrak{R}(\overline{\partial_u\psi} e_4\partial_u\psi)\:\dee\nu=2(r^2+a^2)f(r)e_4\left(\int_{{S(u,\ubar)}} |\partial_u\psi|^2\:\dee\nu\right)\\
    &=\partial_\ubar\left(2(r^2+a^2)f(r)\int_{{S(u,\ubar)}} |\partial_u\psi|^2\:\dee\nu\right)-2e_4((r^2+a^2)f(r))\int_{{S(u,\ubar)}} |\partial_u\psi|^2\:\dee\nu.
\end{align*}
Next we have
\begin{align*}
    \mu f(r)\int_{{S(u,\ubar)}} \mathfrak{R}&(\overline{\partial_u\psi} \mathcal{U}^2\psi)\:\dee\nu\\
    &=T\left(\mu f(r)\int_{{S(u,\ubar)}}a\sin\theta\mathfrak{R}(\overline{\partial_u\psi}\mathcal{U}\psi)\:\dee\nu\right)-\mu f(r)\int_{{S(u,\ubar)}}\mathfrak{R}(\overline{\mathcal{U}\partial_u\psi} \mathcal{U}\psi)\:\dee\nu\\
    &=(\partial_\ubar-\partial_u)\left(\mu f(r)\int_{{S(u,\ubar)}}a\sin\theta\mathfrak{R}(\overline{\partial_u\psi}\mathcal{U}\psi)\:\dee\nu\right)-\dfrac{1}{2}\partial_u\left(\mu f(r)\int_{{S(u,\ubar)}}|\mathcal{U}\psi|^2\:\dee\nu\right)\\
    &\quad\quad+\dfrac{1}{2}\partial_u(\mu f(r))\int_{{S(u,\ubar)}}|\mathcal{U}\psi|^2\:\dee\nu+\mu f(r)\frac{sra^2\mu\cos\theta\sin\theta}{(r^2+a^2)^2}\int_{{S(u,\ubar)}}\mathfrak{I}(\overline{\psi} \mathcal{U}\psi)\:\dee\nu,
\end{align*}
\begin{align*}
    \mu f(r)\int_{{S(u,\ubar)}}\mathfrak{R}\Bigg(\overline{\partial_u\psi}\dfrac{1}{\sin\theta}\partial_\theta&(\sin\theta\partial_\theta\psi)\Bigg)\dee\nu\\
    &=-\mu f(r)\int_{{S(u,\ubar)}}\mathfrak{R}\left(\partial_\theta\overline{\partial_u\psi}\partial_\theta\psi\right)\dee\nu\\
    &=-\dfrac{1}{2}\mu f(r)\int_{{S(u,\ubar)}}\partial_u|\partial_\theta\psi|^2\dee\nu\\
    &=\partial_u\left(-\dfrac{1}{2}\mu f(r)\int_{{S(u,\ubar)}}|\partial_\theta\psi|^2\dee\nu\right)+\dfrac{1}{2}\partial_u(\mu f(r))\int_{{S(u,\ubar)}}|\partial_\theta\psi|^2\dee\nu.
\end{align*}
Combining everything, \eqref{eqn:teukmult} gives 
$$\partial_\ubar\left(\int_{{S(u,\ubar)}}\mathbf{F}_\ubar[\psi]\dee\nu\right)+\partial_u\left(\int_{{S(u,\ubar)}}\mathbf{F}_u[\psi]\dee\nu\right)+\int_{{S(u,\ubar)}}\mathbf{B}[\psi]\dee\nu=\int_{S(u,\ubar)}\mu\Real(\overline{X(\psi)}O(\ubar^{-\beta}))\dee\nu,$$
where 
$$\mathbf{F}_\ubar[\psi]=2(r^2+a^2)f(r)|\partial_u\psi|^2-\dfrac{1}{2}\mu g(r)(|\partial_\theta\psi|^2+|\mathcal{U}\psi|^2)+a\sin\theta f(r)\mu \mathfrak{R}(\overline{\partial_u\psi}\mathcal{U}\psi)+a\sin\theta g(r)\mu \mathfrak{R}(\overline{e_4\psi}\mathcal{U}\psi),$$
$$\mathbf{F}_u[\psi]=2(r^2+a^2)g(r)|e_4\psi|^2-\dfrac{1}{2}\mu f(r)(|\partial_\theta\psi|^2+|\mathcal{U}\psi|^2)-a\sin\theta f(r)\mu \mathfrak{R}(\overline{\partial_u\psi}\mathcal{U}\psi)-a\sin\theta g(r)\mu \mathfrak{R}(\overline{e_4\psi}\mathcal{U}\psi),$$
and 
\begin{align*}
    \mathbf{B}[\psi]&=2(r\mu g(r)-\partial_u((r^2+a^2)g(r)))|e_4\psi|^2+2(r\mu f(r)-e_4((r^2+a^2)f(r)))|\partial_u\psi|^2\\
    &\quad +\dfrac{1}{2}(\partial_u(\mu f(r))+e_4(\mu g(r)))(|\partial_\theta\psi|^2+|\mathcal{U}\psi|^2)\\
    &\quad+\mu g(r)\frac{sra^2\mu\cos\theta\sin\theta}{(r^2+a^2)^2}\mathfrak{I}(\overline{\psi} \mathcal{U}\psi)-4cs(r-M)f(r)|\partial_u\psi|^2\\
    &\quad+4\mu g(r)as\cos\theta\mathfrak{I}(\overline{e_4\psi}T\psi)-4srg(r)\mu\mathfrak{R}(\overline{e_4\psi}T\psi)+g(r)[2r\mu-4cs(r-M)]\mathfrak{R}(\overline{e_4\psi}\partial_u\psi)\\
    &\quad+\dfrac{2g(r)ar\mu}{r^2+a^2}\mathfrak{R}(\overline{e_4\psi}\Phi\psi)-g(r)\left(\mu s+s^2\mu\dfrac{a^4\sin^2\theta\cos^2\theta}{(r^2+a^2)^2}-\mu V\right)\mathfrak{R}(\overline{e_4\psi}\psi)\\
    &\quad-2sg(r)\mu\dfrac{a^2\sin\theta\cos\theta}{(r^2+a^2)}\mathfrak{I}(\overline{e_4\psi}\mathcal{U}\psi)+\mu f(r)\frac{sra^2\mu\cos\theta\sin\theta}{(r^2+a^2)^2}\mathfrak{I}(\overline{\psi} \mathcal{U}\psi)\\
    &\quad+4\mu f(r)as\cos\theta\mathfrak{I}(\overline{\partial_u\psi}T\psi)-4srf(r)\mu\mathfrak{R}(\overline{\partial_u\psi}T\psi)+2r\mu f(r)\mathfrak{R}(\overline{\partial_u\psi}e_4\psi)\\
    &\quad-\dfrac{2f(r)ar\mu}{r^2+a^2}\mathfrak{R}(\overline{\partial_u\psi}\Phi\psi)-f(r)\left(\mu s+s^2\mu\dfrac{a^4\sin^2\theta\cos^2\theta}{(r^2+a^2)^2}-\mu V\right)\mathfrak{R}(\overline{\partial_u\psi}\psi)\\
    &\quad-2sf(r)\mu\dfrac{a^2\sin\theta\cos\theta}{(r^2+a^2)}\mathfrak{I}(\overline{\partial_u\psi}\mathcal{U}\psi),
\end{align*}
which concludes the proof of Lemma \ref{lem:integS}.
\end{proof}
\subsection{Lower bound for the bulk term in $\un$ for $s=-2$}\label{appendix:bulkI}
\begin{lem}\label{lem:bulkpos}
    For $c>0$, $s=-2$, and for $p=p(a,M,c,V)$ chosen large enough, we have in $\un$
    \begin{align}\label{eq:bulkposI}
    \int_{S(u,\ubar)}\mathbf{B}[\psi]\dee\nu\gtrsim (-\mu)\int_{S(u,\ubar)}\ener[\psi]\dee\nu.
\end{align}
\end{lem}
\begin{proof}
We have 
$$\begin{aligned}
&r \mu g(r)-\partial_u\left(\left(r^2+a^2\right) g(r)\right)  =-\mu p r\left(r^2+a^2\right)^p, \\
&r \mu f(r)-e_4\left(\left(r^2+a^2\right) f(r)\right)  =(-\mu)^{-1}[(r-M)-\mu(p+1) r]\left(r^2+a^2\right)^p\\
&\quad\quad\quad\quad\quad\quad\quad\quad\quad\quad\quad\quad\quad \gtrsim (-\mu)^{-1}(1-\mu p)(r^2+a^2)^p,
\end{aligned}$$
in $\un$ for $p$ large, as $r_+-M>0$. We also have
$$\partial_u(\mu f(r))+e_4(\mu g(r))  =-\mu\left(r^2+a^2\right)^{p-1}\left(\dfrac{2 M p r^2}{r^2+a^2}+r \mu-(r-M)\right)\gtrsim -\mu p(r^2+a^2)^p.
$$ 
Denoting the principal bulk term
\begin{align*}
    \mathbf{B}_{pr}[\psi]:=2(r\mu g(&r)-\partial_u((r^2+a^2)g(r)))|e_4\psi|^2+2(r\mu f(r)-e_4((r^2+a^2)f(r)))|\partial_u\psi|^2\\
    &\quad +\dfrac{1}{2}(\partial_u(\mu f(r))+e_4(\mu g(r)))(|\partial_\theta\psi|^2+|\mathcal{U}\psi|^2),
\end{align*}
we have shown 
\begin{align}\label{eqn:bpr}
    \mathbf{B}_{pr}[\psi]\gtrsim (-\mu)(r^2+a^2)^p[p|e_4\psi|^2+(1-\mu p)|e_3\psi|^2+p(|\partial_\theta\psi|^2+|\mathcal{U}\psi|^2)].
\end{align}
The only thing left to prove is that we can take $p$ large enough so that $\mathbf{B}[\psi]-\mathbf{B}_{pr}[\psi]$ can be absorbed in $\mathbf{B}_{pr}[\psi]$ after integrating on ${S(u,\ubar)}$. This is due to the following mix between weighted Cauchy-Schwarz of the type 
$$|ab|\leq\frac{\varepsilon a^2+\varepsilon^{-1}b^2}{2}$$
and the Poincaré inequality \eqref{eqn:poincare}. We have the following bounds : 
\begin{align*}
    \left|\intS \mu g(r)\frac{sra^2\mu\cos\theta\sin\theta}{(r^2+a^2)^2}\mathfrak{I}(\overline{\psi} \mathcal{U}\psi)\dee\nu\right|&\lesssim (-\mu)(r^2+a^2)^p\left(\intS\ener_{deg}[\psi]\dee\nu+\intS|\mathcal{U}\psi|^2\dee\nu\right)\\
    &\lesssim(-\mu)(r^2+a^2)^p\intS\ener_{deg}[\psi]\dee\nu,
\end{align*}
$$\left|\intS4\mu g(r)as\cos\theta\mathfrak{I}(\overline{e_4\psi}T\psi)-4srg(r)\mu\mathfrak{R}(\overline{e_4\psi}T\psi)\dee\nu\right|\lesssim(-\mu)(r^2+a^2)^p\intS\ener_{deg}[\psi]\dee\nu,$$
where we used \eqref{eqn:expreT} to write
$$T=O(\mu)e_3+O(1)e_4+O(1)\mathcal{U}+O(1).$$
We continue with 
$$|2r\mu g(r)\mathfrak{R}(\overline{e_4\psi}\partial_u\psi)|\lesssim(-\mu)(r^2+a^2)^p\ener_{deg}[\psi],$$
$$\left|\intS\dfrac{2g(r)ar\mu}{r^2+a^2}\mathfrak{R}(\overline{e_4\psi}\Phi\psi)\dee\nu\right|\lesssim(-\mu)(r^2+a^2)^p\intS\ener_{deg}[\psi]\dee\nu,$$
where we used \eqref{eqn:exprephi} to get $\Phi=O(\mu)e_3+O(1)e_4+O(1)\mathcal{U}+O(1)$. Next, 
$$\left|\intS-g(r)\left(\mu s+s^2\mu\dfrac{a^4\sin^2\theta\cos^2\theta}{(r^2+a^2)^2}-\mu V\right)\mathfrak{R}(\overline{e_4\psi}\psi)\dee\nu\right|\lesssim(-\mu)(r^2+a^2)^p\intS\ener_{deg}[\psi]\dee\nu,$$
$$\left|2sg(r)\mu\dfrac{a^2\sin\theta\cos\theta}{(r^2+a^2)}\mathfrak{I}(\overline{e_4\psi}\mathcal{U}\psi)\right|\lesssim(-\mu)(r^2+a^2)^p\ener_{deg}[\psi],$$
$$\left|\intS\mu f(r)\frac{sra^2\mu\cos\theta\sin\theta}{(r^2+a^2)^2}\mathfrak{I}(\overline{\psi} \mathcal{U}\psi)\dee\nu\right|\lesssim(-\mu)(r^2+a^2)^p\intS\ener_{deg}[\psi]\dee\nu,$$
\begin{align*}
    \Bigg|\intS4\mu f(r)&as\cos\theta\mathfrak{I}(\overline{\partial_u\psi}T\psi)-4srf(r)\mu\mathfrak{R}(\overline{\partial_u\psi}T\psi)\dee\nu\Bigg|\\
    &\lesssim(-\mu)\varepsilon^{-1}(r^2+a^2)^p\intS|T\psi|^2\dee\nu+(-\mu)(r^2+a^2)^p\varepsilon\intS|e_3\psi|^2\dee\nu\\
    &\lesssim(-\mu)\varepsilon^{-1}(r^2+a^2)^p\intS\ener_{deg}[\psi]\dee\nu+(-\mu)(r^2+a^2)^p\varepsilon\intS|e_3\psi|^2\dee\nu,
\end{align*}
$$|2r\mu f(r)\mathfrak{R}(\overline{\partial_u\psi}e_4\psi)|\lesssim(-\mu)\varepsilon^{-1}(r^2+a^2)^p\ener_{deg}[\psi]+(-\mu)\varepsilon (r^2+a^2)^p|e_3\psi|^2,$$
\begin{align*}
    \Bigg|\intS-\dfrac{2f(r)ar\mu}{r^2+a^2}&\mathfrak{R}(\overline{\partial_u\psi}\Phi\psi)\dee\nu\Bigg|\lesssim\\
    &(-\mu)\varepsilon^{-1}(r^2+a^2)^p\intS\ener_{deg}[\psi]\dee\nu+(-\mu)(r^2+a^2)^p\varepsilon\intS|e_3\psi|^2\dee\nu,
\end{align*}
\begin{align*}
    \Bigg|\intS-f(r)\Bigg(\mu s+&s^2\mu\frac{a^4\sin^2\theta\cos^2\theta}{(r^2+a^2)^2}-\mu V\Bigg)\mathfrak{R}(\overline{\partial_u\psi}\psi)\dee\nu\Bigg|\lesssim\\
    &(-\mu)\varepsilon^{-1}(r^2+a^2)^p\intS\ener_{deg}[\psi]\dee\nu+(-\mu)(r^2+a^2)^p\varepsilon\intS|e_3\psi|^2\dee\nu,
\end{align*}

$$\left|2sf(r)\mu\dfrac{a^2\sin\theta\cos\theta}{(r^2+a^2)}\mathfrak{I}(\overline{\partial_u\psi}\mathcal{U}\psi)\right|\lesssim(-\mu)\varepsilon^{-1}(r^2+a^2)^p\ener_{deg}[\psi]+(-\mu)\varepsilon(r^2+a^2)^p|e_3\psi|^2,$$
$$|-4cs(r-M)g(r)\mathfrak{R}(\overline{e_4\psi}\partial_u\psi)|\lesssim (-\mu)(r^2+a^2)^p\varepsilon^{-1}\ener_{deg}[\psi]+(-\mu)(r^2+a^2)^p\varepsilon|e_3\psi|^2.$$
Notice that thanks to the $\mu$ in front of $|e_3\psi|^2$ in the definition \eqref{eqn:defenerdeg} of $\ener_{deg}[\psi]$, the integral on ${S(u,\ubar)}$ of $\ener_{deg}[\psi]$ can be absorbed in the one of \eqref{eqn:bpr}, for $p$ large enough. Moreover, we chose the value of $\varepsilon$ in the weighted Cauchy-Schwarz inequalities above so that all the terms bounded by a constant times $(-\mu)(r^2+a^2)^p\varepsilon|e_3\psi|^2$ can be absorbed in the term $(1-\mu p)|e_3\psi|^2\geq|e_3\psi|^2$ of \eqref{eqn:bpr}, for $\varepsilon>0$ small enough, after integration on the sphere. The only remaining term in the bulk that we want to absorb is
$$-4cs(r-M)f(r)|\partial_u\psi|^2=-4cs(r-M)(-\mu)(r^2+a^2)^p|e_3\psi|^2.$$ 
As it lacks a factor $\mu$, we cannot bound it in a pointwise manner (nor its integral on ${S(u,\ubar)}$) by $\ener_{deg}[\psi]$. But as $r_+-M>0$, we have $r-M\gtrsim 1$ for $r$ close to $r_+$ in $\un$, say for $r\in[r_+-\varepsilon_0,r_+]$. As we chose $s=-2<0$ and $c>0$, we have
$$-4cs(r-M)f(r)|\partial_u\psi|^2\geq 0,\quad r\in[r_+-\varepsilon_0,r_+].$$
And for $r\in [\rb,r_+-\varepsilon_0]$, we can absorb $-4cs(r-M)f(r)|\partial_u\psi|^2$ in the term 
$$-\mu p(r^2+a^2)^p|e_3\psi|^2\gtrsim -\mu(r_+-\varepsilon_0)p(r^2+a^2)^p|e_3\psi|^2$$
that appears in $\mathbf{B}_{pr}[\psi]$, for $p$ large enough. This concludes the proof of Lemma \ref{lem:bulkpos}.
\end{proof}

\subsection{Lower bound for the bulk term in $\{r_-<r\leq\rb\}$ for $s=+2$}\label{appendix:bulkII}
\begin{lem}\label{lem:posbulkII}
    For $c>1/4$, $s=+2$, $p=p(a,M,c,V)\gg 1$ large enough, and $\rb=\rb(a,M,c,V)$ sufficiently close to $r_-$, we have in $\{r_-\leq r\leq\rb\}$,
    \begin{align}\label{eqn:bulkIIpos}
    \intS\widehat{\mathbf{B}}[\psi]\dee\nu\gtrsim(-\mu)\intS{\ener}[\psi]\dee\nu.
\end{align}
\end{lem}
\begin{proof}
We have 
\begin{align*}
    r \mu g(r)-\partial_u\left(\left(r^2+a^2\right) g(r)\right)  &=-\mu p r\left(r^2+a^2\right)^p, \\
r \mu f(r)-e_4\left(\left(r^2+a^2\right) f(r)\right)  &=(-\mu)^{-1}[(r-M)-\mu(p+1) r]\left(r^2+a^2\right)^p.
\end{align*}
We also have
$$\partial_u(\mu f(r))+e_4(\mu g(r))  =-\mu\left(r^2+a^2\right)^{p-1}\left(\dfrac{2 M p r^2}{r^2+a^2}+r \mu-(r-M)\right)\gtrsim -\mu p(r^2+a^2)^p.
$$ 
Define the principal bulk
\begin{align*}
    \widehat{\mathbf{B}}_{pr}[\psi]:=2(r\mu g(&r)-\partial_u((r^2+a^2)g(r)))|e_4\psi|^2+2(r\mu f(r)-e_4((r^2+a^2)f(r)))|\ethat\psi|^2\\
    &\quad +\dfrac{1}{2}(\partial_u(\mu f(r))+e_4(\mu g(r)))(|\partial_\theta\psi|^2+|\mathcal{U}\psi|^2)-4cs(r-M)f(r)|\ethat\psi|^2.
\end{align*}
Note that unlike in the redshift region, we add a term $-4cs(r-M)f(r)|\ethat\psi|^2$ in the principal bulk. The positive spin will help us get a positive simple bulk term, without the need of replacing $f$ and $g$ by more complicated $\log$ multipliers, as is needed for the scalar wave equation in \cite[p. 22]{scalarMZ}. We have
\begin{align*}
    \widehat{\mathbf{B}}_{pr}[\psi]\gtrsim (-\mu)p(r^2+a^2)^p&|e_4\psi|^2+2(-\mu)^{-1}[(r-M)(1-2cs)-\mu(p+1) r]\left(r^2+a^2\right)^p|\ethat\psi|^2\\
    &\quad +(-\mu)p(r^2+a^2)^p(|\partial_\theta\psi|^2+|\mathcal{U}\psi|^2).
\end{align*}
Notice that in $\deux$, we have $r-M\lesssim -1$ thus for $s=+2$ and $c>1/4$,
\begin{align}\label{eqn:use1/4}
    (r-M)(1-2cs)-\mu(p+1) r=(r-M)(1-4c)-\mu(p+1)\gtrsim 1-\mu p,\quad r\in(r_-,\rb].
\end{align}
This is where we use the positivity of the spin, to get an effective blueshift effect. We have shown 
\begin{align}\label{eqn:bprhat}
    \widehat{\mathbf{B}}_{pr}[\psi]\gtrsim (-\mu)(r^2+a^2)^p[p|e_4\psi|^2+(1-\mu p)|e_3\psi|^2+p(|\partial_\theta\psi|^2+|\mathcal{U}\psi|^2)].
\end{align}
The only thing left to prove is that we can take $p$ large enough so that $\widehat{\mathbf{B}}[\psi]-\widehat{\mathbf{B}}_{pr}[\psi]$ can be absorbed in $\widehat{\mathbf{B}}_{pr}[\psi]$ after integrating on ${S(u,\ubar)}$. This is due to the following mix between weighted Cauchy-Schwarz of the type 
$$|ab|\leq\frac{\varepsilon a^2+\varepsilon^{-1}b^2}{2}$$
and the Poincaré inequality \eqref{eqn:poincare} :
\begin{align*}
    \Bigg|\intS \mu g(r)\frac{sra^2\mu\cos\theta\sin\theta}{(r^2+a^2)^2}\mathfrak{I}(\overline{\psi} \mathcal{U}\psi)\dee\nu\Bigg|&\lesssim (-\mu)(r^2+a^2)^p\left(\intS\ener_{deg}[\psi]\dee\nu+\intS|\mathcal{U}\psi|^2\dee\nu\right)\\
    &\lesssim(-\mu)(r^2+a^2)^p\intS\ener_{deg}[\psi]\dee\nu,
\end{align*}
$$\Bigg|\intS4\mu g(r)as\cos\theta\mathfrak{I}(\overline{e_4\psi}T\psi)-4srg(r)\mu\mathfrak{R}(\overline{e_4\psi}T\psi)\dee\nu\Bigg|\lesssim(-\mu)(r^2+a^2)^p\intS\ener_{deg}[\psi]\dee\nu,$$
where we used again \eqref{eqn:expreT} to get
$$T=O(\mu)e_3+O(1)e_4+O(1)\mathcal{U}+O(1).$$
We continue with 
$$|2r\mu g(r)\mathfrak{R}(\overline{e_4\psi}\ethat\psi)|\lesssim(-\mu)(r^2+a^2)^p\ener_{deg}[\psi],$$
$$|-4sc(r-M)g(r)\Real(\overline{\partial_\ubar\psi}\ethat\psi)|\lesssim (-\mu)(r^2+a^2)^p(\varepsilon|e_3\psi|^2+\varepsilon^{-1}\ener_{deg}[\psi]),$$
$$\Bigg|\intS\dfrac{2g(r)ar\mu}{r^2+a^2}\mathfrak{R}(\overline{e_4\psi}\Phi\psi)\dee\nu\Bigg|\lesssim(-\mu)(r^2+a^2)^p\intS\ener_{deg}[\psi]\dee\nu,$$
where we used \eqref{eqn:exprephi} to get $\Phi=O(\mu)e_3+O(1)e_4+O(1)\mathcal{U}+O(1)$. Next, 
$$\left|\intS-g(r)\left(\mu s+s^2\mu\dfrac{a^4\sin^2\theta\cos^2\theta}{(r^2+a^2)^2}-\mu V\right)\mathfrak{R}(\overline{e_4\psi}\psi)\dee\nu\right|\lesssim(-\mu)(r^2+a^2)^p\intS\ener_{deg}[\psi]\dee\nu,$$
$$\Bigg|2sg(r)\mu\dfrac{a^2\sin\theta\cos\theta}{(r^2+a^2)}\mathfrak{I}(\overline{e_4\psi}\mathcal{U}\psi)\Bigg|\lesssim(-\mu)(r^2+a^2)^p\ener_{deg}[\psi],$$
$$\Bigg|\intS\mu f(r)\frac{sra^2\mu\cos\theta\sin\theta}{(r^2+a^2)^2}\mathfrak{I}(\overline{\psi} \mathcal{U}\psi)\dee\nu\Bigg|\lesssim(-\mu)(r^2+a^2)^p\intS\ener_{deg}[\psi]\dee\nu,$$
\begin{align*}
    \Bigg|\intS4\mu f(r)&as\cos\theta\mathfrak{I}(\overline{\ethat\psi}T\psi)-4srf(r)\mu\mathfrak{R}(\overline{\ethat\psi}T\psi)\dee\nu\Bigg|\\
    &\lesssim(-\mu)\varepsilon^{-1}(r^2+a^2)^p\intS|T\psi|^2\dee\nu+(-\mu)(r^2+a^2)^p\varepsilon\intS|e_3\psi|^2\dee\nu\\
    &\lesssim(-\mu)\varepsilon^{-1}(r^2+a^2)^p\intS\ener_{deg}[\psi]\dee\nu+(-\mu)(r^2+a^2)^p\varepsilon\intS|e_3\psi|^2\dee\nu,
\end{align*}
$$|2r\mu f(r)\mathfrak{R}(\overline{\ethat\psi}e_4\psi)|\lesssim(-\mu)\varepsilon^{-1}(r^2+a^2)^p\ener_{deg}[\psi]+(-\mu)\varepsilon (r^2+a^2)^p|e_3\psi|^2,$$
\begin{align*}
    \Bigg|\intS-\dfrac{2f(r)ar\mu}{r^2+a^2}&\mathfrak{R}(\overline{\ethat\psi}\Phi\psi)\dee\nu\Bigg|\\
    &\lesssim(-\mu)\varepsilon^{-1}(r^2+a^2)^p\intS\ener_{deg}[\psi]\dee\nu+(-\mu)(r^2+a^2)^p\varepsilon\intS|e_3\psi|^2\dee\nu,
\end{align*}
\begin{align*}
    \Bigg|\intS-f(r)\Bigg(\mu s+&s^2\mu\dfrac{a^4\sin^2\theta\cos^2\theta}{(r^2+a^2)^2}-\mu V\Bigg)\mathfrak{R}(\overline{\ethat\psi}\psi)\dee\nu\Bigg|\\
    &\lesssim(-\mu)\varepsilon^{-1}(r^2+a^2)^p\intS\ener_{deg}[\psi]\dee\nu+(-\mu)(r^2+a^2)^p\varepsilon\intS|e_3\psi|^2\dee\nu,
\end{align*}

$$\Bigg|2sf(r)\mu\dfrac{a^2\sin\theta\cos\theta}{(r^2+a^2)}\mathfrak{I}(\overline{\ethat\psi}\mathcal{U}\psi)\Bigg|\lesssim(-\mu)\varepsilon^{-1}(r^2+a^2)^p\ener_{deg}[\psi]+(-\mu)\varepsilon(r^2+a^2)^p|e_3\psi|^2.$$
Notice that thanks to the $\mu$ in front of $|e_3\psi|^2$ in the definition of $\ener_{deg}[\psi]$, the integral on ${S(u,\ubar)}$ of $\ener_{deg}[\psi]$ can be absorbed in the one of \eqref{eqn:bprhat}, for $p$ large enough. Moreover, we choose the value of $\varepsilon$ in the weighted Cauchy-Schwarz inequalities above so that all the terms bounded by a constant times (the integral of) $(-\mu)(r^2+a^2)^p\varepsilon|e_3\psi|^2$ can be absorbed in the term $(1-\mu p)|e_3\psi|^2\geq|e_3\psi|^2$ in \eqref{eqn:bprhat}, for $\varepsilon>0$ small enough. This concludes the proof of Lemma \ref{lem:posbulkII}.

\end{proof}

\section{Computation of $A_m(r)$ and proof of Lemma \ref{prop:teukansatz}}\label{appendix:am(r)}
The polynomial $A_m(r)=(r^2+a^2)^2\mathfrak{f}_{m,2}(r)$ is defined in \cite[Eq. (5.82c)]{pricelaw} by plugging the ansatz 
$$\ansatzm$$
for $\psim$ into the TSI \eqref{eqn:tsi1} and requiring the compatibility
\begin{align}\label{eqn:holds}
    \frac{1}{24}\Delta^2\partial_{r_\mathrm{out}}^4\left(\Delta^2\ansatzm\right)=\frac{1}{\ubar^7}\sum_{|m|\leq 2}A_m(r)Q_{m,2}Y_{m,2}^{-2}(\cos\theta)e^{im\phi_{+}}+O(\ubar^{-8}),
\end{align}
see also \cite[Eq. (5.88), (5.89)]{pricelaw}, where the factor $1/24$ corresponds to $1/(2\mathfrak{s})!$ with $\mathfrak{s}=2$. We recall
$$\partial_{r_\mathrm{out}}=2\mu^{-1}e_4=\partial_r+\mu^{-1}\partial_t+\frac{a}{\Delta}\Phi.$$
More precisally, the computation done in \cite[p. 68,  eq. (5.88)]{pricelaw} gives
\begin{align}\label{eqn:am(r)0}
    A_m(r)=\frac{1}{24}e^{-im\phi_{+}}\Delta^2\partial_{r_\mathrm{out}}^4(\Delta^2e^{im\phi_{+}}),
\end{align}
and then \eqref{eqn:holds} holds, where the $O(\ubar^{-8})$ term is given by the terms where a $\partial_{r_\mathrm{out}}$ falls on an inverse power of $\ubar$. Let us now compute \eqref{eqn:am(r)0}. We have $\partial_{r_\mathrm{out}}(e^{im\phi_{+}})=\frac{2aim}{\Delta}e^{im\phi_{+}}$ thus 
$$\partial_{r_\mathrm{out}}(\Delta^2e^{im\phi_{+}})=(\partial_r\Delta^2+2aim\Delta)\exptruc=2\Delta(2(r-M)+aim)\exptruc.$$
We compute successively
\begin{align*}
    \partial_{r_\mathrm{out}}^2(\Delta^2e^{im\phi_{+}})&=\left(\partial_r+\frac{2iam}{\Delta}\right)(2\Delta(2(r-M)+aim))\exptruc\\
    &=4[2(r-M)^2+3(r-M)aim+\Delta-a^2m^2]\exptruc,
\end{align*}
\begin{align*}
    \partial_{r_\mathrm{out}}^3(\Delta^2e^{im\phi_{+}})&=\left(\partial_r+\frac{2iam}{\Delta}\right)(4[2(r-M)^2+3(r-M)aim+\Delta-a^2m^2])\exptruc\\
    &=4\left[\frac{4(r-M)^2aim}{\Delta}+6(r-M)-\frac{6(r-M)a^2m^2}{\Delta}-\frac{2ia^3m^3}{\Delta}+5aim\right]\exptruc,
\end{align*}
\begin{align*}
    \partial_{r_\mathrm{out}}^4&(\Delta^2e^{im\phi_{+}})\\
    &=\left(\partial_r+\frac{2iam}{\Delta}\right)\left(4\left[\frac{4(r-M)^2aim}{\Delta}+6(r-M)-\frac{6(r-M)a^2m^2}{\Delta}-\frac{2ia^3m^3}{\Delta}+5aim\right]\right)\exptruc\\
    &=4\Bigg[\frac{8(r-M)(a^2-M^2)aim}{\Delta}+6+\frac{6a^2m^2}{\Delta}-\frac{12a^2m^2(a^2-M^2)}{\Delta^2}+\frac{4ia^3m^3(r-M)}{\Delta^2}\\
    &\quad\quad+\frac{2aim}{\Delta}\left(\frac{4(r-M)^2aim}{\Delta}+6(r-M)-\frac{6(r-M)a^2m^2}{\Delta}-\frac{2ia^3m^3}{\Delta}+5aim\right)\Bigg]\exptruc.
\end{align*}
Thus we get 
\begin{align*}
    \Delta^2\partial_{r_\mathrm{out}}^4(\Delta^2e^{im\phi_{+}})=8\Bigg[3\Delta^2&+iam(4(a^2-M^2)(r-M)+6\Delta(r-M))\\
    &+a^2m^2(3\Delta-6(a^2-M^2)-4(r-M)^2-5\Delta)\\
    &+ia^3m^3(2(r-M)-6(r-M))+2a^4m^4\Bigg]\exptruc,
\end{align*}
which finally gives
\begin{align}
    \label{eqn:am(r)}A_m(r)=\frac{1}{3}\Big[3\Delta^2&+(r-M)(4(a^2-M^2)+6\Delta)iam-(2\Delta+6(a^2-M^2)+4(r-M)^2)a^2m^2\nonumber\\
    &-4(r-M)ia^3m^3+2a^4m^4\Big].
\end{align}
\begin{proof}[Proof of Lemma \ref{prop:teukansatz}]
We use Proposition \ref{prop:teukavece3} to get 
\begin{align*}
    \teuk_{+2}\Bigg(&\ansatzp\Bigg)=\\
    &\frac{1}{\ubar^7}\sum_{|m|\leq 2}\Big[\Delta A_m''(r)+2(iam-(r-M))A_m'(r)-4 A_m(r)\Big] Q_{m,\fraks}Y_{m,2}^{+2}(\cos\theta)e^{im\phi_{+}}+\mathrm{Err},
\end{align*}
where we used \eqref{eqn:annulemode} to get $\drond\drond'(Y_{m,2}^{+2}(\cos\theta)e^{im\phi_{+}})=-4Y_{m,2}^{+2}(\cos\theta)e^{im\phi_{+}}$, and where the error term $\mathrm{Err}$ is explicit and defined by
$$\mathrm{Err}:=\Big(a^2\sin^2\theta T^2-4(r^2+a^2)Te_3+2aT\Phi-(6r+4ia\cos\theta)T\Big)\left[\ansatzp\right].$$
It satisfies $e_3^{\leq 1}T^j\carterp^k\mathrm{Err}=O(\ubar^{-8-j})$. It remains only to prove
$$\Delta A_m''(r)+2(iam-(r-M))A_m'(r)-4 A_m(r)=0.$$
We have  
\begin{align*}
    &3A_m'(r)=12(r-M)\Delta+iam(4(a^2-M^2)+6\Delta+12(r-M)^2)-12(r-M)a^2m^2-4ia^3m^3,\\
    &3A_m''(r)=12\Delta+24(r-M)^2+36(r-M)iam-12a^2m^2.
\end{align*}
Thus we can compute 
$$\Delta A_m''(r)+2(iam-(r-M))A_m'(r)-4 A_m(r)=\frac{1}{3}[c_0+c_1iam+c_2a^2m^2+c_3ia^3m^3+c_4a^4m^4],$$
where the coefficients are :
\begin{align*}
    &c_0=12\Delta^2+24(r-M)^2\Delta-24(r-M)^2\Delta-12\Delta^2=0,\\
    &c_1=36\Delta(r-M)+24(r-M)\Delta-2(r-M)(4(a^2-M^2)+6\Delta+12(r-M)^2)\\
    &\quad\quad-4(r-M)(4(a^2-M^2)+6\Delta)\\
    &\quad=24(r-M)\Delta-24(r-M)(a^2-M^2)-24(r-M)^3=24(r-M)(\Delta-\Delta)=0\\
    &c_2=-12\Delta-8(a^2-M^2)-12\Delta-24(r-M)^2+24(r-M)^2+8\Delta+24(a^2-M^2)+16(r-M)^2\\
    &\quad=-16\Delta+16(a^2-M^2)+16(r-M)^2=16(\Delta-\Delta)=0\\
    & c_3=-24(r-M)+8(r-M)+16(r-M)=0\\
    & c_4=8-8=0,
\end{align*}
which concludes the proof of Lemma \ref{prop:teukansatz}.
\end{proof}

\section*{Declarations}
\noindent\textbf{Conflicts of interest.} The author has no relevant financial or non-financial interests to disclose.

\noindent\textbf{Data availability statement.} Data sharing not applicable to this article as no datasets were generated or
analysed during the current study.

\printbibliography

\end{document}